\documentclass[12pt,english,superscriptaddress,nofootinbib,preprintnumbers, aps,longbibliography]{article}
\usepackage{lmodern}
\usepackage[english]{babel}
\usepackage[T1]{fontenc}
\usepackage[latin9]{inputenc}
\usepackage{geometry}
\usepackage{amsmath}
\usepackage{amsthm}
\usepackage{amssymb}
\usepackage[normalem]{ulem}

\usepackage{mathtools}

\newtheorem{fact}{Fact}
\newtheorem{lm}{Lemma}
\newtheorem{tm}{Theorem}
\theoremstyle{definition}
\newtheorem{df}{Definition}

\theoremstyle{corollary}
\newtheorem{cor}{Corollary}

\theoremstyle{remark}
\newtheorem{remark}{Remark}[section]
\theoremstyle{remark}
\newtheorem{example}[remark]{Example}

\newcommand{\D}{{\mathbb{D}}}
\newcommand{\G}{{\mathbb{G}}}
\DeclareMathOperator{\Diff}{Diff}

\newcommand{\M}{\mathcal{M}}
\newcommand{\R}{\mathbb{R}}
\newcommand{\Z}{\mathbb{Z}}
\newcommand{\graph}{\gamma}
\DeclareMathOperator{\Span}{span}

\usepackage{graphicx}
\usepackage{color}
\usepackage{xcolor}
\usepackage{authblk}
\usepackage{etoolbox}
\usepackage{subfigure}
\usepackage{helvet}
\usepackage{float}
\usepackage[makeroom]{cancel}

\usepackage{cite}
\usepackage{hyperref}
\hypersetup{
    colorlinks,%
    citecolor=blue,%
    filecolor=blue,%
    linkcolor=blue,%
    urlcolor=blue
}

\title{Symmetry restriction and its application to gravity}
\author[1]{Wojciech Kami{\'n}ski\thanks{wojciech.kaminski@fuw.edu.pl}}
\author[2,3]{Klaus Liegener\thanks{klaus.liegener@desy.de}}
\affil[1]{Institute of Theoretical Physics, Faculty of Physics,
  University of Warsaw,
  Pasteura 5, 02-093 Warszawa, Poland}
\affil[2]{Louisiana State University, Department of Physics and Astronomy,
202 Nicholson Hall, Baton Rouge, LA 70803, USA}
\affil[3]{II. Institute for Theoretical Physics, University of Hamburg, Luruper Chaussee 149, 22761 Hamburg, Germany}

\date{\today}                     %% if you don't need date to appear

\begin{document}

\maketitle

\abstract{
In the Hamiltonian formulation, it is not a priori clear whether a symmetric configuration will keep its symmetry during evolution. In this paper, we give precise requirements of when this is the case and propose a {\it symmetry restriction} to the phase space of the symmetric variables. This can often ease computation, especially when transcending from the infinite dimensional phase space of a field theory to a possibly finite dimensional subspace. We will demonstrate this in the case of gravity. A prominent example is the restriction of full Hamiltonian general relativity to the cosmological configurations of Robertson-Walker type. We will demonstrate our procedure in this setting and extend it to examples which appear useful in certain approaches to quantum gravity.}
\section{Introduction}\label{sec:notation}
Symmetry groups play a pivotal role in modern physics.  Particular solutions for involved physical theories are found by not looking for the most general solutions, but those which obey certain symmetries, i.e. are invariant under the real space action of some symmetry group $\Phi$. These solution are often those of the most importance and prominent examples are crystal symmetries in solid state physics \cite{Hahn:2002}, spherical symmetric black holes solutions in General Relativity (GR) \cite{Schwarzschild:1916} or its cosmological solutions for homogeneous universes such as the Bianchi classificaction \cite{Bianchi:1898}.
\\
Under symmetry assumptions, the complicated equations of motion simplify often even to ordinary differential equations that can be solve explicitly. Although, it can be done case by case it is useful to know whether there is a rule governing this process.
We are considering in this article the Hamiltonian formulation of dynamics on a phase space of the theory. After imposing symmetry under some group $\Phi$ of canonical transformations one obtains a subspace of said phase space. In many situation this symmetry restricted sector is also a phase space equipped in a natural way with  the {\it restricted} Poisson bracket of the symmetry-invariant variables. We can also define a restricted Hamiltonian, by simply considering the original Hamiltonian on the restricted phase space. The question of interested is now formulated in an obvious way: do the restricted dynamics agree with the full dynamics on the restricted phase space? Under natural conditions, the answer is affirmative as we show in section \ref{section:2_Symmetry_Red}. In case one established this a priori, one is in an advantageous position: instead of having to deal with the possibly infinitely many degrees of freedom (for a field theory) one may reduce to a finite, small number which are better handleable in terms of analytic and numerical investigations. Maybe the best example for this comes in form of gravity and the Robertson-Walker metric describing a spatially isotropic metric-field of 10 degrees of freedom per point that is completely determined by a single degree of freedom: the scale factor \cite{Fri:22,Lem:33,Rob:35,Wal:37}.
After having established the general requirements for such a {\it symmetry restriction}, we will also explicitly check its applicability for GR and its cosmological solutions.\\

In case of a field theory there may appear additional difficulty in form of gauge degrees of freedom. For applications in GR this difficulty is especially prominent as the gauge group consists of all diffeomorphisms of the spacetime. In the Hamiltonian framework the group $G$ of gauge transformations come in pairs with constraints. In order to describe only real physical degrees of freedom one should perform {\it symplectic reduction} with respect to the action of gauge groups as introduced by Mardsen, Weinstein and Meyer \cite{Mey:73,MW:74,MMOPR:07,Sni:13}. Despite the similarity in name to {\it symmetry restriction}, both processes are very different conceptually: the gauge group $G$ identifies via their action points in phase space that are physically equivalent and symplectic reduction mods out $G$ in order to obtain a reduced phase space of gauge-invariant variables. A further question can be formulated: Can we postpone the symplectic reduction to the symmetry restricted phase space where it is easier to handle? Under certain assumption the answer is again affirmative and we describe this process in section \ref{section:2_relation}.\\

The motivation for such general questions are usually very particular applications.
An underlying motivation for this article comes in the form of interest in the semi-classical sector of quantum gravity (QG). A theory of QG has not been constructed as of today to full satisfaction, albeit numerous proposal exists. One possible example comes in the form of Loop Quantum Gravity (LQG) \cite{Thi:07}, an application of Dirac's quantisation procedure of gauge theories \cite{Dirac:30,Dirac:48} to GR's Hamiltonian formulation by Arnowitt, Deser and Misner (ADM) \cite{ADM:62}. The latter one can be reformulated as a $O(3)$-gauge theory in terms of the Ashtekar-Barbero (AB) variables \cite{Ash:86,Bar:94}. Our main result concerns a discretised version of this theory defined on some graph. Its importance steam from the fact that it is expected to be semiclassical limit of a discretised version of Quantum Gravity similar to Algebraic LQG \cite{AQG1}. However, on the graph usually one considers a slightly different group of gauge transformation parametrised by $SU(2)$. Not all $SO(3)$ gauge transformations can be lifted to $SU(2)$ if the manifold has a nontrivial topology. In our paper we point out that these topologically nontrivial gauge transformations can still be implemented on the graph leading to some additional conditions on discretised constraints. In recent years much work went into the question, whether LQG features an ``isotropic'' sector which can be related to a canonical quantisation solely of the scale factor \cite{Engle:07,BF:07,Paw:14}. In answering how symmetry restriction applies to gravity (or furthermore a discretised version  thereof), this article provides the first step of understanding the semi-classical evolution of a symmetry-invariant universe, by reaffirming some of the effective models of Loop Quantum Cosmology: we find that on the reduced phase space of scale factor and its momentum, the flow of the constraints from \cite{DMY:08,DL:17a,DL:17b} indeed agrees with the flow on the phase space of full discretised GR equipped with Thiemann's discretisation of the constraints \cite{Thi98(QSD),Thi98(QSDII),AQG1}. In the conclusion, we outline how the missing piece, i.e. transitioning from the quantum to the effective, classical description, may be made precise.\\

The organisation of this article is as follows:\\
In section \ref{section:2} we study symmetry restriction and symplectic reduction in general for symplectic manifolds $\M$, i.e. a phase space endowed with a symplectic 2-form. For the symplectic reduction, we repeat the standard construction from the literature, to clarify differences and introduce a common notation. Concerning symmetry restriction, we show for a given symmetry group $\Phi$ that the flow of $\Phi$-invariant functions does not leave the $\Phi$-invariant subspace of phase space and agrees with the flow of the reduced function with respect to the reduced symplectic structure, given the latter one is well-defined. This main result of the paper allows, e.g., to study dynamics in a simplified setting without loosing information. We also extend the framework to constrained systems, for which GR presents an example by combining symmetry restriction and symplectic reduction.\\
In section \ref{section:3} we introduce two Hamiltonian formulations of gravity on a spatial manifold $\sigma$: the ADM formulation and one in terms of AB variables. In both cases the complete group of spatial diffeomorphisms, $\Diff(\sigma)$, becomes part of the gauge group. In ADM, we clarify the structure of the gauge transformations on shell  formally (thereby extending known results by Bergman-Komar): every time-preserving diffeomorphism is implemented by some spatial diffeomorphism and a transformation generated by the constraints. We recall how the group $\Diff(\sigma)$ can be translated into an action of symplectomorphisms on the phase space and how a subgroup $\Psi\subset \Diff(\sigma)$ can play the double role of gauge- and symmetry-group, therefore enabling the process of symmetry restriction for GR as well. The AB variables allow a formulation of gravity as a Yang-Mills theory, where to gauge group of diffeomorphisms, one adds moreover the gauge group $O(3)$. We pay special attention to the distinction of gauge groups $O(3)$ and $SO(3)$ (as both description are possible and differ by a slightly different choice of phase space). Also due to the additional gauge group, several versions of a representation of $\Diff(\sigma)$ in terms of symplectomorphisms on the AB phase space become possible and we introduce a preferred choice of representation, once a certain symmetry group $\Psi$ has been chosen for symmetry restriction.\\
The second part of this article applies symmetry restriction to a number of examples: in section \ref{section:4} we focus on cosmological scenarios for classical, continuous GR, such as the isotropic, flat model by Friedmann-Lema\^itre-Robertson-Walker (FLRW). Here, the symmetry group of all translations and rotations can be used for symmetry restriction. We make a distinction between compact and non-compact spatial manifolds as it highlights the following  the following fact: while constraints driving the dynamics of GR are invariant under the FLRW-symmetry group, in the non-compact case reduction of the symplectic manifold to the invariant phase space is ill-defined, thus one cannot compute the dynamics solely on the reduced level. This caveat however, is absent in the non-compact case and therefore symmetry restriction finds full application for, e.g. Bianchi I spacetimes and FLRW with torus-like spatial sections and the positive-closed isotropic universe.\\
In section \ref{section:5}, we turn towards an application that becomes interesting from point of view of quantum gravity: we study classical discretisations of GR in AB variables on compact spatial manifolds, i.e. ``gravity on a graph''. This is a finite dimensional system, which approximates the original field theory.\footnote{As AB variables describe a classical Yang-Mills theory, it makes sense to build the discretisation as closely as possible to the framework of Lattice Gauge theories and thus employs holonomies and fluxes.} The restriction to a graph forces us introduce an approximation of most constraints, as only the action of the $SO(3)$ part of the $O(3)$-gauge group can be lifted exactly in form of $SU(2)$-valued fields. However, when doing so we will find that the lift of the remaining global $O(3)$-transformation surprisingly carries information of the topology of the original continuous manifold when lifted to the graph. In order to apply symmetry restriction, we will see that it makes sense to chose a subgroup of diffeomorphisms that preserve the chosen graph, and it is noteworthy that it becomes a non-trivial problem to find a duple -- graph and symmetry group -- to serve as the approximation of some continuous symmetry-invariant system. An example where it can be achieved presents a discretisation Bianchi I spacetimes, which we carry out explicitly: after deriving the reduction of the symplectic form we show that the approximated constraints are found to be invariant and thus symmetry restriction can be fully carried out, yielding a reduced Hamiltonian on a 2-dimensional phase space, whose trajectory agrees exactly with the trajectory of Bianchi I on the full discrete GR system. Finally, the same steps can be repeated for FLRW and yield the same constraint from \cite{DMY:08} that was already postulated in \cite{DL:17a} (and computed directly in \cite{HL:19}) to describe the effective dynamics of the graph and replaced the initial singularity by an asymmetric bounce, which we thus deduce to be caused solely by the discreteness effects.\\
In the conclusions \ref{section:6} we outline how this classical result -- that also FLRW on a graph can be described by only using the scale factor as degree of freedom albeit with a modified scalar constraint -- may be extended to the quantum regime: we give a justified conjecture that the expectation values of observables in coherent states sharply peaked on FLRW on a graph will follow above reduced trajectory.\\

Throughout this manuscript, we streamline the following notation :\\
\begingroup
\setlength{\tabcolsep}{10pt} % Default value: 6pt
\renewcommand{\arraystretch}{1.5} % Default value: 1
\begin{tabular}{ll}
$\,C^\infty (\sigma)$ & smooth functions from space $\sigma \mapsto \mathbb{R}$\\
$\,C^\infty (\sigma,V)$ & smooth functions from space $\sigma$ to some space $V$\\
$\,C^\infty_0 (\sigma)$ & compactly supported smooth functions $\sigma\mapsto\mathbb{R}$\\
$\,C^\infty(\sigma)_\Psi$ & smooth functions on $\sigma$, invariant with respect to diffeomorphisms in $\Psi$\\
$\prescript{m}{n}{C}^\infty (\sigma)$ & $m$-vector and $n$-covector fields on $\sigma$ (where we omit $m,n$ if being zero)\\
$\prescript{m}{n}{\widetilde{C}}^\infty (\sigma)$ & densitised $(m,n)$-tensor fields on $\sigma$\\
$\prescript{(m)}{[n]}{C}^\infty(\sigma)$ & symmetric $m$-vector and antisymmetric $n$-covector fields on $\sigma$\\
$\prescript{m}{n}{T}_p(\sigma)$ & $(m,n)$-tensors at point $p\in\sigma$\\
$\prescript{m}{}{T}_p (\sigma)_\Psi$ & $m$-vectors at $p\in\sigma$, invariant with respect to diffeomorphisms in $\Psi$
\end{tabular}
\endgroup

\newpage
\section{Reduction and Restriction Theorems for Symplectic Manifolds}
\label{section:2}
We give a geometric perspective on the symmetry restriction processes on a phase space (symplectic manifold).  We will consider a different scenario than known symplectic reduction. Our procedure will restrict allowed configurations to those that satisfy some prescribed symmetry given by $\Phi$ a subgroup of symplectic transformations.\footnote{Let us notice the specific choice of $\Phi$ thus in typical application in GR our symmetric configuration will be in a specific gauge.}\\
In subsection \ref{section:2_Setup} we recap basic notions of differential geometry. We introduce a group $\Phi$ that acts on a phase space in the form of symplectomorphims, i.e. diffeomorphims that respect the symplectic structure. A symplectic structure allows to obtain from any function $f$ a 1-parameter family of diffeomorphims on $\M$ (called its evolution) via the Hamiltonian vector fields. We demonstrate: if $f$ is invariant under the transformations induced by $\Phi$ then its evolution commutes with $\Phi$-transformations.\\
Then, we will restrict our attention in subsection \ref{section:2_Symmetry_Red} to $\overline{\M}$, the subset of $\M$ which is invariant under the action of $\Phi$ -- later on we will associate this set with phase space points which employ a certain symmetry corresponding to $\Phi$. If $\overline{\M}$ is a symplectic manifold, then every $\Phi$-invariant function $f$ induces an evolution that stays in $\overline{\M}$, indicating that the symmetry of a system cannot be lost under evolution induced by $f$. This is of interest in applications where $f$ is the Hamiltonian of a system, thus generating {\it time} evolution. Typically the evolution of a Hamiltonian is computed via Poisson brackets on $\M$ and we show that it is equivalent to compute the evolution of $f|_{\overline{\M}}$ (the Hamiltonian restricted to the invariant submanifold) with respect to the symplectic structure on $\overline{\M}$. This formulates our theorem for {\it symmetry restriction}.\\
In order to deal with gauge degrees of freedom we recall {\it symplectic reduction} of Mardsen-Weinstein-Meyer in subsection \ref{section:2_Symplectic_Red}. It presents a mathematical formulation to implement constraints of Hamiltonian systems. An algebra of constraints $\hat{J}$ gives rise to a (subgroup of a)  gauge group $G$ that can be represented as symplectomorphisms on the phase space $\M$. One is interested in the constraint surface (where vanishing of all constraints has been imposed) and in equivalence classes of $G$ thereon (points in an orbit of $G$ are understood to be physically indistinguishable). We recall the {\it symplectic reduction theorem} by which $\M/\!\!/ G$, the space of the latter equivalence classes, can be endowed with a symplectic structure. Moreover, the {\it symplectic reduction of dynamics theorem} states that any gauge-invariant Hamiltonian can be projected onto $\M/\!\!/ G$ and that the projection of its flow agrees with the flow of the projected Hamiltonian with respect to the reduced Poisson brackets.\\
As we are interested to apply {\it symmetry restriction} in settings where constraints are present we use subsection \ref{section:2_relation} to elucidate on the relation between both restriction and reduction procedures. The symplectic reduction is only applicable if the group $G$ acts freely, i.e. no points on the phase space are invariant under a subset $\Phi\subset G$. However, these points are exactly describing physical configurations with symmetry. To solve the theory at these points, symmetry restriction turns out to be the complement to symplectic reduction: we outline how symmetry restriction with $\Phi\subset G$ can be used to give rise to a new constraint algebra and a corresponding symmetry group $G^{red}$ on $\overline{\M}$ for which symplectic reduction becomes applicable again.\\

\subsection{Symplectic manifolds and group action} 
\label{sec:symplectic-def}

We  consider both a finite dimensional setup, which is main topic of this work, as well as certain infinite dimensional generalisations. This covers our interests in cosmology. The following sections will mostly be self-consistent, in presenting proofs for all relevant statements. However, for further references in the literature, we refer to standard text books \cite{KN:96,Nak:03}.\\

Let $(\mathcal M, \omega)$ be a finite dimensional symplectic manifold, that is $\omega\in {}_{[2]}C^\infty(\M)$ is a non-degenerate  antisymmetric 2-form on $\M$.
We remind that, given a symplectic structure $\omega$, there exists for any function $f\in C^\infty(\M)$ a {\it Hamiltonian vector field} $\chi_f\in {}^1C^\infty(\M)$ defined via the equation
\begin{align}\label{hamiltonian_vf}
\omega (\chi_f,\;\cdot\;)= df(\;\cdot\;)
\end{align}
We call $f$ the {\it Hamiltonian} of $\chi_f$. The Poisson bracket is defined by
\begin{equation}
\{f,g\}=\chi_f(g)
\end{equation}
Thus, $\{f,g\}=-\omega(\chi_f,\chi_g)$.

Let now $(\mathcal M, \omega)$ be a Hilbert manifold, possibly infinite dimensional  with a weakly symplectic two form $\omega\in {}_{[2]}C^\infty(\M)$ (see \cite{lang:85} for the introduction to Banach manifolds). The weakly symplectic structure means that for every $m\in \M$ and any vectors $X^1,X^2\in T_m\M$
\begin{equation}
\omega_m(X^1,\;\cdot\;)=\omega_m(X^2,\;\cdot\;)\Longrightarrow X^1=X^2
\end{equation}
However, as it is not strongly symplectic manifold, the Hamiltonian vector field might not exists. In fact this often happens in applications in field theory.

We  say that $m\in \M$ is an $f$-regular point if there exists a vector $X_m\in T_m\M$ such that
\begin{equation}\label{F-regular}
\omega_m(X_m,\;\cdot\;)=df|_m
\end{equation}
In the finite dimensional setup all points are regular. We define the domain of the Hamiltonian vector field as a set $U\subset \M$ of points for which such $X_m$ exists.\footnote{If it exists then it is unique due to weak non-degeneracy.} We define $\chi_f$ on $U$ by the formula (\ref{F-regular}), i.e. ${\chi_f}|_m:=X_m$.

Throughout the paper we will always assume that all functions and fields are smooth.

\subsubsection{Group actions on symplectic manifold}
\label{section:2_Setup}
Not always there exists a generator for some vector field $X$ preserving symplectic form, because $\omega (X,\;\cdot\;)$ might be not exact (it is only closed). However, in many instances, one can easily compute such a generator.

The presymplectic potential $\xi$ is a $1$-form satisfying
\begin{equation}
\omega=d\xi
\end{equation}
Such a form is not unique even if exists. The following is a very well known but important fact:

\begin{fact}\label{fact-presymplectic}
Suppose that the vector field $X$ preserve a presymplectic potential $\xi$
\begin{equation}
\mathcal{L}_X\xi=0
\end{equation}
then it preserves the symplectic form and $X=\chi_f$ where 
\begin{equation}
f=-\xi(X)
\end{equation}
\end{fact}

\begin{proof}
To show this, we employ the concept of Noether current in the language of differential geometry. First we notice that
\begin{equation}
\mathcal{L}_X\omega=\mathcal{L}_Xd\xi=d\left(\mathcal{L}_X\xi\right)=0
\end{equation}
thus the symplectic form is preserved.
Further, due to Cartan formula
\begin{align}\label{Cartan_formula-1}
\mathcal{L}_X  =  \iota_X \circ d + d\circ \iota_X
\end{align}
where $\iota_X$ means contraction. Applying it to $\xi$ we find
\begin{equation}
0=\mathcal{L}_X\xi=\omega(X,\,\cdot\,)+d(\xi(X))
\end{equation}
with the interior product $(\iota_X d\xi)(\,\cdot\,)=(\iota_X\omega)(\,\cdot\,):= \omega(X,\,\cdot\,)\;$. Thus 
\begin{equation}
X=\chi_f
\end{equation}
where the function $f=-\xi(X)$ (we used here nondegeneracy of the symplectic form). Finally
\begin{equation}
X[\;\cdot\;]=\chi_f(\;\cdot\;)=\{f,\;\cdot\;\}
\end{equation}
and the proof is completed.
\end{proof}

Let $\Phi$ be a group such that it allows a representation in form of the symplectomorphisms of $\M$. By this we mean the following: Each $\phi\in\Phi$ gets associated with a diffeomorphism on $\M$ whose action (in an abuse of notation) we denote by 
\begin{align}
\phi:\M\to\M,\hspace{30pt} m \mapsto \phi(m)
\end{align}
This diffeomrophism can be extended to acting on functions and tensor fields on $\M$ as well and we choose the following convention:

\begin{df}\label{definition:1}
Let $\phi:\M\to\M$ be a diffeomorphism, then we define the push-forward action $\phi_*$ on functions $f\in C^\infty(\M)$, vector fields $X\in \prescript{1}{}{C}^\infty(\M)$ and 1-forms $\sigma\in \prescript{}{1}{C}^\infty(\M)$ as 
\begin{align}
\phi_* f := f \circ \phi^{-1}, \hspace{20pt} (\phi_* X )(f)|_m:= X(\phi^{-1}_* f)|_{\phi^{-1}(m)} \hspace{20pt} (\phi_* \sigma)(X)|_m := \sigma (\phi^{-1}_* X)|_{\phi^{-1}(m)}
\end{align}
where $m\in\M$. These generalise to any tensor $T: (^1C^\infty(\M))^a\times (_1C^\infty(\M))^b \to C^\infty(\M)$:
\begin{align}
(\phi_* T)[X_1,...,X_a;\sigma_1,...,\sigma_b]|_m := T[\phi_*^{-1} X_1 ,...;\phi_*^{-1} \sigma_1,...] |_{\phi^{-1}(m)}
\end{align}
\end{df}

\begin{cor}
This definition is motivated by the fact that it preserves the group product on $\Phi$ with respect to the action of $\phi_*$, in the sense that for $\phi,\phi'\in\Phi$
\begin{align}
\phi\textcolor{white}{'}_* \phi'_* = (\phi\; \phi')_*
\end{align}
Moreover, contraction is preserved: Let $a'<a$ and $b'<b$ then
\begin{align}\label{contraction}
&\phi_*(T[X_1,...X_{a'}, \;\cdot\;,...,\;\cdot\; ;\sigma_1,...,\sigma_{b'},\;\cdot\;,...\;\cdot\;]) =\\ &\hspace{50pt}(\phi_* T)[\phi_* X_1, ...,\phi_*X_{a'},\;\cdot\;,...,\;\cdot\;;\phi_*\sigma_1,...\phi_*\sigma_{b'},\;\cdot\;,...,\;\cdot\;]\nonumber
\end{align}
\end{cor}

In fact, we can define also the map just on the space of tensors in the point.
We denote this linear map by
\begin{equation}
\phi_{*m}\colon \prescript{n}{k}{T}_{\phi^{-1}(m)}(\M)\rightarrow \prescript{n}{k}{T}_{m}(\M) 
\end{equation}
where $\prescript{n}{k}{T}_{m}(\M)$ are $n$ vectors and $k$ covectors at the point $m$.

The group $\Phi$ acts via {\it symplectormorphisms} if each $\phi\in \Phi\subset {\rm Symp}(\M)$ is a canonical transformation, i.e. the symplectic structure $\omega$ is invariant under its action:
\begin{align}\label{symplect}
\phi_* \omega = \omega
\end{align}
In this manuscript we consider only groups $\Phi$ for which (\ref{symplect}) holds. Let us remind that for a symplectomorphisms $\phi$ the action of pull-back on a vector field $\chi$ is given by $\phi^* \chi = (\phi^{-1})_* \chi$, so that $\phi_* \circ \phi^* = id$. Structures that are invariant under the action of $\Phi$ are paramount for our objective, thus we spell out explicitly:
\begin{df}
We say that a function $f$ or a vector field $X$ are $\Phi${\it -invariant} if for all $\phi\in\Phi$ respectively holds:
\begin{align}
\phi_* f = f,\hspace{40pt} \phi_* X = X
\end{align}
If $X$ is defined only on a subset $U\subset\M$ (the domain) then it is $\Phi${\it -invariant} if $\phi(U)=U$ for every $\phi\in \Phi$ and for $\phi\in \Phi$ and $m\in U$:
\begin{align}
\phi_{*\phi(m)} X|_m = X|_{\phi(m)}
\end{align}
\end{df}

\begin{lm}\label{lemma_invF_to_invVF}
Let $f$ be a $\Phi$-invariant function and $\Phi$ preserves the symplectic form, then the Hamiltonian vector field $\chi_f$ is also $\Phi$-invariant.\\
If $m\in\M$ is $f$-regular, then $\phi(m)$ is also $f$-regular and $\phi_{*\phi(m)}\chi_f|_{m}=\chi_f|_{\phi(m)}$.
\end{lm}

\begin{proof}
From $\Phi$-invariance of $\omega$ and $f$ follows 
\begin{align}
&\omega|_{\phi(m)}(\phi_{*\phi(m)}\chi_f|_{m},\;\cdot\;)\overset{(\ref{symplect})}{=}(\phi_{*\phi(m)}\omega|_m)(\phi_{*\phi(m)}\chi_f|_{m},\;\cdot\;)=\nonumber\\
&\overset{(\ref{contraction})}{=}\phi_{*\phi(m)}(\omega|_m(\chi_f|_m,\;\cdot\;))=\phi_{*\phi(m)}(df|_m)=d(\phi_{*}f)|_{\phi(m)}=df|_{\phi(m)}
\end{align}
and since the symplectic form is by definition weakly non-degenerate
\begin{equation}
\chi_f|_{\phi(m)}:=\phi_{*\phi(m)}\chi_f|_m
\end{equation}
i.e. the vector field is invariant under $\Phi$.
\end{proof}

\begin{lm}\label{lemma_invVF_to_comFLow}
Let $X$ be a $\Phi$-invariant vector field. Its unique 1-parameter family $\varphi_t$ of diffeomorphisms is defined via the following equation for any $f\in C^\infty(\M)$
\begin{align}
\left.\dfrac{df \circ \varphi_t}{dt}\right|_{t = 0} = X(f).
\end{align}
Then, the flow commutes with the transformations in $\Phi$, i.e.: 
\begin{align}\label{commuting_flows}
\varphi_t\circ \phi = \phi \circ \varphi_t
\end{align}
In other words, the evolution due to $X$ commutes with the $\Phi$-transformations.	
\end{lm}

\begin{proof}
Let us notice that
\begin{equation}
\tilde{\varphi}_t=\phi^{-1}\circ \varphi_t\circ \phi
\end{equation}
satisfies
\begin{equation}
\left.\dfrac{df \circ \tilde{\varphi}_t}{dt}\right|_{t = 0} = \phi^*(X)(f)=X(f),
\end{equation}
because $X$ is $\Phi$-invariant. From uniqueness follows
\begin{equation}
\tilde{\varphi}_t=\varphi_t\Longrightarrow \varphi_t\circ \phi = \phi \circ \varphi_t
\end{equation}
that finishes the proof.
\end{proof}

\subsection{Invariant submanifolds and symmetry restriction}
\label{section:2_Symmetry_Red}

Our interest is in the set of points $m\in\M$ that are invariant under the action of $\Phi$, which we denote by:
\begin{align}
\label{Mbar-def}
\overline{\mathcal M} := \{m \in \mathcal M : \phi(m) = m \ \text{for all} \ \phi \in \Phi\}
\end{align} 
that we assume to be a submanifold.\footnote{
It will be true in application to general relativity on a finite lattice as well as (in suitable sense) for FRW with compact spatial slices. In the case of classical FRW with noncompact sections the situation is more complicated.}
Every function on the phase space $\M$ can be restricted to a function on $\overline{\M}$ in the following way: Let $f\in C^\infty(\mathcal{M})$ then we define its restriction to $\overline{\M}\subset \M$ as
\begin{align}
f|_{\overline{\M}} (m) := f (m),\hspace{20pt} m\in\overline{\M}
\end{align}
Let us remark that there is no preferred way to restrict vector fields $X\in \prescript{1}{}{C}^\infty(\M)$ to $\overline{\M}$, in general unless they are tangent to $\overline{\M}$. By this we mean that the vector field is actually $X \in \prescript{1}{}{C}^\infty(\overline{\M})$. However, the restriction of any $n$-form $\Omega$ to $\overline{\M}$ is possible by defining:
\begin{align}
\Omega|_{\overline{\M}} \;\;\; : \;\;\; (\prescript{1}{}{C}^\infty(\overline{\M}))^n \hspace{10pt}&\to \hspace{10pt} C^\infty(\overline{\M})\\
\bar{X}_1,...,\bar{X}_n \hspace{10pt}& \mapsto\hspace{10pt} \Omega|_{\overline{\M}}(\bar{X}_1,...,\bar{X}_n):=\Omega(\bar{X}_1,...,\bar{X}_n)|_{\overline{\M}}
\end{align}

\begin{df}
We say that $\Phi$ is {\it clean} if the following conditions are satisfied
\begin{enumerate}
\item The set $\overline{\M}$ is submanifold of ${\mathcal M}$.
\item It is a symlectic manifold, i.e $\overline{\omega} :=\omega|_{\overline{\mathcal M}}$ is (weakly) non-degenerate.
%\item If the vector $v\in T_m\M$ is invariant under action of $\Phi$ for $m\in\overline{\M}$ then $v$ is tangent to $\overline{\M}$.
\end{enumerate}
\end{df}

\begin{remark}
One can also consider this condition locally, without assuming that the set $\overline{\M}$ is non-singular everywhere.
\end{remark}

\begin{remark}
Not all actions are clean even if $\overline{\M}$ is a submanifold. The counter example is an action of $\R$ with Hamiltonian $\frac{1}{2}x^2$ on manifold $(x,p)$ with standard symplectic form. One can check that $\overline{\M}=\{(x,p)\colon x=0\}$ is not a symplectic submanifold.
\end{remark}

However the following fact holds 

\begin{fact}\label{Fact1}
Suppose that $\Phi$ is compact and $\overline{\M}$ is a manifold, then $\Phi$ is clean.
\end{fact}

\begin{proof}
We consider $v\in T_m\M$ for $m\in \overline{\M}$ and such that
\begin{equation}
\phi_{*m}(v)=v
\end{equation}
We can extend it to the vector field $X$ on $\M$ such that $X|_m=v$. Such an extension is not unique, however the following holds independently. We define (by group averaging)
\begin{equation}
Y:= \int_\Phi {\rm d}\mu_H(\phi)\; (\phi_*X)
\end{equation}
where $\mu_H$ is the unique normalised left- and right-invariant Haar measure on $\Phi$. The vector field $Y$ is invariant under $\Phi$ and $Y|_m=v$. By lemma \ref{lemma_invVF_to_comFLow} the flow $\varphi_t$ generated by this vector is also invariant. Let us now notice that
\begin{equation}
\forall_{\phi\in \Phi}\ \phi\varphi_t(m)=\varphi_t\phi(m)=\varphi_t(m)
\end{equation}
thus $\varphi_t(m)\in \overline{\M}$ and the vector $Y$ is tangent to $\overline{\M}$ in $m$. We come to the conclusion that $v$ is tangent to $\overline{\M}$.

Let us now assume that $\overline{\omega}$ is degenerate, thus there exists a nontrivial vector $v\in T_m\overline{\M}$ such that
\begin{equation}
\overline{\omega}(v,\cdot\,)=0
\end{equation}
As $\omega$ is weakly non-degenerate, there exists $v'\in T_m\M$ such that (after embedding $v$ to $T_m\M$)
\begin{equation}
\omega(v,v')=1
\end{equation}
We now notice as $\phi(m)=m$
\begin{equation}
1=(\phi_*\omega)(\phi_*(v),\phi_*(v'))=\omega(v,\phi_*(v'))
\end{equation}
thus also
\begin{equation}
\omega(v,\int_\Phi {\rm d}\mu_H(\phi)\; \phi_*(v'))=1
\end{equation}
However, $\int_\Phi {\rm d}\mu_H(\phi)\; \phi_*(v')$ is an invariant vector, so it belongs to $T\overline{\M}$, that contradicts assumption about degeneracy.
\end{proof}

The goal of {\it symmetry restriction} is to simplify computations on the $\Phi$-invariant submanifold $\M$, by reducing all components before the actual computation. Thus, it is important to note:

\begin{fact}
For $f_1,f_2$ arbitrary functions on $\M$
\begin{equation}
(f_1+fg_2)|_{\overline{\mathcal M}}=f_1|_{\overline{\mathcal M}}+f_2|_{\overline{\mathcal M}}, \ \ \ \ \ (f_1\,f_2)|_{\overline{\mathcal M}}=f_1|_{\overline{\mathcal M}}\;f_2|_{\overline{\mathcal M}}
\end{equation}
\end{fact}

Further, we want to investigate the flow of Hamiltonians under such restriction. In particular, it is of interest to find the relation of a Poisson bracket $\{\cdot,\cdot\}_{\overline{\M}}$ with the Poisson bracket $\{\cdot,\cdot\}$ on $\mathcal M$:\\
Assuming the action of $\Phi$ is clean, then by fact \ref{Fact1} the manifold $\overline{\M}$ is equiped with symplectic form $
\overline{\omega}=\omega|_{\overline{\mathcal M}}$ and $\overline{\mathcal M}$ is a legitimate ``reduced phase space''. In particular, $\overline{\omega}$ defines a Poisson bracket $\{\cdot,\cdot\}_{\overline{\mathcal M}}$ on functions on $\overline{\mathcal M}$  by:
\begin{align}\label{def_pb}
\{f_1,f_2\}_{\overline{\M}}:=\overline{\omega} (\chi^{\overline{\omega}}_{f_1}, \chi^{\overline{\omega}}_{f_2}) = \chi^{\overline{\omega}}_{f_1} (f_2)
\end{align}
with $\chi^{\overline{\omega}}_f$ the {\it Hamiltonian vector fields} of $f$ defined by (\ref{hamiltonian_vf}) with respect to the symplectic structure $\overline{\omega}$.

\begin{lm} \label{lemma:1a-v1}
Let $\M$ be finite dimenional and $f$ be a $\Phi$-invariant function on $\M$ (i.e., $\phi_*f = f$ for all $\phi \in \Phi$) and $\Phi$ is clean then the Hamiltonian vector field $\chi_f$  is tangent to $\overline{\M}$.
\end{lm}

\begin{proof}
Due to lemma \ref{lemma_invF_to_invVF} $\chi_f$ is $\Phi$-invariant and due to lemma \ref{lemma_invVF_to_comFLow} its flow $\varphi_t$ commutes with $\phi\in\Phi$, i.e.: $\varphi_t\circ\phi=\phi\circ \varphi_t$. Let us now consider a point $m\in\overline{\M}$ then
\begin{equation}
 \phi (\varphi_t(m))= \varphi_t (\phi(m))=\varphi_t (m)
\end{equation}
thus also $\varphi_t(m)\in\overline{\M}$ and the vector $\chi_f$ generating the flow is tangent to $\overline{\M}$.
\end{proof}

The following is a version of previous lemma in the infinite dimensional setup:

\begin{lm} \label{lemma:1a}
Let $\Phi$ be compact and $f$ be a $\Phi$-invariant function on $\M$ (i.e., $\phi_*f = f$ for all $\phi \in \Phi$) and $\Phi$ is clean then the Hamiltonian vector field $\chi_f$ is tangent to $\overline{\M}$ in every $f$-regular point.
\end{lm}

\begin{proof}
Due to lemma \ref{lemma_invF_to_invVF} $\chi_f$ is $\Phi$-invariant. Let us consider a $f$-regular point $m\in \overline{\M}$. The vector $v=\chi_f|_m$ satisfies
\begin{equation}
\phi_{*m}v=v
\end{equation}
We can extend $v$ to a smooth vector field $X$ on $\M$. We define
\begin{equation}
Y:= \int_\Phi {\rm d}\mu_H(\phi)\; (\phi_*X)
\end{equation}
Let us notice that $Y_m=v$. Due to lemma \ref{lemma_invVF_to_comFLow} its flow $\varphi_t$ commutes with $\phi\in\Phi$, i.e.: $\varphi_t\circ\phi=\phi\circ \varphi_t$. Let us now consider a point $m\in\overline{\M}$ then
\begin{equation}
 \phi (\varphi_t(m))= \varphi_t (\phi(m))=\varphi_t (m)
\end{equation}
thus also $\varphi_t(m)\in\overline{\M}$ and the vector $v$ is tangent to $\overline{\M}$.
\end{proof}

\begin{lm} \label{lemma:3}
Let $F$ be a function on $\mathcal M$ such that its Hamiltonian vector field $\chi_F^\omega$ is tangent to $\overline{\mathcal M}$ at a $F$-regular point $m\in\overline{\M}$, and let $g$ be an arbitrary function on $\M$. Moreover, let $\Phi$ be clean\footnote{This ensures that $\bar{\omega}$ is non-degenerate, otherwise the theorem is in general not true. E.g. a constraint can produce a flow on its constraint surface, while its restriction is identically equal to zero. This is because the constraint surface has in general a degenerate symplectic structure.}. Then $m$ is also $F|_{\overline{\M}}$-regular and
\begin{equation}
\{F,g\}|_m=\{ F|_{\overline{\M}}, g|_{\overline{\M}} \}_{\overline{\M}}|_m
\end{equation}
In particular, if $\chi_F^\omega$ is tangent to $\overline{\mathcal M}$ everywhere (e.g. when $F$ is $\Phi$-invariant) then
\begin{equation}
(\{F,g\})|_{\overline{\mathcal M}}=\{ F|_{\overline{\mathcal M}}, g|_{\overline{\mathcal M}} \}_{\overline{\mathcal M}}
\end{equation}
\end{lm}
\begin{proof}
Since the vector field $\chi_F^{\omega}$ (associated with the full symplectic form $\omega$) is tangent to $\overline{\mathcal M}$ we can restrict it to $\overline{\M}$. We denote the restriction by $\chi_F^\omega|_{\overline{\M}}$. We have at $m$
\begin{equation}\label{lm_pb_restricted}
\chi^\omega_F|_{\overline{\M}}(g|_{\overline{\mathcal M}})|_m=\chi^{\omega}_F(g)|_m=\{F,g\}|_m
\end{equation}
The fact that $\chi^\omega_F$ is tangent allows us to restrict 
\begin{equation}
\omega|_{\overline{\mathcal M}}(\chi^\omega_F|_{\overline{\M}},\cdot)=
\omega(\chi^\omega_F,\cdot)|_{\overline{\mathcal M}}
\end{equation}
thus at $m$
\begin{equation}
\omega|_{\overline{\mathcal M}}(\chi^\omega_F|_{\overline{\M}},\cdot)=({\rm d}F)|_{\overline{\mathcal M}}={\rm d}|_{\overline{\mathcal M}}(F|_{\overline{\mathcal M}})
\end{equation}
However, we can define a Hamiltonian vector field $\chi^{\overline{\omega}}_F$ on the reduced manifold by
\begin{equation}
\omega|_{\overline{\mathcal M}}(\chi^{\overline{\omega}}_F,\cdot)|_m={\rm d}|_{\overline{\mathcal M}}(F|_{\overline{\mathcal M}})|_m
\end{equation}
with $\overline{\omega}=\omega|_{\overline{\mathcal M}}$ which is by assumption weakly non-degenerate. Hence, we have at $m$
\begin{align}\label{lm_vf_restricted}
\chi^{\overline{\omega}}_F|_m:=\chi^{\omega}_F|_{\overline{\M},m}
\end{align}
This result is crucial for the following line of equations:
\begin{align}
\{F,g\}|_m\overset{(\ref{lm_pb_restricted})}{=}\chi^\omega_F|_{\overline{\M}}(g|_{\overline{\mathcal M}})|_m\overset{(\ref{lm_vf_restricted})}{=}\chi^{\overline{\omega}}_F(g|_{\overline{\mathcal M}})|_m\overset{(\ref{def_pb})}{=}\{F|_{\overline{\mathcal M}},g|_{\overline{\mathcal M}}\}_{\overline{\mathcal M}}|_m
\end{align}
which ends the proof.
\end{proof}

This lemma is important as it allows to get easy access to the flow generated by symmetric functions. For example, we will see in subsection \ref{section:2_relation} an application thereof in the case of theories with gauge symmetries\footnote{The example of interest for a system with gauge symmetries and constraints is general relativity. Here,  $\Phi$ is a subgroup of gauge transformations and we will see that many constraints are already identically satisfied on $\overline{\M}$.}. Further, the lemma concerns the dynamics on the reduced manifold, which is enunciated in the following theorem.
\begin{tm} \label{theorem:2}
{\bf (Symmetry Restriction of Dynamics)} Let $(\overline{\mathcal M},\overline{\omega})$  be the symmetry restriction of $(\mathcal M, \omega)$ on which $\Phi$ acts clean, and let $H$ be a $\Phi$-invariant function. From lemma \ref{lemma:3}, it follows that the flow generated by $H|_{\overline{\mathcal{M}}}$ agrees with the flow of $H$ on $\overline{\M}$. In other words, the evolution with respect to $H$ of any observable $O:\mathcal{M} \to \mathbb{R}$, when restricted to $\overline{\mathcal M}$, agrees with the evolution of $O|_{\overline{\mathcal{M}}}$ with respect to $H|_{\overline{\M}}$ computed via $\{\cdot,\cdot\}_{\overline{\M}}$.
\end{tm}

\subsection{Constraints and symplectic reduction}
\label{section:2_Symplectic_Red}

This section recalls the basic results of Marsden-Weinstein-Meyer symplectic reduction for constrained systems \cite{Mey:73,MW:74,MMOPR:07,Sni:13}. Constraints (denoted in the following by $\hat{J}$) are the generators of gauge transformation, i.e. a subgroup of symplectomorphisms  on a phase space $(\M,\omega)$. To allow comparison with our framework (and for later application) we state the reduction theorems for symplectic structures and the dynamics.\\

Assume that a finite dimensional Lie group $G$ acts symplectically on $\M$
\begin{equation}\label{def_actionG}
\Pi\colon G\rightarrow {\rm Symp}(\M), \quad{\rm that\; is}\quad \Pi(g)_*\omega=\omega
\end{equation}
with the $\Pi(g)_*$ as in definition \ref{definition:1}.\\
We assume that any one parameter family $\{g_t\}_{t\in\mathbb{R}}$ in $G$ posses a global generator $\hat{J}$ (in the sense that $\chi_{\hat{J}}$ generates the flow $\Pi(g_t)$) and moreover that the assignment of the generators is linear (i.e. for two families $g_t,g'_t$ with generators $\hat{J}, \hat{J}'$ respectively we have that $d/dt\;\Pi(g_tg'_t)|_{t=0}=\chi_{\hat{J}+\hat{J}'}$). Both properties are automatically given, if $(\Pi,G)$ allows for a {\it momentum map}: \\

\begin{df} {\bf (Momentum Map)} Let $(\M, \omega)$ be a symplectic manifold on which a Lie Group $G$ acts symplectically and smoothly via $\Pi$. A map $J: \M \to \mathfrak{g}^\ast$ (the dual of the Lie algebra of G) is said to be a \emph{momentum map}, iff 
\begin{enumerate}
\item $J$ is equivariant that is $\forall g\in G$:
\begin{align}\label{equivariance}
J(\Pi(g)m)=Ad^\ast_{g^{-1}} J(m), \hspace{20pt}\text{where}\hspace{20pt} Ad_g(\tau):= \frac{d}{dt}\; g\,e^{t\tau} g^{-1}|_{t=0}
\end{align}
\item For all elements $\tau\in\mathfrak{g}$ it is
\begin{align}\label{momentum_map:flow_generator}
\chi_{\hat{J}(\tau)} = \frac{d}{dt}\Pi(e^{t\tau})|_{t=0}
\end{align}
where $\hat{J}(\tau) :\mathcal M \to \R$ is defined by $\hat{J}(\tau)(m)=J(m)(\tau)$.
\end{enumerate}
\end{df}

In particular, every point $m\in\M$ is $\hat{J}(\tau)$-regular for $\tau\in \mathfrak{g}$.\footnote{This is true for Gauss and vector constraints in GR as they form a moment map. However, the scalar constraints do not belong to this class.}
Let us remind known facts about the momentum map (see \cite{MMOPR:07})

\begin{fact}
The momentum map is not unique. If one $J_1$ exists, another can be obtained by adding a constant map $\Delta J\colon \M\rightarrow g^*$ with property that 
\begin{align}\label{fact:constant_map}
Ad_g^*\Delta J=\Delta J\;.
\end{align}
Then, $J_2=J_1+\Delta J$ fulfills (\ref{equivariance}) and (\ref{momentum_map:flow_generator}). Conversely, the difference between any two momentum maps $J_1,J_2$ is a  constant $\Delta J$ obeying (\ref{fact:constant_map}).
\end{fact}

\begin{proof}
Adding such 
$\Delta J\colon \M\rightarrow g^*$ 
preserves the properties of the momentum map: \eqref{momentum_map:flow_generator} due to being constant and equivariance due to (\ref{fact:constant_map}).\\
On the other hand if $J_1,J_2$ are two momentum maps then their difference $\Delta J=J_1-J_2$ satisfies
\begin{equation}
d\widehat{\Delta J}(\tau)=d\hat{J_1}(\tau)-d\hat{J_2}(\tau)=
\omega(\chi_{\hat{J_1}(\tau)},\cdot)-\omega(\chi_{\hat{J_2}(\tau)},\cdot)
\end{equation}
However, by (\ref{momentum_map:flow_generator}) $\chi_{\hat{J_1}(\tau)}=\chi_{\hat{J_2}(\tau)}$ and thus $\widehat{\Delta J}(\tau)$ is constant. Using this fact we can check
\begin{align}
&\Delta J (m) = \Delta J(\Pi(g^{-1})(m))=J_1 (\Pi(g^{-1})m)-J_2(\Pi(g^{-2})(m))=\nonumber\\
&=Ad^*_g J_1(m)-Ad^*_g J_2 (m)=Ad^*_g \Delta J(m)
\end{align}
that is showing the fact.
\end{proof}

\begin{remark}
The requirement 1. of ``equivariance'' is sometimes dropped in the literature. If we consider $\Delta J$ without (\ref{fact:constant_map}), then equivariance is not preserved. Equivariance always holds up to a constant \cite{MMOPR:07}, however not always one can find $\Delta J$ such that $J+\Delta J$ is equivariant. The example is the group $\R^2$ with $\hat{J}(\alpha,\beta)=\alpha x+\beta p$. Let us notice that $J$ is an additional structure i.e. not every pair $(\Pi,G)$ admits a momentum map.
\end{remark}

However: 

\begin{fact}\label{lm:presymplectic}
Suppose that there exists a presymplectic potential 
\begin{equation}
\xi\in \prescript{}{1} C^\infty(\M),\quad \omega=d\xi
\end{equation}
that is $G$-invariant (i.e. $\Pi(g)_\ast\xi=\xi $).
Then the vector field $X_\tau=\frac{d}{dt}\Pi(e^{t\tau})|_{t=0}$ defines a momentum map via
\begin{equation}
\hat{J}(\tau)=-\xi(X_\tau)\;.
\end{equation}
\end{fact}

\begin{proof}
Let us notice that as $\Pi(g)^\ast\xi=0$ we also have
\begin{equation}
{\mathcal L}_{X_\tau}\xi=0
\end{equation}
and from Cartan formula
\begin{equation}
 {\mathcal L}_{X_\tau}\xi=d(\xi(X_\tau))+(d\xi)(X_\tau,\cdot)
\end{equation}
we get 
\begin{equation}
\omega(X_\tau,\cdot)=-d(\xi(X_\tau))
\end{equation}
Thus $-\xi(X_\tau)$ is a generator of $X_\tau$. Let us now notice that as $\xi$ is invariant under group action
\begin{equation}
\Pi(g)^*(\xi(X_\tau))=\xi(\Pi(g)^*X_\tau)
\end{equation}
However $\Pi(g)^*X_\tau=X_{Ad_g\tau}$ thus
\begin{equation}
\Pi(g)^*(\xi(X_\tau))=\xi(X_{Ad_g\tau})
\end{equation}
that is showing equivariance.
\end{proof}

\begin{fact}
Let $G$ be compact and acting symplectically on $\M$. If the symplectic structure $\omega$ of $\M$ is exact then there exists a moment map.
\end{fact}

\begin{proof}
Let us notice that $\omega$ being exact means
\begin{equation}
\omega={\rm d}\tilde{\xi}
\end{equation}
with a presymplectic potential $\tilde{\xi} \in \Omega^1(\M)$.
However, in general $\tilde{\xi}$ is not $G$-invariant. Let us notice that
\begin{equation}
\xi:=\int_G {\rm d}\mu_H(g)\ \Pi(g)_*\tilde{\xi}
\end{equation}
with $\mu_H$, the normalised Haar-measure on $G$, satisfies: $\Pi(g)_*\,\xi=\xi $, i.e. $G$-invariance, and also ${\rm d}\xi=\omega$ because $\omega$ is $G$-invariant by definition (\ref{def_actionG}). 
Thus, we can construct the moment map $J$ by fact \ref{lm:presymplectic}
\begin{equation}\label{J:compact}
\hat{J}(\tau):=-\xi\left(\left.\frac{d}{dt}\Pi\left(e^{t\tau}\right)\right|_{t=0}\right)
\end{equation}
\end{proof}

In  constrained systems we understand $\hat{J}(\tau)$ as the set of all constraints and are interested in the locus of the constraints. That is all $m\in\M$ such that $J(m)(\tau)=0$ for all $\tau\in \mathfrak{g}$ -- in other words: such that $J(m)$ is the $0$ map on $\mathfrak{g}$. Hence, the constraint surface is defined as the preimage of $0$:
\begin{equation}\label{J-1}
J^{-1}(0):=\{m\in \M\colon J(m)=0\}
\end{equation}
Note that the constraint surface is preserved by action of $G$ due to equivariance of $J$.\\
We consider all points in $\M$ to be physically equivalent that lie on the same orbit 
\begin{equation}
[m] := \{m' \in \mathcal M : \Pi(g)(m') = m \ \text{for some} \ g \in G\}
\end{equation}

\begin{remark}
The symplectic reduction is a description of first class constraints. Only first class (on shell) constraints are targeted by the momentum map.
\end{remark}

In nice situations (for example in the neighbourhood of the points where stabiliser $G_m=\{g\in G\colon \Pi(g)(m)=m\}$ is trivial) the resulting quotient space 
\begin{equation}\label{MmodulusG}
\M/\!\!/G=\{[m]\colon m\in J^{-1}(0)\}
\end{equation}
is a manifold with the projection
\begin{align}\label{projection_MWM}
{\bf P}^{J}\colon J^{-1}(0)\hspace{10pt}&\rightarrow\hspace{10pt} \M/\!\!/G\\
m' \hspace{10pt}&\mapsto\hspace{10pt} [m] {\rm \;such\;that\;}m'\in [m]
\end{align}

It turns that $\M/\!\!/G$ inherits a symplectic structure $\omega_J$ due to the symplectic reduction theorem of Mardsen-Weinstein-Meyer:

\begin{tm}{\bf (Symplectic Reduction Theorem \cite{Mey:73,MW:74,MMOPR:07})}\label{tm:sympl_red}
Let $(\M,\omega)$ be a symplectic manifold on which $G$ acts symplectically and has a momentum map $J:\M \to \mathfrak{g}^*$.\\
When $\M/\!\!/G$ is a manifold then $\M//G$ is moreover a symplectic manifold with symplectic from $\omega_J$ defined by
\begin{align}
{\bf P}_*^J \omega_J = \omega |_{J^{-1}(0)}
\end{align}
with ${\bf P}^J_*$ denoting the push-forward of ${\bf P}^J$ from (\ref{projection_MWM}).\footnote{Note that while restricting $\omega$ to $J^{-1}(0)\subset \M$ is in general degenerate, said degeneracy is due to the equivalence classes and thus removed by $P^J$.}
\end{tm}

Physically, we are interested only in gauge invariant observables $O$, by which we mean those functions whose restriction to the constraint surface is invariant under the action of $G$, that is: $O|_{J^{-1}(0)}$ is $G$-invariant. A gauge invariant observable can be transferred to
$\M/\!\!/G$ and its flow stays in $J^{-1}(0)$, which is the following basic fact about symplectic reduction \cite{MMOPR:07,MR86}

\begin{tm} {\bf (Symplectic Reduction of Dynamics)}\label{thm:3}
Let $H$ be a $G$-invariant function on $J^{-1}(0)$ and $H^J: \M/\!\!/G\to\mathbb{R}$ defined by $H=H^J\circ {\bf P}_J$. Then, $\chi_H^\omega=\chi_H$ is also $G$-invariant and its flow leaves $J^{-1}(0)$ invariant. Finally, $\chi_{H^J}^{\omega_J}={\bf P}^J_{\ast}\;(\chi_H^\omega |_{J^{-1}(0)})$.
\end{tm}

\begin{remark}
The Hamiltonian $H$ does not need to be invariant on the whole phase space. In particular, the following transformation rule
\begin{equation}
\forall_{g\in G}\hspace{5pt}  g_*H=H+J(\lambda_g),\quad \lambda_g\colon \M\rightarrow  {\mathfrak{g}}
\end{equation}
is covered by the theorem. Moreover, even when the constraints do not form a Lie algebra but some linear space $W$, the statement can be generalised given that the following additional condition holds:
\begin{equation}
\forall_{\tau\in W} \hspace{5pt} \{\hat{J}(\tau), H\}=\hat{J}(\lambda_\tau),\quad \lambda_\tau\colon \M\rightarrow  {\mathfrak{g}}
\end{equation}
This ensures that the dynamics can be defined unambiguously on the symplectically reduced phase space.
\end{remark}

\begin{remark}
We finish this section with a brief remark on the situation when $G$ is {\it not} generated by its Lie algebra. A typically situation for such a scenario are not-connected groups. If $G$ is not connected, we denote the component connected to the identity by $G^o$ and (keeping mind that all components of $G$ are isomorphic to $G^o$)
\begin{align}
\mathcal{G}:= G/G^o
\end{align}
is called the {\it mapping class group}. Now, $\mathfrak{g}$ is the tangent space of $G$ at the identity element, thus its flow can only generate $G^o$. Hence, we see that the moment map is related to and describes only $G^o$ not the whole $G$.
The symplectormorphisms associated to the rest $G$ cannot be expressed via constraints,
they are so called ``large'' gauge transformations. Of course, they preserve \eqref{J-1}  as well, and we take the modulus with respect to the whole gauge group \eqref{MmodulusG}. Alternatively, one may perform symplectic reduction restricted to $G_o$ and treat the large diffeomorphisms manually afterwards.\footnote{An example for this is GR, where $\Diff(\sigma)$ is not connected. The diffeo-constraint $D$ actually only generates $\Diff^0(\sigma)$ and there are many nontrivial diffeomorphisms that still can be implemented by symplectic transformations but are not of this form. An intersting aspect is that the canonical analysis forces us only to implement $\Diff^o(\sigma)$ not all of $\Diff(\sigma)$. It is therefore not agreed upon, whether the elements of $\mathcal{G}(\sigma)$ produce physical change or not \cite{Giu95,Giu97,Giu06}.} In particular the Hamiltonian on $J^{-1}(0)$ is supposed to be invariant also under large gauge transformations.
\end{remark}

\subsection{Relation between symmetry restriction and symplectic reductions}
\label{section:2_relation}
This section will elucidate on the relation between symmetry restriction and symplectic reduction. Let us therefore consider $\Phi\subset G\,$.\footnote{This situation is similar to the one encountered in GR which will be topic of the next sections} By restricting to $\overline{\M}$ from (\ref{Mbar-def}) we are thus considering only symmetric configurations and moreover only in a specific gauge.\footnote{In this case $\overline{\M}$ will usually not be in the smooth part of $J^{-1}(0)$, but this issue is not relevant for our construction although it might be important if one wants to develop some perturbation theory to go beyond exact symmetry.} \\

First, we want to point out that the existence of a non-empty $\overline{\M}$ poses restrictions on the framework of symplectic reduction: theorem \ref{tm:sympl_red} requires $\M/\!\!/ G$ to be a manifold. In general there will be singular points, precisely those on which $G$ does not act free. In other words, at $\overline{\M}/\!\!/G\subset \M/\!\! /G$ the later one will not be a manifold and thus symplectic reduction is not applicable.\footnote{Especially in GR the diffeomorphism group is such that does not act free and hence gives rise to the mentioned problem at the points where symmetries are present. Of course, these are exactly the points into which we are interested over the course of this manuscript.} In a certain sense, {\it s} will be complementary to the symplectic reduction: first, we will symmetry reduce to $\overline{\M}$ on which a non-degenerate symplectic form exists -- thereby taking care of the singular points -- and afterwards we can perform a symplectic reduction with respect to the remaining symplectomorphisms in $G$.\\
Hence, we ask what part of $G$ remains after symmetry restriction with respect to $\Phi\subset G$ has been performed.\\

Our first considerations concern the constraint surface for which we are obviously interested into determining $J^{-1}(0)\cap \overline{\M}$. A possible strategy is to compute
\begin{align}
\hat{\overline{J}}(\tau):=\hat{J}(\tau)\mid_{\overline{\M}},\hspace{30pt}\forall \tau\in\mathfrak{g},
\end{align}

i.e. the restriction of the constraints to $\overline{\M}$, in order to find out which constraints are already trivially satisfied (i.e. for which $\tau$ it is $\hat{\overline{J}}(\tau)=0$) and which subset $\mathfrak{g}_\Phi\subset \mathfrak{g}$ remains to be imposed. In general, this may be quite cumbersome (especially in cases where the constraints do not form a Lie algebra), however also a general characterisation of the remaining constraints is possible:\\

\begin{lm}\label{lemma5}
The only non-vanishing constraints on $\overline{\M}$ are $\hat{\overline{J}}$ for which
\begin{align}
\overline{J}: \overline{\M} \to V
\end{align}
with 
\begin{equation}\label{remaining_constraints_V}
V=\{v\in\mathfrak{g}^*\colon \forall \phi\in \Phi\subset G,\ Ad^\ast_{\phi^{-1}}v=v\}
\end{equation}
Moreover, if $\Phi$ is compact then $V$ is the dual to the Lie algebra
\begin{equation}\label{remaining_constraints_gphi}
\mathfrak{g}_\Phi=\{\tau\in \mathfrak{g}\colon \forall_{\phi\in\Phi}\;Ad_\phi(\tau)=\tau\}
\end{equation}
\end{lm}
\begin{proof}
When restricting the moment map $J$ to $\overline{\M}$, it follows automatically for $\phi\in\Phi$ and $m\in\overline{\M}$:
\begin{align}\label{equiv_on_mbar}
J(m)=J(\phi_* m)= Ad^*_{\phi^{-1}} J(m)
\end{align}
Thus, the possible remaining constraints are $\hat{\overline{J}}$ for which $\overline{J}:\overline{\M}\to V$ and $V$ is given by (\ref{remaining_constraints_V}).\\
Now, assume that $\Phi$ is compact, then by group-averaging we have a $Ad_\Phi^*$-invariant scalar product $\langle .,.\rangle$ on $\mathfrak{g}$  which allows us to identify covectors with vectors, hence $\mathfrak{g}\equiv \mathfrak{g}^*$.\\
Therefore, we can associate with every $J(m)\in \mathfrak{g}^*$ and Lie algebra element via:
\begin{align}
J(m)( \;.\;) =: \langle s_{J(m)}, \;.\;\rangle
\end{align}
We will show that $s_{J(m)}$ is actually in $\mathfrak{g}_\Phi$ from (\ref{remaining_constraints_gphi}): Let $s\in\mathfrak{g}$ arbitrary and $v\in \mathfrak{g}^*$, $t\in \mathfrak{g}$ such that $\langle t,.\rangle =\alpha(.)$. Then: $\forall g \in G$:
\begin{align}
(Ad^*_g \alpha)(s)=\alpha (Ad_{g^{-1}} s) = \langle t, Ad_{g^{-1}} s\rangle= \langle Ad_g t ,s \rangle
\end{align}
where we used in the last step that the scalar product is invariant under the adjoint action. Now, $s$ was arbitrary, hence $(Ad^*_g \alpha)(.) = \langle Ad_g t ,.\rangle$. If we take $\alpha=J(m)$ with $m\in\overline{\M}$ and $g=\phi\in\Phi$  we can use equivariance (\ref{equiv_on_mbar}):
\begin{align}
\langle s_{J(m)},\; .\; \rangle =J(m) = Ad^*_\phi J(m) =\langle Ad_\phi s_{J(m)} ,\;.\; \rangle
\end{align}
Thus, $s_{J(m)}\in\mathfrak{g}_{\Phi}$.\\
It remains to show that all other constraints are vanishing: Let $\tau\perp \mathfrak{g}_\Phi$ be an element orthogonal to (\ref{remaining_constraints_gphi}). Then:
\begin{align}
J(m)(\tau) = \langle s_{J(m)}, \tau \rangle =0
\end{align}
Hence, $\hat{J}(\tau)=0$ for all $m\in \overline{\M}$. Finally, from linearity of $J$, when decomposing any $\tilde{\tau}\in\mathfrak{g}$ as $ \tilde{\tau}= \tau + \tau'$ with $\tau\perp \mathfrak{g}_\Phi$ and $\tau'\in\mathfrak{g}_\Phi$ it follows
\begin{align}
J(\tilde{\tau})=J(\tau)+J(\tau') =J(\tau')
\end{align}
\end{proof}
The $\hat{\overline{J}}(\tau)$ with $\tau\in \mathfrak{g}_\Phi$ are in general not vanishing must be treated separately. However, all other constraints are identically satisfied on $\overline{\M}$.\\ 

\begin{lm}
A known fact is that any equivariant momentum map $J$ has the property:
\begin{align}\label{standard_PB_formula}
\{\hat{J}(\tau_1), \hat{J}(\tau_2)\} = \hat{J}([\tau_1,\tau_2])
\end{align}
If either $\tau_1\in \mathfrak{g}_\Phi$ or $\tau_2\in\mathfrak{g}_\Phi$ then, also the Poisson bracket between two reduced constraints has the standard formula:
\begin{align}\label{reduced_PB_formula}
\{\hat{\overline{J}} (\tau_1),\hat{\overline{J}}(\tau_2)\}=\hat{\overline{J}}([\tau_1,\tau_2])
\end{align}
\end{lm}
\begin{proof}
First, we show (\ref{standard_PB_formula}): Let $\tau_1,\tau_2\in \mathfrak{g}$. From equivariance (\ref{equivariance}) take derivative at $t=0$ for all $m\in \M$:
\begin{align}\label{lemma_sidestep1}
\frac{d}{dt}|_{t=0} J(\Pi(e^{t\tau_2})m)= \frac{d}{dt}|_{t=0} J(M)(Ad_{e^{-t\tau_2}} \; .\; )=J(m)(ad_{-\tau_2}\;-\;) =J(m)([\;.\;, \tau_2])
\end{align}
where $ad_{\tau}=[\tau,\;.\;]$ denotes the standard derivative of the $Ad$, i.e. the adjoint action of the Lie algebra. Then
\begin{align}
\{\hat{J}(\tau_1),\hat{J}(\tau_2)\}&=d\hat{J}(\tau_1)(\chi_{\hat{J}(\tau_2)})=\chi_{\hat{J}(\tau_2)}[\hat{J}(\tau_1)]=\frac{d}{dt}|_{t=0}\Pi(e^{t\tau_2})[\hat{J}(\tau_1)]=\mathcal{L}_{\Pi(e^{t\tau_2})}\hat{J}(\tau_1)=\\
&=\frac{d}{dt}|_{t=0} \Pi(e^{t\tau_2})_* \hat{J}(\tau_1) \overset{(\ref{lemma_sidestep1})}{=} \hat{J}([\tau_1,\tau_2])
\end{align}
Regarding the second claim, if $\tau_1\in\mathfrak{g}_\Phi$, we get from (\ref{equivariance}) that
\begin{align}
J(\Pi(\phi)m)[\tau_1]=Ad^*_{\phi^{-1}} J(m)[\tau_1] = J(m)[\tau_1]
\end{align}
and thus
\begin{align}
\hat{J}[\tau_1] (\Pi(\phi)m)=\hat{J}[\tau_1](m)
\end{align}
that is, $\hat{J}[\tau]$ is $\Phi$-invariant. Hence, we can make use of \ref{lemma:1a}  to see that the Hamiltonian vectorfield $\chi_{\hat{J}[\tau]}$ is tangent to $\overline{\M}$ and with lemma \ref{lemma:3} we infer
\begin{align}
\{\hat{\overline{J}}(\tau_1),\hat{\overline{J}}(\tau_2)\}=\{\hat{J}(\tau_1),\hat{J}(\tau_2)\}|_{\overline{\M}}
\end{align}
all further steps are exactly as before.
\end{proof}

\begin{remark}
{\bf (Nongroup constraints)} 
There exist  physical application where one is interested in constraints which do not form a finite dimensional Lie algebra but some linear space $W$. \footnote{E.g. General Relativity, where the diffeomorphism constraints form the infinite dimensional Lie Algebroid whose structure functions depend on the metric (for an infinite dimensional group G, not every 1-parameter group is generated by constraints).} However, the important criterion to repeat the construction above is equivariance. That is, almost everything goes as before, if we have a moment map (now valued in linear space $W^*$ dual to $W$) 
\begin{equation}
J\colon \M\rightarrow W^*
\end{equation}
and the group $\Phi$ acts linearly on this space (by a generalisation of $Ad_g^*: \mathfrak{g}\to \mathfrak{g}$)
\begin{equation}
\kappa\colon \Phi\times W^*\rightarrow W^*
\end{equation}
such that $J$ is equivariant 
\begin{equation}
\kappa(g) (J(m))=J(g(m))
\end{equation}
Then the remaining constraints are characterised by $\overline{J}\colon \overline{\M}\rightarrow V$ where
\begin{equation}
V=\{w\in W^*\colon \forall g\in \Phi,\ \kappa(g)(w)=w\}
\end{equation}
Suppose now that $m\in \overline{\M}$ is $\hat{J}(w)$-regular for some $w\in V$. Then as $\hat{J}(w)$ is $\Phi$-invariant
\begin{equation}
\{\hat{J}(w),\hat{J}(w')\}|_m=\{\hat{J}(w)|_{\overline{\M}},\hat{J}(w')|_{\overline{\M}}\}|_{\overline{\M},m}
\end{equation}
The constraint algebra can be completely described on $\overline{\M}$. Let us stress that in this case we do not assume that every point is $\hat{J}(w)$-regular for all $w\in V$.
\end{remark}

\vspace{10pt}

\paragraph{Residual gauge transformations}

We turn now towards the second concern regarding the remnant of the gauge group $G$ after symmetry restriction with respect to $\Phi \subset G$: every $g\in G$ generates transformations which identify elements which we interpret to be physically equivalent. Next to ensuring that all $J$ are implemented on the constraint surface, we must ask further, what is the fate of the physical equivalent phase space points or in other words: how do the orbits of $g$ behave with respect to $\overline{\M}$?\\

We consider first a situation when two points $m,m'\in \overline{\M}$ are connected by a path $m(t)\in\overline{\M}$ which follows a trajectory generated by constraints
\begin{equation}\label{gauge:J}
\frac{dm}{dt}=\chi_{\hat{J}(\tau(t))}
\end{equation}
where $\tau(t)\in W$ is smooth. 

The flow is tangent to $\overline{\M}$ thus
\begin{equation}
\chi_{\hat{J}(\tau(t))}=\chi^{\overline{\omega}}_{\hat{\overline{J}}(\tau(t))}
\end{equation}
We used the fact that $V\subset W^*$ thus we can act on $\tau(t)$. We see that the points $m$ and $m'$ are also connected by a curve generated by constraints on $\overline{\M}$. The functions $\overline{J}(v)$ constitute (possibly overcomplete) basis of constraints restricted to $\overline{\M}$. Their flows are thus also a flow of some combinations of constraints.

In order to analyse it further we introduce a subgroup
\begin{equation}\label{gauge:G}
\overline{G}=\{g\in G\colon \forall_{m\in \overline{\M}} g(m)\in \overline{\M}\}
\end{equation}
These transformations acts symplectically on $\overline{\M}$. Together transformations \eqref{gauge:J} and \eqref{gauge:G} generate a group of reduced gauge transformations.

However, this is not enough to describe  gauge equivalence on $\overline{\M}$. It may happen that for two points $m,m'\in \overline{\M}$ there exists $h\in G$ such that $h\not\in \overline{G}$, but
\begin{equation}
h(m)=m'
\end{equation}
Similarly, we can have a flow generated by constraints which leave $\overline{\M}$ and then return. Apparently, in order to describe gauge equivalence we need to take into account more then the flow of constraints restricted to $\overline{\M}$ and the group of transformations preserving $\overline{\M}$. In what follows, we will say the two points are {\it globally related} if they are gauge equivalent in $\M$ but not connected by a flow of reduced gauge transformations. There is no obvious structure behind this identifications.\\

Summarizing, after restriction to $\Phi$-invariant configurations the remaining constraints  are given by $\overline{J}$. The gauge transformations are given by reduced gauge transformations, but there exist also global identifications.

\begin{remark} {\bf (Reduction in stages)}
Let us now assume that we have two clean groups $\Phi_1\subset \Phi_2$ and $\Phi_1$ is a normal subgroup of $\Phi_2$ (i.e. $\phi^{-1}\Phi_1\phi\subset \Phi_1$ $\forall \phi\in\Phi_2$) and let
\begin{equation}
H=\Phi_2/\Phi_1
\end{equation}
In such a case $\Phi_2$ preserves $\overline{\M}_1$ due to
\begin{align}
\phi_1 \phi_2 (m)=\phi_2 \phi_1 ' (m)  = \phi_2 (m) \quad \Rightarrow \quad \phi_2(m)\in \overline{M}_1
\end{align}
and its action factorises by $H$. We also have
\begin{equation}
J|_{\overline{\M}_2}=\overline{J}_1|_{\overline{\M}_2},\quad \overline{\M}_2=\{m\in\overline{\M}_1\colon \forall_{h\in H} h(m)=m\}
\end{equation}
We can thus perform reduction in stages.
\end{remark}

\subsubsection{Hamiltonian dynamics}

In constrained systems, the Hamiltonian is not uniquely defined. We can always correct it by adding constraint generators. However, in the systems that we will consider it does not pose any problem.

\begin{fact}
Let  $\Phi\subset G$ be a compact subgroup of the gauge transformations. There exists a $\Phi$-invariant Hamiltonian. Every two such Hamiltonians differ on $\overline{\M}$ by a combination of $\overline{J}$ constraints.
\end{fact}

\begin{proof}
Every two Hamiltonians differ by a combination of constraints $J$. After restricting to $\overline{\M}$ it becomes combination of $\overline{J}$ constraints. In order to obtain a $\Phi$-invariant Hamiltonian $H$ we can choose any Hamiltonian $H'$ and define
\begin{equation}
H:=\int_\Phi d\mu_H(\phi)\ \phi_*H'=\int_\Phi d\mu_H(\phi)\ (H'+J(\lambda_\phi))=
H'+\int_\Phi d\mu_H(\phi)\ J(\lambda_\phi)
\end{equation}
because $\int_\Phi d\mu_H(\phi)=1$ for normalised Haar measure.
This is also a Hamiltonian, but it is now $\Phi$-invariant.
\end{proof}

This situation covers our application, even if we consider deparametrised models.

\newpage
\section{Application to General Relativity}
\label{section:3}
In this section, we will apply the framework of {\it symmetry restriction} (developed in the previous section) to the Hamiltonian formulation of general relativity (GR). For this, we adopt in subsection \ref{section:3_ADM} the language of Arnowitt-Deser-Misner (ADM) \cite{ADM:62}, where the phase space is coordinatised by the spatial metric $q$ and its canonical conjugated momentum $P$ over some spatial manifold $\sigma$. In the ADM formalism GR turns out to be a totally constraint system, i.e. its Hamiltonian vanishes on the constraint surface due to being a combination of constraints, which encode the invariance under temporal and spatial diffeomorphisms. Background-independence enters the framework further by not only demanding vanishing of the spatial diffemorphism constraint, but actually demanding that all (even the non-connected part) of $\Diff(\sigma)$ forms a gauge group. $\Diff(\sigma)$ is the group of diffeomorphims on $\sigma$ and we devote subsection \ref{s311} to recall how this can be translated into an action of symplectomorphisms $\D(\Diff(\sigma))$ on the ADM phase space $\M$. The group of gauge transformations of GR is a product of $\Diff(\sigma)$ and the transformations generated by the flow of the constraints, see subsection \ref{section:3.3_gaugetrafos}. Afterwards, the stage is set for an application of the symmetry restriction theorem, which we discuss in subsection \ref{s312}: a symmetry group $\Phi\subset \D(\Diff(\sigma))$ needs to be chosen and due to spatial diffeomorphism-invariance of the theory, it is possible to find Lagrange multipliers inside the Hamiltonian of GR such that is invariant under $\Phi$ and hence makes theorems \ref{theorem:2} and \ref{thm:3} applicable.\footnote{It applies also to various deparametrised models.}
As we have already discussed in subsection \ref{section:2_relation}, in general there will be constraints whose vanishing remains to be imposed on the $\Phi$-invariant phase space $\overline{\M}\subset \M$ and we discuss their fate.\\
Having applications to quantum gravity in mind, we will also discuss a second incarnation of a phase space describing GR: section \ref{section:3_connection} introduces the connection formulation where we recall the phase space $\M_{AB}$ coordinatised by the Ashtekar-Barbero variables. $\M_{AB}$ is actually the phase space of a $SU(2)$ Yang-Mills theory, thus endowed with an additional Gauss constraint $G$. Further, there are two additional sets of constraints $D'$ and $C$, reassembling the spatial diffeomorphism $D$ and scalar constraints of the previous section. However, we highlight in section \ref{s321} an often ignored fact: the action of $D'$ does not produce purely spatial diffeomorphisms on $\M_{AB}$ but instead also induces a transformation on the Lie-algebra indices of the connection and its conjugated momentum. Of course, this can be remedied when realising that $D'$ and the equivalent of $D$ differ only by a Gauss constraint. In other words: many equivalent representation of ``spatial diffeomorphisms'' exist on $\M_{AB}$ which all agree only on shell, i.e. once vanishing of the Gauss constraint has been imposed. Hence, we discuss in subsection \ref{s322} standard symplectic reduction of $G$ by which we project back onto the familiar ADM phase space. When doing so, we are faced with one final problem in subsection \ref{s323}, namely that not every symmetry group $\Phi\subset\D(\Diff(\sigma))$ acting on $\M$ can be translated into a group $\Phi_{AB}$ of symplectomorphisms on $\M_{AB}$ which are generated by $D'$. We discuss in which situations a suitable extensions to a subgroup of symplectomorphisms generated by $D$ and $G$ can be achieved and how the Gauss transformation has to be chosen for a given diffeomorphism. Lastly, we prove that in such a case the symmetry restriction of $\M_{AB}$ with respect to $\Phi_{AB}$ is isomorphic to the symmetry reduction of the ADM phase space $\M$ with respect to $\Phi$, thus showing commutativity between symmetry restriction and symplectic reduction in our application.

\subsection{ADM formulation}
\label{section:3_ADM}

Let $\sigma$ be a 3-dimensional, spatial, orientable manifold (either compact or no-compact at the moment).\\
The Arnowitt-Deser-Misner variables \cite{ADM:62}\footnote{Throughout this article we work on the ``reduced ADM'' phase space where the primary constraints have been implemented, fixing lapse function and shift vector to Lagrange multipliers. However, basically the whole setting of symmetry restriction could also be lifted to the phase space including lapse and shift functions and their momenta, by realizing that the only $4$ metric invariant under a group $\Psi$ acting on the spacelike sections is a metric of the form
\begin{equation}
-N(x,t)^2dt^2+2q_{ab}(x,t)N^a(x,t)dx^bdt+q_{ab}(x,t)dx^adx ^b
\end{equation}
where for every $t$, $N$ is an invariant function $N^a$ invariant vector and $q_{ab}$ invariant metric. } coordinatise the phase space $\mathcal{M}$ of general relativity theory on $\sigma$, described by a metric field  $q_{ab}\in {}_{(2)}C^{\infty}(\sigma)$ and a conjugated momentum field $P^{ab}\in {}^{(2)}\widetilde{C}^{\infty}(\sigma)$:
\begin{align}
\M_{ADM} = \{ (q_{ab},P^{ab})\, :\, q\in {}_{(2)} C^\infty(\sigma), P\in {}^{(2)}\widetilde{C}^\infty(\sigma)\, ,\, \det(q)>0  \}
\end{align}
We will be using index notation with spatial indices $a,b,..=1,2,3$. The symplectic form on $\mathcal {M}$ is $\omega=d\xi$ with the symplectic potential $\xi= (2 /\kappa) \int_\sigma {\rm d}^3x\; P^{ab}(x)dq_{ab}(x)$ and reads explicitly:
\begin{align}\label{symplform-ADM}
\omega = \frac{2}{\kappa}\int_\sigma {\rm d}^3x \; dP^{ab}(x)\wedge dq_{ab}(x)
\end{align}
As usual Poisson brackets are defined through $\{f,g\}=-\omega(\chi_f,\chi_g)$:
\begin{align}
\{q_{ab}(x),P^{cd}(y)\}=\kappa \,\delta^c_{(a}\delta^d_{b)}\delta^{(3)}(x,y)
\end{align}
with $\delta_{(a}\delta_{b)}:=\frac{1}{2}(\delta_a\delta_b+\delta_b\delta_a)$.\\
We keep in mind that spatial indices $a,b...$ are pulled up and down with the spatial metric $q_{qb}$ or its inverse denoted by $q^{ab}$.

Let us shortly describe functional analytic side of this construction.
In the case of the cosmological models we choose as our Hilbert manifold an open subset of
\begin{equation}
(q_{ab},P^{ab})\in H^{c}(\Sigma,d\mu)
\end{equation}
where $H^c$ is a weighted Sobolev space with $c>0$ large enough such that the second derivatives of the fields are continuous. The condition for open subset is then $\det q\not=0$ for every point. \\
The weight needs to be chosen in such a way that symmetric configurations belong to this space. For the case of Bianchi models with noncompact Cauchy surfaces this leads to the symplectic form being ill-defined (see section \ref{section:4}). However in the case of compact Cauchy surfaces, the weight has no matter and we can define the symplectic form $\omega$ by the standard formula (\ref{symplform-ADM}).

We are mainly interested in the following class of functions on $\M$
\begin{equation}
F=\int_\Sigma N(x) f(J^2q, J^2P)
\end{equation}
where $J^n$ denotes $n$-th jets %jet = derivative up to n-kind, i.e. (q,\partial q,\patial^2 q)
and $N(x)$ is compactly supported smooth function on $\sigma$ and $f$ is smooth function of its variables. The crucial property of $F$ is that it is  a smooth function on $\M$. Moreover, every point of $\M$ with $q_{ab}$ and $P^{ab}$ smooth is $F$-regular. We are exclusively interested in those smooth points. The Banach manifold structure is a technical detail needed for applying our theory. Once the proper statements are proved we can take $c\rightarrow \infty$ to obtain a statement in smooth category. In what follows we will skip the details of this process.

Similar considerations apply also for the Ashtekar-Barbero phase space which is introduced later in \ref{section:3_connection}.

\subsubsection{Constraints and gauge transformations}

We subject this phase space to the following two families of constraints:
\begin{itemize}
	\item The spatial Diffeomorphism (or vector) constraint:
	\begin{align}\label{Diffeo_constraint-ADM}
	D_a =P^{bc}\partial_a q_{bc}-2\partial_bP^b_a+D_{\rm matter}
	\end{align}
	\item The Scalar constraint:
	\begin{align}\label{Scalar_constraint-ADM}
	C = \frac{\kappa}{\sqrt{\det(q)}}(P^{ab}P_{ab}-\frac{P^2}{2})-\frac{\det(q)}{\kappa}R^{(3)}+C_{\rm matter}
	\end{align}
\end{itemize}
where $P:=P^{a}_a$ and $R$ is the Ricci scalar of spatial geometry, i.e. a function of the metric $q_{ab}$. The functions $D_{\rm matter}$ and $C_{\rm matter}$ denote the contribution of some matter content.\\

\subsubsection{Spatial diffeomorphisms and $\Psi$-invariance}
\label{s311}
The terminology ``diffeomorphism constraint'' comes from the fact that $D_a$ indeed generates the {\it connected component} ${\Diff}^o(\sigma)$ of the spatial-diffeomorphim group in the following way: since the  the symplectic form is defined in a covariant way 
 every diffeomorhism $\psi\in \Diff(\sigma)$ defines a symplectic transformation in the following way
\begin{equation}
\D(\psi): \M \to \M,\quad \D(\psi)(P^{ab},q_{ab}):=(\psi_*(P^{ab}),\psi_*(q_{ab}))
\end{equation}
where the action of the diffeomorphism on a tensor field $T$ is defined by
\begin{equation}\label{diff-action}
(\psi_\ast T_{a...}^{b...})(x):=t_{a'...}^{b'...}(\psi^{-1}(x))\frac{\partial(\psi^{-1})^{a'}(x)}{\partial x^a}...\frac{\partial\psi^{b}(y)}{\partial y^{b'}}|_{y=\psi^{-1}(x)}...
\end{equation}
Let $\{\psi_t\}_{t\in\mathbb{R}}$ be a one parameter family of diffeomorphisms generated by a smooth compactly supported vector field $\xi^a$. Then $D[\vec{\xi}]$, the diffeomorphism constraint smeared against $\xi^a$, is actually the generator of $\D(\psi_t)$, that is:
\begin{align}
\{D[\vec \xi], f\} = \mathcal L_{\vec \xi} f = \lim_{t\to 0} \frac{1}{t}(\D(\psi_{-t})_{*}f-f)
\end{align}
for all functions $f: \M \to \M$.\\
We remind that the diffeomorphism constraints (smeared with smooth compactly supported vector fields) generate elements of $\D(\Diff(\sigma))$, but not all of them.

Particularly important will be the following subset of diffeomorphisms. We assumed that $\sigma$ is orientable. We denote by $\Diff^+(\sigma)$ the group of diffeomorphisms that preserves orientation and by $\Diff^-(\sigma)$ the remaining diffeomorphisms (those which changes orientation). These notions are independent of the particular choice of orientation, but they are defined only on orientable manifolds.\\

In the later sections we will be interested in systems that employ symmetries, that is there is a group of diffeomorphisms $\Psi\subset \Diff(\sigma)$ for which the spatial metric $q$ and its momentum $P$ are $\Psi$-invariant, i.e. $\Psi_*q=q$.\footnote{Common examples for symmetries in GR are often described via Killing vectors $k_i$ such that $\mathcal{L}_{k_i} q=0$. Obviously, this can be translated into the existence of 1-parameter groups of diffeomorphisms $\psi$ for which the metric remains invariant. However, in this manuscript we consider the more general case, including also symmetries not generated by Killing vectors.} As the diffeomorphism group is {\it simultaneously} a gauge group of GR, the group $\Psi$ can be deformed into another diffeomorphically equivalent group $\Psi'\sim \Psi$, i.e. only its internal geometry matters. In what follows, we will {\it fix a gauge} in the sense that we choose only one specific realisation $\Psi$ of the symmetry group and are interested in its representation as
\begin{align}
\Phi := \D (\Psi)
\end{align}
of symplectomorphisms on $\M$. Then, we will study symmetry restriction with respect to $\Phi$. Let us remind that $\overline{\M}$ is not a physical phase space as it will not be invariant under all constraints: it only represents symmetric configurations in a specific gauge. However, via a gauge-transformation (i.e. some diffeomorphisms) one may of course map $\overline{\M}$ into $\overline{\M}'$ , the invariant submanifold for $\Psi'\sim \Psi$, i.e. $\exists\, \xi\in \Diff(\sigma)$ with $\xi \circ\Psi' =\Psi$. But as $\xi$ is a diffeormophism, it has to be modded out at the end anyway. Hence, choosing this specific gauge is unproblematic.\\

\subsubsection{Gauge transformations in canonical theory}
\label{section:3.3_gaugetrafos}

According to the principle of equivalence, two solutions are equivalent if they differ by a spacetime diffeomorphism. We will now describe what does it mean in the canonical theory. The initial data $(q_{ab}, P^{ab})$ satisfying scalar and diffeomorphism constraints give rise via evolution \cite{ChoquetBruhat:69} to the maximal development. This maximal development is a globally hyperbolic spacetime $M$ with an embedding of the Cauchy surface
\begin{equation}
\Sigma'\colon \sigma\rightarrow M,\quad \Sigma=\Sigma'(\sigma)
\end{equation}
The restricted metric is equal to $q_{ab}$ when pulled back by $\Sigma'$ and $P^{ab}$ is related to the pulled back extrinsic curvature $K^{ab}$ by
\begin{equation}
P^{ab}=\frac{1}{\kappa}\sqrt{\det q}(K^{ab}-q^{ab} K^{cd}q_{cd})
\end{equation}
In order to determine $K^{ab}$ we fixed time orientation of the spacetime.

We can now ask the question when two initial data are equivalent. Namely, what is the condition for maximal developments of two initial data to be diffeomorphic by a diffeomorphism which does not change the time orientation? We will now answer this question.

\begin{lm}
Suppose that we have two embeddings of Cauchy surfaces into the globally hyperbolic solution of Einstein equations $M$
\begin{equation}
\Sigma'_i\colon \sigma\rightarrow M,\ \Sigma_i=\Sigma'_i(\sigma)\quad i=0,1,
\end{equation}
then there exists a smooth map
\begin{equation}\label{eq:124}
\Sigma'\colon \sigma\times\R\rightarrow M
\end{equation}
such that for every $t$, $\Sigma(t)=\{\Sigma'(x,t)\colon x\in \sigma\}$ is a Cauchy surface and $\Sigma(i)=\Sigma_i$ for $i=0,1$. 
\end{lm}

\begin{proof}
In fact, let us denote by $c$ the maximal time distance between $\Sigma_0$ and $\Sigma_1$. We consider a metric on $M\times \R$ (we denote additional dimension by $b$)
\begin{equation}
g_{\mu\nu}dx^\mu dx^\nu+(c^2+1)db^2
\end{equation}
It is also globally hyperbolic.
Let us take the surface $K=\Sigma_0\times \{0\}\cup \Sigma_1\times \{1\}$. It is an achronal set because two connected components lay far enough. By the theorem of \cite{BS:05}
one can extend this set to a smooth Cauchy surface $\tilde{\Sigma}\subset M\times \R$. This set is isomorphic with $\sigma\times \R$. The isomorphism is the map 
\begin{equation}
\tilde{\Sigma}'\colon \sigma\times\R\rightarrow M\times \R.
\end{equation}
We obtain $\Sigma'=\pi\circ \tilde{\Sigma}'$ where $\pi\colon M\times \R\rightarrow M$ is a projection.
\end{proof}

Additionally, we consider two embeddings of Cauchy surfaces $\tilde \Sigma_i'\colon \sigma\rightarrow M$ with $\Sigma_i=\tilde \Sigma'_I(\sigma)$ and corresponding data
$(q_{ab}^i,K_{ab}^i)$ defined by
\begin{equation}
q_{ab}^i=\tilde{\Sigma}_i^* g_{ab},\ K_{ab}^i=\tilde{\Sigma}_i^* \nabla_{(a} n^{\Sigma_i}_{b)}
\end{equation}
where $n^{\Sigma_i}_{\mu}$ is a normal vector to the Cauchy surface 
and we take $\Sigma'_i$ and $\Sigma'(t)$ as in the lemma. There exists a smooth diffeomorphism
\begin{equation}
\kappa_i\in \Diff(\sigma),\quad \tilde\Sigma'_i\circ\kappa_i=\Sigma_i'
\end{equation}
because corresponding Cauchy surfaces coincide. We now consider the metric and extrinsic curvature pull-backed by $\Sigma'(t)\circ\kappa_0$
\begin{equation}\label{eq:path-qK}
{q_{ab}'}^t, {K_{ab}'}^t
\end{equation}
with ${q_{ab}'}^0=q_{ab}^0$ and  ${K_{ab}'}^0=K_{ab}^0$. As we pulled back a solution to Einsteins equation, $t\mapsto({q_{ab}'}^t, {K_{ab}'}^t)$ is a path in the phase space, in fact on the constraint surface. We introduce tangent vector to this path
\begin{equation}\label{eq:tangent-qK}
\frac{d}{dt}{q_{ab}'}^t, \frac{d}{dt}{K_{ab}'}^t
\end{equation}

\begin{lm}
The tangent vector \eqref{eq:tangent-qK} in the phase space $\M_{ADM}$ belongs to the directions generated by constraints. 
\end{lm}

\begin{proof}
Suppose that we have a one parameter family of embeddings of Cauchy surface
\begin{equation}
\Sigma_t=\Sigma'(t,\sigma), \quad \Sigma'\colon I\times \sigma\rightarrow M
\end{equation}
with $\Sigma'$ smooth. We consider family of induced metric and extrinsic curvature
\begin{equation}
{q_{ab}'}^t, {K_{ab}'}^t
\end{equation}
The ADM momentum $P^{ab}$ is determined by them. Let us now consider a push forward vector $X={\Sigma'}^*\partial_t$ from $t\times \sigma$ on the surface $\Sigma_t$. We can decompose $X$
\begin{equation}
X_t=N_t'\vec{n}_t+\vec{N}_t'
\end{equation}
where $\vec{n}_t$ is future unit normal to $\Sigma_t'$ and $\vec{N}_t$ is tangent to $\Sigma_t'$. We can pull-back with $\Sigma_t'=\Sigma'(t,\cdot)$ both this objects to $\sigma$ (as the vector is tangent) to obtain
\begin{equation}
N_t={\Sigma_t'}_*N_t',\quad \vec{N}_t={\Sigma_t'}_*\vec{N}_t',
\end{equation}
The geometric interpretation of Hamiltonian formulation of GR says that
\begin{equation}
\frac{d}{dt}{q_{ab}'}^t=\{{q_{ab}'}^t, N_tC+N_t^iC_i\},\quad
\frac{d}{dt}{{P^{ab}}'}^t =\{{{P^{ab}}'}^t, N_tC+N_t^iC_i\},
\end{equation}
The flow is generated by $N(t)C+N^i(t)C_i$ thus it is tangent to the constraints direction.
\end{proof}

Finally the data $q_{ab}^1, K_{ab}^1$ is not exactly ${q_{ab}'}^1, {K_{ab}'}^1$ but it differs by a diffeomorphism
\begin{equation}
({\kappa_1}^{-1}\kappa_0)_*({q_{ab}'}^1,{K_{ab}'}^1)=(q_{ab}^1,K_{ab}^1)
\end{equation}
We proved that the gauge transformations are generated on shell, by scalar and vector constraints together with all spacial diffeomorphisms $\Diff(\sigma)$ if $\sigma$ is compact.\footnote{Thereby extending statements such as \cite{Blohmann:10} which were only concerned with the bracket relations stemming from the connected part of $\Diff(\sigma)$.}

Additionally, the map
$(q_{ab},P^{ab})\rightarrow (q_{ab},-P^{ab})$ corresponds to one particular time reversing diffeomorphism. Thus time reversing diffeomorphisms cannot be implemented as symplectic transformations (they exchange the sign in symplectic form).

In principle, one is interested in performing symplectic reduction from section \ref{section:2_Symplectic_Red} with respect to the group $\Diff(\sigma)$ and all scalar constraints in order to end up with only the physical degrees of freedom. 
However, in all generality the reduction is  a very hard task that has not been achieved up to full satisfaction as of today. Nonetheless, we will see that in the context of symmetry restriction many of the diffeomorphisms can be dealt with.
	\subsubsection{Fate of Einstein Equations under s}
\label{s312}
Assume a certain $\Phi= \mathbb{D}(\Psi)$ is fixed and we perform symmetry restriction to the invariant submanifold $\overline{\M}$. The generator of the dynamics in GR (at least for compact Cauchy surfaces) is totally constrained to vanish, that is it is a sum of the constraints multiplied by their respective Lagrange multipliers. Therefore, we have actually a whole family of Hamiltonians, parametrised by the possible choices of Lagrange multipliers $\lambda,\vec{\rho}$ that is:
\begin{align}
H(\lambda,\vec{\rho}):=\int_\sigma{\rm d}x\;\; [ \, \lambda(x)\,C(x)\, +\, \rho^a(x)\, D_a(x) \,]
\end{align}
All of the $H(\lambda,\vec{\rho})$ produce the same dynamics (i.e. gauge transformations) thus it is a free choice to pick the $\lambda,\vec{\rho}$ which are most suitable for our situation. For the purpose of symmetry restriction, it motivates itself to choose them such that
\begin{align}\label{transfLapseShift}
\phi_* \lambda= \psi_* \lambda,\hspace{30pt} \phi_* \vec{\rho} = \psi_* \vec{\rho}
\end{align}
Take note that we allow $\lambda,\rho$ to be dependent on the ADM phase space data.

Let us notice that for any point $m\in \M$ and $\phi=\D(\psi)$
\begin{equation}\label{spatial_diffeo_inv_of_GR}
C(\phi(m))=\psi_*C(m),\quad D(\phi(m))=\psi_*D(m)
\end{equation}
from spacial diffeomorphism invariance of the theory. Thus, it follows with (\ref{transfLapseShift}) that $H(\lambda,\vec{\rho})$ is $\Phi$-invariant and consequently theorem \ref{theorem:2} applies, i.e. the dynamics is correctly reduced to $\overline{\mathcal M}$.\\

However, this is not enough to ensure that all Einstein equations are satisfied. Indeed, physical points in the phase space are those for which the constraints $C[N]$ and $D[\vec N]$ vanish for all $N$ and $\vec N$ (of compact support). Fortunately, restriction to $\overline{\mathcal M}$ already takes care of many of these constraints, as we now show by utilizing the tools from  \ref{section:2_relation}. We understand the constraints as functions from the phase space into $\mathbb{R}$ (for simplicity we consider here smooth setting), and first focus on the general case where $\Phi$ might be non-compact.
From $\Phi$-invariance of $\overline{\mathcal{M}}$, it follows from (\ref{spatial_diffeo_inv_of_GR}) as $\phi(m)=m$ for $m\in\overline{\M}$ that 
\begin{align}
C\mid_{\overline{\M}}(m),\quad D\mid_{\overline{\M}}(m)
\end{align}
are invariant under $\Psi$ as functions on $\sigma$. Hence, we are in the first situation of lemma  \ref{lemma5} and can determine the remaining constraints on $\overline{\M}$ as follows:\\
Let us assume that there exists a finite dimensional basis of invariant functions and vector fields
\begin{equation}
C^\infty(\sigma)_\Psi=\Span\{f_i\;:\;f_i =\Psi_* f_i \}_{i\in I},\quad \prescript{1}{}{C}^\infty(\sigma)_\Psi=\Span\{\vec{v}_j\;:\;v_j =\Psi_* v_j\}_{j\in J}
\end{equation}
We can then define functions on the reduced phase space
\begin{equation}
c_i\colon \overline{\M}\rightarrow \R,\quad d_j\colon \overline{\M}\rightarrow \R,
\end{equation}
defined by identities
\begin{equation}\label{remaining_constraints_ADM}
C|_{\overline{\M}}=\sum_{i\in I} c_i f_i,\quad D|_{\overline{\M}}=\sum_{j\in J} d_j \vec{v}_j
\end{equation}
In summary, any possible remaining Einstein constraint, that is not vanishing on $\overline{\mathcal M}$, must be of the form (\ref{remaining_constraints_ADM}).\\

In case that $\Phi$ is compact, the second description of lemma \ref{lemma5} can be used: we parametrise the constraints by scalar functions $N$ and vector fields $\vec{N}$ spanning the hypersurface deformation algebra
\begin{align}
D[\vec N]:=\hat{D}(\vec N)\;:\; \M \to \mathbb{R},\hspace{30pt}
C[N]:=\hat{C}(N)\;:\; \M \to \mathbb{R}
\end{align}
Then, the only remaining constraints are these for which
\begin{align}
Ad_\phi (\vec{N}) =\vec{N},\hspace{30pt} 
Ad_\phi (N) = N \hspace{20pt}\forall \phi\in  \Phi
\end{align}
With some standard manipulation, we see for any vector field $v\in\mathfrak{g}$ that for $f: G=\Diff(\sigma) \to \mathbb{R}$:
\begin{align}
Ad_\phi (v)f &= \phi_* v\phi_*^{-1} f = \phi_* [ d(\phi_*^{-1}f)(v) ] = [\phi_*d(\phi_*^{-1}f)](\phi_* v)=[d(\phi_*\phi_*^{-1}f)](\phi_*v)=\nonumber\\
&=df(\phi_* v)=(\phi_*v)f
\end{align}
(using (\ref{contraction}) in the third step)
thus requiring $\vec{N}=\psi_*\vec{N}$ and similar for scalar functions $N=\psi_* N$. The characterisation of the remaining constraints in form (\ref{remaining_constraints_ADM}) follows. \\

The last step consists out of analysing the remaining gauge orbits on $\overline{\M}$: Suppose that we have a constraint $C'$ which preserves $\overline{\M}$. Then its Hamiltonian vector field $\chi^\omega_{C'}$ is tangent to $\overline{\M}$ and thus for any function
\begin{equation}
\{C',g\}|_{\overline{\M}}=\{C'|_{\overline{\M}},g|_{\overline{\M}}\}_{\overline{\M}}
\end{equation}
and the action can be only nontrivial if $C'|_{\overline{\M}}\not=0$. The functions $c_i$ and $d_j$ constitute (possibly overcomplete) bases of scalar and vector constraints restricted to $\overline{\M}$. Thus, each one creates a flow which is equivalent to the flow of some combinations of constraints.\\
Of course there might be some diffeomorphisms that leave $\overline{\M}$ and return thus providing physical identification that need to be check case by case. We will call such transformations global. An example to be discussed later are the permutation of scale factors and conjugated momenta in Bianchi I model.

\subsubsection{Deparametrisation } \label{cosmology-deparam} 

Our methods apply also to some deparametrised models, if the gauge transformations of the deparametrised model contain subgroup $\Psi$. Typically,
in this case we assume that the scalar constraint was solved at the unreduced level (i.e., on the full phase space), for example with the use of additional scalar field. Thus, we are left with a physical Hamiltonian, $H = C[1]$, together with the remaining constraints $\vec{D}$.

Let us notice that this might not be true if for the solving constraints we used additional nonsymmetric objects like a choice of observer etc.

\subsection{Connection formulation of General Relativity}
\label{section:3_connection}

The Ashtekar-Barbero variables \cite{Ash:86,Bar:94} coordinatise the phase space $\mathcal{M}_{AB}$ of a $SO(3)$ Yang-Mills theory of a connected manifold $\sigma$, described by a $SO(3)$ connection $A^I_a \;:\; \sigma\to \mathbb{R}$ and a non-abelian electric field $E^b_J\;:\; \sigma \to \mathbb{R}$  (with internal $\mathfrak{su}(2)$-indices $I,J,...=1,2,3$ and spatial indices $a,b,..=1,2,3$). We assume that electric field is nondegenerate. This means that we restrict to electric fields such that $\det E \not= 0$.

From the symplectic potential
\begin{align}\label{potential_AB}
\xi_{AB} = \frac{2}{\kappa\beta} \int_\sigma {\rm d}^3x\; E^a_I(x) dA^I_a(x)
\end{align}
we obtain the symplectic form on $\mathcal {M}_{AB}$ given by:
\begin{align}\label{symplform}
\omega = \frac{2}{\kappa \beta}\int_\sigma {\rm d}^3x \; dE^a_I(x)\wedge dA^I_a(x)
\end{align}
where $\beta\in \mathbb{R}\setminus\{0\}$ is the Immirzi parameter. As usual Poisson brackets are defined through $\{f,g\}=\omega(\chi_f,\chi_g)$ and read:
\begin{align}\label{AB_PB}
\{E_I^a (x), A_J^b(y)\} = \frac{\kappa\beta}{2}\delta^a_b \delta^J_I \delta^{(3)}(x,y)
\end{align}
We will subject this phase space to the following three families of constraints:
\begin{itemize}
	\item The Gauss constraint:
	\begin{align}\label{Gauss_constraint}
	G_J = \partial_a E^a_J +\epsilon_{JKL}A^K_aE^a_L
	\end{align}
	\item The spatial diffeomorphism (or vector) constraint:
	\begin{align}\label{Diffeo_constraint}
	D_a' = \frac{2}{\kappa\beta}F_{ab}^J E^b_J -\beta G_J  K^J_a +D_{\rm matter}
	\end{align}
	\item The scalar constraint:
	\begin{align}\label{Scalar_constraint}
	C = \frac{{\rm sgn}(\det(E))}{\kappa}(F^J_{ab}-(1+\beta^2)K_{a}^MK_{b}^N\epsilon_{MNJ})\epsilon_{JKL}\frac{E^a_KE^b_L}{\sqrt{|\det(E)|}}+C_{\rm matter}
	\end{align}
\end{itemize}
where $F_{ab}$ is the curvature of the connection
\begin{equation}
F_{ab}^I=\partial_a A^I_b-\partial_b A^I_a+\epsilon^I_{JK}A_a^J A_b^K
\end{equation}
$K_{ab}$, the extrinsic curvature of the ADM formalism\footnote{We point out that in the ADM formulation $K_{ab}$ is symmetric and a short manipulation shows that this  immediately forces the antisymmetric part of the right hand side of (\ref{extrCurvature}) to vanish, i.e. $K_{[a}^I e_{b]}^I=0$. This is in turn is equivalent to the Gauss constraint (\ref{Gauss_constraint}).}, can be expressed by the formula
\begin{equation}\label{extrCurvature}
K_{ab}= K_{(a}^I E_{b)}^I\sqrt{|\det(E)|},\hspace{30pt}\beta K^I_a:= A^I_a -\Gamma^I_a
\end{equation}
with $\Gamma^I_a$ being the spin connection \cite{Thi:07} and
\begin{align}\label{cotriad}
E^I_b := \frac{1}{2}{\rm sgn}(\det(E)) \epsilon^{abc} \epsilon_{IJK} \frac{E^a_J E^b_K}{|\det(E)|}
\end{align}
with the property that $E^a_J E^J_b=\delta^a_b$ and $E^a_I E^J_a =\delta^J_I$. As before, $D_{\rm matter}$ and $C_{\rm matter}$ denote the contribution of some matter content.\\

\subsubsection{(Kinematical) Symplectic Reduction to ADM phase space}

Finally, we subject this phase space to the gauge groups $O(3)$ whose connected component is generated by the Gauss constraint. The group of $O(3)$ gauge transformation (in case of connected manifold) can be written as a product of group of $SO(3)$ gauge transformations (pointwise multiplication)
\begin{equation}
\G^+=C^\infty(\sigma, SO(3))
\end{equation}
and the two element group $\{{\mathbb I},-{\mathbb I}\}$, that corresponds to $\pm {\mathbb I}$ constant transformations.
\begin{equation}\label{eq-GG}
\G=\G^+\times \{{\mathbb I},-{\mathbb I}\}
\end{equation}
The action for $O\in \G^+$ is given by a symplectic transformation\footnote{ Let us notice that if $O$ is constant then the connection transforms also by simple rotation. This will be the case in many applications for example in FRW $k=0$ cosmology.}
\begin{equation}\label{action_of_gauss_so(3)}
\G(O)(E,A)=(O^J_IE_J^a,O^I_JA^J_a+\frac{1}{2}\epsilon_{J}^{\;\;IK}(\partial_a O^J_{\ L}){O^{-1}}^L_{\ K})
\end{equation}
and such transformation of compact support are generated by Gauss constraints $G^I$.

The action of $-{\mathbb I}$ is implemented by a symplectic transformation
\begin{equation}
\G(-{\mathbb I}) (E,A):= (-E^a_I, 2\Gamma^I_a- A^I_a)
\end{equation}
It easy to check that this is a valid choice, as $2\Gamma^I_a-A^I_a$ transforms as a connection and the constraints are actually invariant under this transformation (actually only after adding a corresponding Gauss constraint to the scalar constraint as well).

We perform now  Marsden-Weinstein-Meyer symplectic reduction with respect to the group of $O(3)$ gauge transformations, that is imposing the Gauss constraint. The locus of the constraint $\G^{-1}(0) \subset \M_{AB}$ can actually be coordinatised by the $SO(3)$ gauge-invariant functions $q_{ab}$ and $P^{cd}$ which are canonically conjugated (see \cite{Thi:07} for details) 
\begin{align}\label{temp_equ147}
q^{ab}:= \frac{\delta^{IJ} E_I^a E_J^b}{|\det(E)|},\hspace{30pt} P^{ab}:= \frac{1}{\beta\kappa}(q^{c(a}q^{b)d}-q^{ab}q^{cd})(A^J_c-\Gamma_c^J)q_{de} E^e_J
\end{align}
where $P^{ab}$ is symmetric due to the Gauss constraints. From symplectic reduction we also know that we have a map
\begin{equation}
{\bf P}^G\colon \G^{-1}(0) \rightarrow \M_{AB}/\!\!/G \equiv \M_{ADM}
\end{equation}
where $\G^{-1}(0)$ denotes points where the Gauss constraints vanish. The reduced phase space becomes equivalent to the ADM phase space.  Conversely, if two points  $m,m'\in \M_{AB}$ give the same ADM data, i.e. $(q,P)(m)=(q,P)(m')$ they are related by a $O(3)$ gauge transformation (it can be a large gauge transformation).
This is in fact fibration with $\G$ group. The map is given by (\ref{temp_equ147}).
\\

Finally, one has still to impose equations (\ref{Diffeo_constraint})-(\ref{Scalar_constraint}) for all $x\in\sigma$ or one can equivalently study the vanishing of $C[N],\,\vec{D}[\vec{N}]$, i.e. their respective smearings against a lapse function $N$ and a shift vector $\vec{N}$ (of compact support).\\
Let us notice that the diffeomorphism and scalar constraints are $O(3)$
 invariant thus they can be pulled down to the ADM phase space (in fact the scalar constraint is invariant only up to Gauss constraint, but it is unimportant when Gauss constraints vanishes). After reducing the phase space by the Gauss constraint the diffeomorphism constraints are equal to $\vec{D}_{ADM}$ thus they are generators of spacelike diffeomorphisms. Similarly scalar constraints reduce to the ADM scalar constraints, thus the theory on Ashtekar Barbero phase space is equivalent to ADM formulation of General Relativity.\\

\begin{remark}
Some instances in the literature define the Ashtekar-Barbero variables with gauge group $SO(3)$ \cite{Thi:07}. As both groups share the same Lie algebra $su(2) \cong so(3)$, they obey the same Gauss constraint $G$.
However, in order to ensure that symplectic reduction from the extended phase with respect to the chosen symmetry group returns only $\M_{ADM}$, it is actual necessary to adapt the Ashtekar-Barbero phase space to the correspondingly chosen gauge group. That is:
\begin{align}
\M_{AB}^{O(3)} &=\{(E_I\in \prescript{1}{}{C}^\infty(\sigma),A^J\in \prescript{}{1} {C}^\infty(\sigma))\, ,\, \det(E) \neq 0\},\\
\M_{AB}^{SO(3)} &= \{(E_I\in \prescript{1}{}{C}^\infty(\sigma),A^I\in \prescript{}{1}{C}^\infty(\sigma))\, , \, \det(E) > 0 \}
\end{align}
Despite the fact that quantisations based on $SO(3)$ is better understood than $O(3)$, it is useful to allow also $O(3)$ transformations.  Throughout this manuscript we use the choice $\M_{AB} \equiv \M^{O(3)}_{AB}$ simplifying the notation accordingly.
\end{remark}

\subsubsection{Reappearance of the spatial diffeomorphism group}
\label{s321}
We devote this subsection to recall the relation of the vector constraint (\ref{Diffeo_constraint}) of the connection formulation with the spatial diffeomorphism constraint of the ADM formalism (\ref{Diffeo_constraint-ADM}).\\

First, we point out that there exists a natural generalisation of the action of the diffeomorphisms \eqref{diff-action} on the spatial manifold $\sigma$ for tensors with internal indices: Let $\Psi\subset \Diff(\sigma)$ and $T$ some $(n,m)$-tensor field with $n_3$-many internal indices $I_1,...I_{n_3}$. The action $\forall \psi\in\Psi$ on $T$ is defined by
\begin{align}\label{ActionOnGaussInv}
(\psi_* T_{a...\;I_1..I_{n_3}}^{b...})(x):=(\psi_\ast t^{a...}_{b...})_{I_1,...I_{n_3}}(x)=T_{a'...\; I_1...I_{n_3}}^{b'...}(\psi^{-1}(x))\frac{\partial(\psi^{-1})^{a'}(x)}{\partial x^a}...\frac{\partial\psi^{b}(y)}{\partial y^{b'}}|_{y=\psi^{-1}(x)}...
\end{align}
That is, the action does not touch internal indices (we regard them as labels).\\
We translate this into an action of symplectormorphisms on the phase space $\M_{AB}$:
\begin{align}\label{Action_on_AE_var}
\D_{AB}(\psi)f(A^I_a,E^a_I):=f(\psi_*(A^I_a),\psi_* (E^a_I))
\end{align}
Obviously, we are now interested in its generator $D$ such that $\forall \psi\in \Psi$ there exists $\vec{N}$ with
\begin{align}\label{generator_of_ABdiffeos}
\psi_* f = \sum_{n=0}^\infty \frac{t^n}{n!} \{ \int_\sigma {\rm d}^3x\; D_a(x) N^a(x)  ,f\}_{(n)}
\end{align}
It is tempting to think that $D$ is the actually the diffeomorpism constraint $D'$ from (\ref{Diffeo_constraint}). However it is not true. 
In order to determine the actual generator we will use fact \ref{fact-presymplectic} and the presymplectic potential $\xi_{AB}$ \eqref{potential_AB}\footnote{The crucial part is the diffeomorphism invariance of $\xi_{AB}$ in order for this derivation to work. Similarly, one may use that $\xi_{AB}$ is invariant under SO(3) transformation as well in order to derive $G_I$ from (\ref{Gauss_constraint}) as the generator of gauge transformations. In fact $\xi_{AB}$ is also invariant under all $O(3)$ transformation. It is enough to check it for $\G(-{\mathbb I})$, i.e. $\G(-{\mathbb I})_*\xi_{AB}=\xi_{AB}-2\int E_I^a\delta A_a^I=\xi_{AB}$ by \cite{Thi:07} eq. (4.2.27).}

The following is well known \cite{Ash:91%eq 3 p. 88
,AL04%eq 2.29 
 }:

\begin{fact}\label{spt_diffeo_AB}
Let $\{\psi_t\}_{t\in\mathbb{R}}$ be a smooth 1-parameter group of diffeomorphisms acting via (\ref{ActionOnGaussInv}), such that each $\psi_t$ is identity outside of some compact support and $\psi_0={\rm id}$.\\
 The proper generator of the action $\D_{AB}(\psi_t)$ is $D$ in the sense that there exists some compactly supported vector field $N$ such that
\begin{equation}
\frac{d}{dt}|_{t=0}\D_{AB}(\psi_t) =\{ D[\vec{N}] ,\;.\;\}\;,\quad D_a:=\frac{2}{\kappa\beta}( F_{ab}^IE^b_I-G_IA^I_a)
\end{equation}
\end{fact}

\begin{proof}
We use fact \ref{fact-presymplectic} that is valid also in our infinite dimensional setting. We can compute
for $X=\frac{d}{dt}\D_{AB}(\psi_t)$
\begin{align}
\xi_{AB}(X)=\int {\rm d}^3x\; E^a_I(x) \, X[A^I_a(x)] =\int {\rm d}^3x \;E^a_I(x)\frac{dA^I_a(x)}{ds} =\int {\rm d}^3x\; E^a_I(x) \mathcal{L}_{\vec{N}} A^I_a(x)
\end{align}
because we know that for spatial diffeomorphisms there exists a vector field $N$ on $\sigma$ such that
\begin{align}
\frac{dA^I_a(x)}{ds}=\mathcal{L}_{\vec{N}} A^I_a(x)
\end{align}
Therefore, the generator of the spatial diffeomorphisms $\mathbb{D}_{AB}(\psi_t)$ is given by
\begin{align}
F&= -\xi_{AB}(X) = - \int {\rm d}^3x\; E_I^a(x) \mathcal{L}_{\vec{N}}A^I_a(x) = \int {\rm d}^3x\; [E_I^a N^b\partial_b A^I_a + E^a_I A^I_b \partial_a N^b](x)\nonumber\\
&=  \int {\rm d}^3x\;  [E^a_I N^b \partial_{[a} A^I_{b]} -E^a_I \partial_a (A^I_b N^b)](x)\nonumber\\
&=  \int {\rm d}^3x\;  [E^a_I N^b(\partial_{[a}A^I_{b]}+\epsilon_{IJK} A^J_a A^K_b)+A^I_b N^b (\partial_a E^a_I+\epsilon_{IJK} A^J_a E^a_K)](x)\nonumber\\
&= \int {\rm d}^3x\; N^b(x) [E^a_I F^I_{ab} + A^I_b G_I](x) =  D[\vec{N}] 
\end{align}
It is easy to perform a consistency check (recall that $E^c_K$ is a density of weight 1)
\begin{align}
\{D[N], E^c_K(x)\}&= N^a\partial_a E^c_K - E^a_K\partial_a N^c+ (\partial_a N^a) E^c_K  =\mathcal{L}_{\vec{N}} E^c_K\\
\{D[N], A^K_c(x)\}&= N^a \partial_a A_c^K +A^K_a \partial_c N^a =\mathcal{L}_{\vec{N}} A^K_c
\end{align}
which proves the fact.\\
\end{proof}

The constraint $D'$ differs from $D$ by a Gauss constraint. Although all inner indices are contracted $D$ is not $SO(3)$ invariant, because of the nontensor transformation of the connection.

Concerning this situation, one should note that $\Lambda^IG_I$ for
$\Lambda^I$ dependent on phase space functions does not generate Gauss transformation on the whole phase space. However, if we restrict to the locus of the Gauss constraint by $\G^{-1}(0)$, then
\begin{equation}
\{\Lambda^I G_I, F\} = \Lambda^I\{G_I,F\}+\{\Lambda^I,F\} G_I=\Lambda^I\{G_I,F\}
\end{equation}
we obtain again $SO(3)$  transformation\footnote{We remind that Gauss constraints generate only $SO(3)$ not $O(3)$ gauge transformations.}. From this follow an important fact. Addition of the term $G_J  K^J_a$ in (\ref{Diffeo_constraint}) is not only vanishing on the Gauss constraints surface, but also it does not alter the action of the generator on the $O(3)$ invariant functions. Actually, we are allowed to modify
the vector constraint (\ref{Diffeo_constraint}) with any element of the form
\begin{align}
{\bf D}[\vec N, \Lambda]= D [\vec N] +G[\Lambda]
\end{align}
for any choice of $\mathfrak{so}(3)$-smearing $\Lambda$.

\begin{remark}
The generator $D$ seems distinguished by its relation to the diffeomorphism transformations, thus the lack of $O(3)$ invariance may be surprising. However, there is a simple geometric reason for that. In fact, we should think about internal indices rather in terms of $O(3)$ bundle. Such bundles in $3$ dimensions always admit three independent global sections and these are our three indices. Transformation generated by $D$ is not changing this particular trivialisation, thus it is not purely geometric but depend on additional choice. For interpretation of other possible generators see \cite{Simone2019}.
\end{remark}

In summary, we are interested in both the transformations described by $\G(O)$, i.e. Gauss transformations, as well as $ \D_{AB}(\psi)$, i.e. the spatial diffeomorphisms. Let us notice that $\G(O)$ and $\D_{AB}(\psi)$ are symplectic transformations on $\M$.\\
The group generated by these transformations is a semidirect product
\begin{equation}
\D_{AB}(\Diff( \sigma))\ltimes \G
\end{equation}
where $\G$ is defined by \eqref{eq-GG}.
Let us notice that
\begin{equation}\label{eq:semi-direct}
\D_{AB}(\psi)\G(O)\D_{AB}(\psi^{-1})=\G(\psi^*O)
\end{equation}
We will use the basis of gauge transformations given by $\D_{AB}$ and $\G$, but in the moment map it is more useful to use the basis given by $G$ and $D'$ or $D'':=\frac{2}{\kappa \beta}F^J_{ab}E^b_J$ (instead of $D$). This is due to the nicer version of equivariance property \eqref{equiv-prop:diff}. 
For the same reason we will use the scalar constraint given in equation (\ref{Scalar_constraint}), albeit one can also be modified it by the Gauss constraint to obtain a generator of the flow acting on spatial indices only \cite{Thi:07}.

\subsubsection{Total symplectic reduction}
\label{s322}
Indeed, the symplectic reduction (in stages) from section \ref{section:2_Symplectic_Red} fits into the outlined framework, as the division of constraints presents a proper choice. We have seven constraints per point $G_I(x), {D''}^a(x),  C(x)$, which can be brought to together as the corresponding momentum maps:
\begin{equation}
G\colon \M_{AB}	\rightarrow C^\infty(\sigma, so(3)^*),\quad
D''\colon \M_{AB}\rightarrow \prescript{}{1}{C}^\infty(\sigma),\quad
C\colon \M_{AB} \rightarrow C^\infty(\sigma),
\end{equation}
We should regard theses spaces as a dual to the Lie algebra.
In other words
\begin{align}
G[\,\cdot\,] (\,\cdot \,) :\hspace{20pt} \M_{AB} &\times  C_0^\infty(\sigma,so(3)) \hspace{10pt}\to \mathbb{R}\\
(A^I_a,E^b_J)(x) &\times \Lambda(x) \hspace{50pt}\mapsto G[\Lambda] (A,E)\nonumber
\end{align}
Let us notice that $O$ naturally acts on $C^\infty(\sigma, so(3)^*)$ by pointwise co-adjoint action. Diffeomophism $\psi$ acts naturally on all functions or one forms.
One can show that these are indeed momentum maps as they obey the criterion {\it equivariance}  (\ref{equivariance}) that is
\begin{equation}\label{equiv-prop:gauss}
G\circ \D(\psi^{-1})=\psi^*G,\quad
D''\circ \D(\psi^{-1})=\psi^*D,\quad
C\circ \D(\psi^{-1})=\psi^*C
\end{equation}
(due to fact \ref{spt_diffeo_AB}) and
\begin{equation}\label{equiv-prop:diff}
G\circ \G(O^{-1})=\operatorname{Ad}_O^*G,\quad
D''\circ \G(O^{-1})=D'',\quad
C\circ \G(O^{-1})=C
\end{equation}
This equivariance property will be important for the classification of the remaining constraints under symmetry restriction (see below).\\

Finally, one has still to impose equations (\ref{Diffeo_constraint})-(\ref{Scalar_constraint}) for all $x\in\sigma$ or one can equivalently study the vanishing of $C[N],\,\vec{D}[\vec{N}]$, i.e. their respective smearings against a lapse function $N$ and a shift vector $\vec{N}$ (of compact support).\\
Let us notice that the diffeomorphism and scalar constraints are $SO(3)$
 invariant thus they can be pulled down to the ADM phase space. After reducing the phase space by Gauss constraint the diffeomorphism constraints are equal to $\vec{D}_{ADM}$ thus they are generators of spacelike diffeomorphisms.\\

\subsubsection{Lifting $\Phi$-actions from $\M_{ADM}$ to $\M_{AB}$}
\label{s323}
In this section, we expand on how the action of (suitable) symplectomorphism groups $\D(\Psi)$ on $\M_{ADM}$ can be lifted to $\M_{AB}$, making the framework of {\it symmetry restriction} applicable here as well.\\

Let us introduce a symmetry group $\Psi$, which acts on the spatial manifold $\mathcal{\sigma}$ via diffeomorphisms, i.e. $\psi:\sigma \to \sigma$. We have seen that this group can be translated into a group of symplectomorphisms $\D(\Psi)$ on the ADM- phase space with (\ref{diff-action}):
\begin{align}
\D(\psi) (P^{ab}, q_{cd}):= (\psi_* ( P^{ab}),\psi_*(q_{ab}))
\end{align}
and that for every $\psi\in\Psi$ there exists plenty of symplectormorphisms on $\M_{AB} $ whose projections on $\M_{AB} /\!\!/G$ agree. One particular choice was $\D_{AB}(\Psi)$ with (\ref{Action_on_AE_var}), i.e.
\begin{align}
\D_{AB}(\psi)(A^I_a,E^a_I):=(\psi_*(A^I_a),\psi_* (E^a_I))
\end{align}
where $\psi_*$ has to be understood via (\ref{ActionOnGaussInv}). However, the later choice might be impractical for our purposes, as the following problem occurs: Not every group $\Psi$ that allows for a non-empty invariant subspace $\overline{\M}_{ADM}\subset \M_{ADM}$ under $\D(\Psi)$ allows for a non-empty invariant subspace $\overline{\M}\subset\M_{AB}$ under $\D_{AB}(\Psi)$.\\

\begin{example}
Let $\Psi_x$ be the group of rotations around $x\in\sigma=\mathbb{R}^3$, i.e. $SO(3)$, with $\Pi(\psi)^a_b:= \partial(\psi^{-1})^a/\partial x^b (y)$. Then there  exists a set of invariant $(q,P)$, e.g. $q_{ab}=a(t) \delta_{ab}$. However, no invariant $E_I$ exists, as there are no invariant vectors under rotations.
\end{example}

The remedy is to compensate rotation of the space indices by a $O(3)$ transformation of the internal indices. Although at first unintuitive, we will later see that only when choosing non-vanishing Gauss transformation we will be able to find the invariant subset $\bar {\mathcal{M}}_{AB} \subset \mathcal M_{AB}$ which agrees with the cosmological examples we have in mind.
In order to be in a situation where we can apply section 2 theorem 4, it is necessary to lift the action of $\Phi$ on $\M_{ADM}$ to an action on $\M_{AB}$. Let us explain what it means. We have the surjection of the groups
\begin{align}
G_{AB}\colon \D_{AB}(\Diff(\sigma))\ltimes \G\rightarrow \D(\Diff(\sigma))\\
G_{AB}(\G(O)\D_{AB}(\psi))=\D(\psi)
\end{align}

Now we want to find a group $\Phi_{AB}\subset \D_{AB}(\Diff(\sigma))\ltimes \G$ such that
\begin{equation}
G_{AB}(\Phi_{AB})=\D(\Psi)
\end{equation}
This already implies ${\bf P}^G(G^{-1}(0)\cap \overline{\M}_{AB})\subset \overline{\M}_{ADM}$. In the following, we ask whether symplectic reduction with respect to gauge $O(3)$ transformation commutes with symmetry restriction.\\

The construction of $\Phi_{AB}$ is done in the following way:
We are given a group of diffeomorphisms $\Psi$. First, determine the $\D(\Psi)$-invariant subspace of $\M_{ADM}$. Given some parametrisation ${q_o}_{ab}$ of the $\D(\Psi)$-invariant metric, we {\it choose} some arbitrary vector field $E^a_{oI}$ satisfying
\begin{equation}\label{find_Eo_from_qo}
{q_o}^{ab}=\frac{{E_o}_I^a {E_o}_J^b}{|\det(E_o)|} \delta_{IJ}
\end{equation}
Dependent on the choice $E_o$, for any $\psi\in\Psi$ the lift $\Phi_{AB,E_o}(\psi)$ to the Ashtekar-Barbero phase space will be chosen in form of an element $(\psi, O(\psi)) \in \D_{Ab}(\Diff(\sigma))\ltimes \G$ which acts on $\M_{AB}$ via
\begin{equation}
\G(O(\psi))\D_{AB}(\psi)
\end{equation}
such that:
\begin{equation}\label{ensure_inv_of_E}
([O(\psi)]^I_J \circ \psi_\ast \;{E_o}_I)^b (x)\overset{!}{=}{E_o}_J^b(x) \\
\end{equation}
where $O(\psi)$ is orthogonal from invariance of the metric. From non-degeneracy of $E^a_I$ and from the fact that $\psi$ is orientation-preserving, we know that $\det(O)=1$, i.e. $\det(\psi_* E) =\det(E)$. In other word, we have introduced the additional, non-trivial condition that the diffeomorphism $\psi$ lifted to $\M_{AB}$ also leaves the electric field ${E_o}_I^a$ invariant. The formula for matrix $O$ is given by
\begin{equation}\label{GaugeActionFromDiffeoForSymmetry}
[O(\psi)]^I_J(y)={E_o}_J^a(y)\psi_\ast({E_o}^I_a)(y)
\end{equation}
since $E^I_a E_I^b = \delta^b_a$.

So far, we have not considered the connection. In principle, demanding the analogue of (\ref{ensure_inv_of_E}) might impose additional restrictions. But we show in theorem \ref{theorem4} that there is no need to do this: the symmetry reduced phase space $\overline{\M}_{AB}$ reduces to $\overline{\M}_{ADM}$ by only considering $E_o$ for the construction of $O(\psi)$.

Thus, given that $\Psi$ is orientation preserving, such transformations form a group that will be denoted by 
\begin{equation}\label{final_sym_group_on_Ab}
\Phi_{AB,E_o}=\{\;(\,\D_{AB}(\psi)\,,\, {E_o}_J^a\psi_\ast({E_{o}}^I_a) \,)\;\colon \psi\in\Psi\}\subset \D_{AB}(\Diff(\sigma))\ltimes \G
\end{equation}
 Let us denote
\begin{equation}
\overline{\M}_{AB}=\{m\in \M_{AB}\colon \phi(m)=m,\ \forall \phi\in\Phi_{AB,E_o}\}
\end{equation}

We devote the remainder of this section to analyse $\M_{AB}$ for being a symplectic manifold and its relation to $\overline{\M}_{ADM}\subset \M_{ADM}$, the invariant submanifold under $\D(\Psi)$.\\

\begin{lm}
The space $\overline{\M}_{AB}$ is a manifold.
If $\Psi$ is compact and $\overline{\M}_{AB}$ is finite dimensional then $\Phi_{AB, E_o}$ is clean.
\end{lm}

\begin{proof}
We will first show that the space $\overline{\M}_{AB}$ is a manifold. Let us notice that the difference of two invariant connections is an invariant form valued in internal space
\begin{equation}
{A^1}^I_a-{A^o}^I_a\in \prescript{}{1}{C}^\infty(\sigma,\R^3)_\Phi
\end{equation}
The space of invariant forms is from assumption of the lemma finite dimensional. Similarly, $\prescript{1}{}{C}^\infty(\sigma,\R^3)_\Phi$ -- the space of invariant $E^a_I$ -- is finite dimensional. Thus we can parametrise $\overline{\M}_{AB}$ as follows
\begin{equation}
\overline{\M}_{AB}=\{(E^a_I, A^I_a)\colon E^a_I\in \prescript{1}{}{C}^\infty(\sigma,\R^3)_\Phi,\ \det E^a_I\not=0 ,\ A^I_a={A^o}^I_a+L^I_a,\ L^I_a\in \prescript{}{1}{C}^\infty (\sigma,\R^3)_\Phi\}
\end{equation}
This shows that $\overline{\M}_{AB}$ is a manifold.\\
This realisation together with compact $\Psi$ allows that we can now revoke fact \ref{Fact1} to show that $\overline{\omega}$ is clean. In fact the proof of the lemma needs to be adapted to the infinite dimensional setting.
\end{proof}

However, of course this does not tell whether $\overline{\M}_{AB}$ is empty. But even then it is a manifold.

\begin{lm}
We denote the intersection of the Gauss constraint locus and $\Phi_{AB,E_o}$-invariant subspace by $\overline{G^{-1}(0)}:= G^{-1}(0)\cap \overline{\M}_{AB}$.
The map ${\bf P}^G\colon \overline{G^{-1}(0)}\rightarrow \overline{\M}_{ADM}$ is surjective.\footnote{Attention: $\overline{G^{-1}(0)}$ is the constraint surface and not gauge-invariant space, thus the map is not injective.}
\end{lm}

\begin{proof}
Suppose that $(q_{ab},K_{ab})\in\overline{\M}_{ADM}$ then the claim is that there exists $E^I_a$ invariant under $\Phi_{AB,E_o}$ group (spatial diffeomorphisms and $SU(2)$ transformations) and such that
\begin{equation}
q_{ab}=E^I_a E^J_b \delta_{IJ} |\det(E)|
\end{equation}
which describe a point in $\overline{G^{-1}(0)}$:
\begin{equation}
(E^a_I,A_a^I=\Gamma^I_a+\beta K_{ab}E^b_J\delta^{IJ}\sqrt{|\det(E)|})\in\overline{\M}_{AB}
\end{equation}
due to the fact that $K_{ab}$ is symmetric. We will now find such a suitable $E^a_I$.\\
For the metric $q_{ab}$ we can introduce 
\begin{equation}
{q_o}^{ac}q_{cb}
\end{equation}
that is positive and symmetric in $q_o$ scalar product for every $x$. We can thus define $s^a_b$ as the positive square root  that depends smoothly on $\sigma$
\begin{equation}
s^a_cs^c_b={q_o}^{ac}q_{cb}
\end{equation}
and moreover $s_{ac}:=s^b_aq_{o,bc}$ is choosen positive and symmetric (which makes it unique). Contracting both sides with $q_{o,ad}$ gives:
\begin{align}\label{lemma6_stepinbetween}
q_{db}=s_{dc}s^c_b
\end{align}
Let us notice that 
\begin{equation}
E^I_a:=\det(s)^{-1}{E_o}^I_b s^b_a
\end{equation}
satisfies 
\begin{align}
q_{ab}\overset{(\ref{lemma6_stepinbetween})}{=}s_b^{b'}s_a^{a'}q_{o,a'b'}=E^I_a E^J_b \delta_{IJ} |\det(E)|
\end{align}
 thus $(E^a_I,A^I_a)$ is invariant under $\Phi_{AB,E_o}$. Thus ${\bf P}^G$ is surjective. 
\end{proof}

So although this identification is not injective: if two points give the same ADM data, they are related by a Gauss transformation:

\begin{lm}\label{lemma11}
Let us suppose that two points in $\overline{\M}_{AB}$ are related by Gauss transformation $\G(O)$ then $\G(O)$ commutes with $\Phi_{AB,E_o}$
\begin{equation}
\phi \G(O) \phi^{-1}=\G(O),\quad \forall \phi\in \Phi_{AB,E_o}
\end{equation}

\end{lm}

\begin{proof}
We assume that $\G(O)$ connects two points $(E_i,A_i)\in\overline{\M}$
\begin{equation}
\G(O)(E_1,A_1)=(E_2,A_2)
\end{equation}
then $O\in C^\infty(\sigma,SO(3))$ is uniquely determined by
\begin{equation}
O^I_J{E_1}^J_a={E_2}^I_a
\end{equation}
because using $|\det(E_1)|{E_1}^J_b{q_1}^{ba}{E_1}^I_a=\delta^{IJ}$ we get
\begin{equation}
O^I_J=\delta_{KJ}|\det(E_1)|{E_1}^K_b{q_1}^{ba}{E_2}^I_a
\end{equation}
Now for $\phi=\G(O(\psi))\D_{AB}(\psi)$ we have
\begin{equation}
\phi\G(O)\phi^{-1}=\G(O(\psi))\G(\psi_*O)\G(O(\psi)^{-1})=\G(O(\psi)\psi_*OO(\psi)^{-1})
\end{equation}
is also $O(3)$ gauge transformation. Moreover
\begin{equation}
\G(O(\psi)\psi_*OO(\psi)^{-1})(E_1,A_1)=(E_2,A_2)
\end{equation}
because $(E_i,A_i)\in\overline{\M}$. We see that
\begin{equation}
O(\psi)\psi_*OO(\psi)^{-1}=O
\end{equation}
and the transformation $\G(O)$ commutes with $\Phi_{AB,E_o}$.
\end{proof}

A consequence of this lemma is that no global identifications due to the Gauss group occur. In order to perform symplectic reduction we only need to know the remaining Gauss constraints and the sub group of gauge transformations that commutes with the symmetry group.

Let us introduce a group $\overline{\mathbb{G}}\subset \mathbb{G}$ of the $O(3)$ gauge transformations that commutes with $\Phi_{AB,E_o}$. It acts on $\overline{\M}_{AB}$ by symplectic transformations. As usually not all of them can be generated by Hamiltonians due to noncompact support. Let us now assume that $\Phi_{AB,E_o}$ is compact and consider only compactly supported invariant gauge transformations $\overline{\mathbb{G}}_0$. In the setting of compact Cauchy surfaces, every $O(3)$ gauge transformation is compactly supported, i.e. $\mathbb G_0 = \mathbb G$. The generators of $\overline{\mathbb{G}}_0$ are in one to one correspondence with $\Phi_{AB,E_o}$-invariant, $so(3)$-valued functions with compact support.

Let us now consider a Gauss constraint as a function
\begin{equation}
G\colon \M_{AB}\rightarrow C^\infty(\sigma,so(3))
\end{equation}
Restricted to $\overline{\M}_{AB}$ 
\begin{equation}
\overline{G}\colon \overline{\M}_{AB}\rightarrow C^\infty(\sigma,so(3))_\Phi
\end{equation}
Let us notice that $C^\infty(\sigma,so(3))_\Phi$ plays a role of dual to the Lie algebra of the group $\overline{\mathbb{G}}_0$.
Invariant $so(3)$-valued function with compact support integrated with the Gauss constraint are generators of $\Phi_{AB,E_o}$-invariant gauge transformation with compact support.

\begin{tm}\label{theorem4}
Let $\Phi$ be compact. The phase space $\overline{\M}_{ADM}$ is a symplectic reduction with respect to group $\overline{\mathbb{G}}$ and momentum map $\overline{G}$ of the phase space $\overline{\M}_{AB}$.
\end{tm}

\begin{proof}
Follows from  previous lemmas.
\end{proof}

\subsection{Classification of residuals}

We can now analyse fate of Einstein equations
\begin{equation}
C_{\overline{\M}}\in C^\infty(\sigma)_\Psi,\quad D_{\overline{\M}}\in {}^1C^\infty(\sigma)_\Psi
\end{equation}
(where $\M$ is either $\M_{ADM}$ or $\M_{AB}$).\\
Let us turn to the Gauss constraint. Let us notice that in $\Psi\subset \Diff^+(\sigma)$ then
\begin{equation}
\frac{1}{\sqrt{|\det(E_o)}|}\delta_{IJ}{E_o}^J_a\overline{G}^I{q_o}^{ab}\in \prescript{1}{}{C}^\infty(\sigma)_\Psi
\end{equation}
thus the Gauss constraint can be analyse in parallel to the diffeomorphism constraint.

Similarly, if $O\in C^\infty(\sigma, O(3))_{\Phi}$ then
\begin{equation}
|\det( E_o)|O_{IJ}{E_o}^I_a  {E_o}^J_b
\end{equation}
is an $\Psi$-invariant bi-covector. Thus we can restrict our attention to the action of diffeomorphism group $\Psi$.

The following fact can be used in symmetry restriction, because it will help to determine   the $\Phi$-invariant submanifold $\overline{\M}$ for $\Phi=\mathbb{D}(\Psi)$. 
We will classify remaining constraints and gauge transformation under some assumption on the group $\Psi$. The assumptions are as follows. 
\begin{enumerate}
\item We can introduce a finite dimensional Lie group structure on $\Psi$ and the action is smooth
\item The group acts transitively on $\sigma$:
\begin{equation}
\forall_{x,x'\in\sigma}\exists_{\psi\in\Psi} \psi(x)=x'
\end{equation}
\end{enumerate}
These assumptions will be satisfied in our examples.

Let us consider $\Psi$-invariant tensor field 
\begin{equation}
T\in \prescript{n}{m}{C}^\infty(\sigma)_\Psi=\{T\in \prescript{n}{m}{C}^\infty(\sigma)\colon \forall_{\psi\in \Psi}\psi_*T=T\}
\end{equation}
One notices that $\phi_* T(x)$ depends only on the coefficients of tensor in point $\phi^{-1}(x)$. We denote this linear map by
\begin{equation}
\phi_{*x}\colon \prescript{n}{m}{T}_{\phi^{-1}(x)}(\sigma)\rightarrow \prescript{n}{m}{T}_{x}(\sigma) 
\end{equation}
where $\prescript{n}{m}{T}_p(\sigma)$ is a vector space of $m$ vector $n$-covector tensors at point $p\in\sigma$.
Let us now choose $x_0\in \sigma$ and consider $T(x_0)\in \prescript{n}{m}{T}_{x_0}(\sigma)$. We define the stabiliser group
\begin{equation}\label{stabiliser_def}
\Psi_0=\{\psi\in \Psi\colon \psi(x_0)=x_0\}
\end{equation}
The stabilizing group acts on $\prescript{n}{m}{T}_{x_0}(\sigma)$ via $\psi_{*x_0}$. The tensor $T(x_0)$ is invariant under this action. 
\begin{equation}
T(x_0)\in \prescript{n}{m}{T}_{x_0}(\sigma)_{\Psi_0}=\{T_0\in \prescript{n}{m}{T}_{x_0}(\sigma)\colon \forall_{\psi\in \Psi_0} \psi_{*x_0}T_0=T_0\}
\end{equation}
In fact the following holds:

\begin{fact}\label{fact:6_transitivity}
There is a bijection
\begin{equation}
R_{x_0}\colon \prescript{n}{m}{C}^\infty(\sigma)_\Psi \ni T\rightarrow T(x_0)\in \prescript{n}{m}{T}_{x_0}(\sigma)_{\Psi_0}
\end{equation}
\end{fact}

\begin{proof}
The map is well-defined. We need to show that it is bijective. We will define its inverse map.
For an invariant tensor $T_0\in \prescript{n}{m}{T}_{x_0}(\sigma)_{\Psi_0}$, we introduce a tensor field defined by
\begin{equation}
T(\psi(x_0))=\psi_{*x_0}T_0,\quad \psi\in \Psi
\end{equation}
It is well-defined because if $\psi(x_0)={\psi'}(x_0)$ then $\psi_0=\psi^{-1}\psi'\in \Psi_0$
\begin{equation}
T(\psi'(x_0))=\psi'_{*x_0}T_0=\psi_{*x_0}{\psi_0}_{*x_0}T_0=\psi_{*x_0}T_0=T(\psi(x_0))
\end{equation}
Tangent space $T_{x}\sigma$ can be identify with a subspace of the tangent space of $\Psi$ at $id$ as it acts smoothly.  Thus every curve in $\gamma\in\sigma$ can be lifted to a curve $\psi(t)$ in $\Psi$ such 
\begin{equation}
\gamma(t)=\psi(t)(x_0)
\end{equation}
and the tensor field is smooth.
Summarizing, we defined a map
\begin{equation}
Q_{x_0}\colon \prescript{n}{m}{T}_{x_0}(\sigma)_{\Psi_0}\rightarrow \prescript{n}{m}{C}^\infty(\sigma)
\end{equation}
We can check that in fact $Q_{x_0}(T_0)\in \prescript{n}{m}{C}^\infty(\sigma)_\Psi$ and
\begin{equation}
R_{x_0}Q_{x_0}={\mathbb I}
\end{equation}
thus $R_0$ is surjective.

However, $R_{x_0}$ is also injective: If $T(x_0)=0$ then as for every $x\in\sigma$ there exists $\psi\in \Psi$ such that $\psi(x_0)=x$
\begin{equation}
T(x)=(\psi_*T)(x)=\psi_{*x}T(x_0)=0.
\end{equation}
So finally we proved bijectivity.
\end{proof}

Let us notice that scalars are always preserved by $\Psi_0$ thus in the transitive case there exists exactly one dimensional space of invariant scalar spanned by function equal to $1$ i.e. constant functions.

Let us now consider $O(3)$ gauge transformations. We notice the following
\begin{align}
&\left(\G(O(\psi))\D_{AB}(\psi)\right)\G(O)\left(\G(O(\psi))\D_{AB}(\psi)\right)^{-1}=\nonumber\\&=\G(O(\psi))\G(\psi^*O)\G(O(\psi))^{-1}=\G(O(\psi)(\psi^*O)O(\psi)^{-1})
\end{align}
Thus the $\Phi$-invariant $O(3)$ transformations satisfy
\begin{equation}
O=O(\psi)(\psi^*O)O(\psi)^{-1},\quad \psi\in \Psi
\end{equation}
As before we introduce
\begin{equation}
O(3)_{\Psi_0}=\{O\in O(3)\colon \forall_{\psi\in \Psi_0}O=O(\psi)OO(\psi)^{-1}\}
\end{equation}
By the same method as before we obtain:

\begin{fact}	
There is a bijection
\begin{equation}
H_0\colon C^\infty(\sigma,O(3))_\Phi\ni O\rightarrow O(x_0)\in O(3)_{\Psi_0}
\end{equation}
\end{fact}
Thus, $O(3)_{\Psi_0}$ includes all remaining $O(3)$-gauge transformations on $\overline{\M}$. There are no additional points of $\overline{\M}$ globally related by $O(3)$ gauge transformations  (which can be seen by application of lemma \ref{lemma11}).\\

Finally, we have only partial result about residual diffeomorphisms. Let us suppose that $\phi\in \Diff(\sigma)$ such that
\begin{equation}
\forall_{\psi\in \Psi}\phi^{-1}\psi\phi\in \Psi
\end{equation}
then
\begin{equation}
\forall_{x\in \overline{\M}_{ADM}}\phi(x)\in \overline{\M}_{ADM}
\end{equation}
In fact, if we can lift $\phi$ in the same way as we lifted $\Psi$ to Ashtekar-Barbero phase-space then we obtain similar result for this lifted diffeomorphism. 
\section{Examples in classical cosmology}
\label{section:4}

In the previous section, we derived a theorem stating that the flow of a $\Phi$-invariant function $F$ on the $\Phi$-invariant submanifold of some phase space agrees with the flow of the reduced function $\bar{F}$ on the reduced phase space with respect to $\Phi$. This can in general allow to simplify computations, as no longer needs to deal with the full phase space but with a much simpler setup where symmetries have been taken care of.\\
We will now take advantage of this theorem by applying it to the connection formulation of canonical general relativity in terms of Ashtekar-Barbero variables and some cosmological subsectors of it. By cosmological subsectors we mean those spacetimes that feature a high degree of symmetry and consider in this article explicitly the following families: (1) isotropic, (flat) spacetimes of Robertson-Walker type \cite{Fri:22,Lem:33,Rob:35,Wal:37}, (2) homogeneous spacetimes of Bianchi I type \cite{Bianchi:1898} and (3) spacetimes of the isotropic, open $k=-1$ spacetime. As we will see, all of these families (parametrised by only finitely many degrees of freedom) can be described by some symmetry groups. Symmetry restriction of the canonical version of general relativity with respect to these groups allows then to restrict the action of $\Phi$-invariant constraints (which drive the dynamics) to the phase space spanned by the parameters of the cosmological family in question. We will study this for non-compact as well as toroidal topologies.\\
After this warm-up exercise, we will in the next section turn towards a discretisation of the theory (having applications to quantum gravity in mind as we will outline in the conclusion).

There are essentially three levels of difficulty here. In the case of gravity on graph we are in finite dimensional situation we considered so far. Although we did not consider infinite dimensional situation at the case at hand situation is similar if $\Phi$ is clean and compact as in case of cosmology with compact spacelike sections.  In these situations our results as described below apply completely. If the spacelike sections are not compact then it turns out that $\Phi$ is not clean (even the restriction of symplectic form is ill-defined).

We will consider first examples with non-compact groups. In these cases the group action is not clean and even the restriction of the symplectic form is not well-defined, but
some of the examples can be corrected and in this improved form they are used extensively in mini-super-space models.

		\subsection{Examples in non-compact cosmology}
		\label{section:5_non_compact}

We will now consider warm-up examples from classical cosmology. They cannot be completed to full satisfaction but present easy examples of the overall strategy. Thus, we present them in increasing order of difficulty:\\
First the flat Friedmann-Lema\^itre-Robertson-Walker case, describing a isotropic and homogenous universe. Due to the high degree of symmetry present, it is easy to find a group  $\Phi$ such that symmetry restriction with respect to it, gives a submanifold parametrised only by two real numbers, as is often used in classical cosmology. Second, we look at a generalisation, the Bianchi I model, where the requirement of isotropy is dropped. Again, we find a symmetry group suitable to reduce to the invariant submanifold. Third, we expand also to the non-flat case, i.e. the open, hyperbolic FLRW universe.\\
A common feature of all these non-compact examples is that there is symplectic structure on each invariant submanifold descending from the respective full phase space. Therefore, it it impossible to make use of tools to simplify computations such as the ``restriction theorem for dynamics'' from section \ref{section:2_Symmetry_Red}.

\subsubsection{Flat FRWL metric}

The spatial manifold is chosen to be $\sigma=\R^3$ on which a following symmetry group $\Psi\subset\Diff(\sigma)$ acts. We choose it to be isomorphic to $\Psi \cong SO(3) \rtimes \R^3$ and acting by
\begin{itemize}
\item three spatial translations, generated by spacelike killing vector field $\vec\xi_i =\partial_i$ obeying the Lie algebra of $\R^3$, i.e. $[\vec{\xi}_i,\vec{\xi}_j]=0$
\item three spatial rotations, generated by spacelike killing vector field $\vec{\xi}_i'=\epsilon_{ijk}x^j\partial_k$ obeying the Lie algebra of $SO(3)$, i.e. $[\vec \xi_i', \vec \xi_j'] = -\epsilon_{ijk} \vec \xi_k'$
\end{itemize}  
In general, we would like to lift $\Psi$ to an action on $\M_{ADM}$. However, we point out that one can not write $\Phi=\D(\Psi)$ in the form of $\vec D[\vec N]$ with some compactly supported vector fields $\vec{N}$, as the elements in $\Psi$ act non-trivially everywhere on $\sigma$. Luckily, this is not necessary after all:\\
Let us notice that $x=0\in\sigma$ has the stabiliser from $\Psi_0=SO(3)$, see (\ref{stabiliser_def}), and $\Psi$ acts transitively. Hence, we can make use of fact \ref{fact:6_transitivity} to immediately determine the possible spaces of invariant vectors and invariant symmetric bivectors:
\begin{equation}
T_0(\sigma)_{\Psi_0}=\{0\}, \ \prescript{(2)}{}{T}_0(\sigma)_{\Psi_0}=\{p\,\delta_{ab}\;:\; p\in\mathbb{R}\}
\end{equation}
There is a one dimensional space of invariant scalars and symmetric bivectors and no invariant vectors (or pseudovectors) under $SO(3)$. Thus 
\begin{equation}
C^\infty(\sigma)_\Psi=\Span\{1\},\quad
{}^1 C^\infty(\sigma)_\Psi=\{0\},\quad
{}^{(2)}C^\infty(\sigma)_\Psi=\Span\{{q_o}^{ab}\},
\end{equation}
where ${q_o}^{ab}=\delta^{ab}$. This gives us a full characterisation of the $\D(\Psi)$-invariant subspace:
\begin{align}\label{invariant_phasespace_cosmology}
\overline{\M}_{ADM}&=\{(q_{ab}, P^{ab})\; \colon \D(\Psi) q_{ab}= q_{ab},\; \D(\Psi) P^{ab}=P^{ab}\}\\
&= \{(\tilde{p}\,{q_o}_{ab},\alpha\,{q_o}^{ab})\colon \tilde{p}>0,\alpha\in \R\}\nonumber
\end{align}
From the discussion of (\ref{remaining_constraints_ADM}) this also implies that no non-vanishing vector constraints remain on $\overline{\M}_{ADM}$. And only one $\sigma$-independent scalar constraints remains to be imposed, namely:
\begin{align}
C|_{\overline{\M}_{ADM}}(x)= -\frac{3\kappa}{2}\sqrt{p}\,\alpha^2+C_{\rm matter}|_{\overline{\M}_{ADM}}
\end{align}
\\
We investigate the situation for the same symmetry group on the connection phase space. As we already know the parametrisation of $q_{o,ab}= \tilde{p}\delta_{ab}$, we choose
\begin{align}\label{fiducial_triad_FRW}
E_{oI}^a = \delta^a_I
\end{align}
in line with (\ref{find_Eo_from_qo}). Thus, we have a singled out representation $\Phi_{AB, E_o}$ on $\M_{AB}$ which upon symplectic reduction of Gauss constraint $G$ will reduce to the action of $\D(\Psi)$ on $\M_{ADM}$, namely (\ref{final_sym_group_on_Ab}):
\begin{align}
\Phi_{AB,E_o}=\{(\D_{AB}(\psi),(\frac{\partial \psi}{\partial x})^J_I)\; \colon\; \psi \in \Psi\}
\end{align}
where the derivative does not depend on the point.
One can check that the only invariant connection and densitised triad are (using that $\Gamma(e_o)=0$)
\begin{align}
\overline{\M}_{AB}=\{(E_{I}^a=p \,\delta^a_I, A_{a}^I=c\, \delta^I_a)\; \colon\; p,c\in\mathbb{R}\}
\end{align}
In order to identify $\M_{AB}$ with $\M_{ADM}$, one notes that the Gauss constraint is trivially satisfied. However, there still exists a nontrivial group of $O(3)$ gauge transformations, because
\begin{equation}
O(3)_{\Psi_0}=\{{\mathbb I},-{\mathbb I}\}
\end{equation}
has remaining nontrivial transformation $O(-{\mathbb I})$.
Taking into account that $\Gamma^I=0$ we derive
\begin{equation}
O(-{\mathbb I})(c,p)\rightarrow (-c,-p)
\end{equation}
and the phase space agrees with (\ref{invariant_phasespace_cosmology}) up to this identification.
We obtain directly the map
\begin{equation}
\tilde{p}=|p|,\quad \alpha=-\frac{2c}{\kappa\beta}\operatorname{sgn} p
\end{equation}
And similar as before, we know that also the diffeomorphism constraints are already satisfied on $\overline{\M}_{AB}$ and there is only one scalar constraint left for any point $x\in\sigma$:
\begin{align}
C|_{\overline{\M}_{AB}}(x)=-\frac{6}{\kappa \beta^2}c^2 \sqrt{|p|}+ C_{\rm matter}|_{\overline{\M}_{ADM}}
\end{align}
that is the pull-back of the scalar constraint from $\overline{\M}_{ADM}$.

\begin{remark}
It is important to emphasise that the restriction of the symplectic form is ill-defined, irrespective of using $\omega$ from (\ref{symplform-ADM}) or (\ref{symplform}). To see this, we show that for $\M_{AB}$ any $\Phi=\D(\Psi)$-invariant vector field $X_{\rm inv} $ there exists {\it no} $\Phi$-invariant, finite generator $F$, i.e.
with $\omega(X_{\rm inv} ,\,.\,) \neq dF$. First, note that the basis of $\Phi$-invariant vector fields is $\frac{\delta}{\delta A^I_a(x)},\, \frac{\delta}{\delta E^a_I(x)}$. Thus
\begin{align}
\omega(\frac{\delta}{\delta A^I_a(x)},\,.\,) = \frac{2}{\kappa\beta} \int_\sigma {\rm d}^3x\; \delta^J_b dE^b_J(x)\overset{!}{=}d F \hspace{10pt}\Rightarrow\hspace{10pt} F\equiv  \frac{2}{\kappa\beta}\int_\sigma {\rm d}^3x \delta^J_b E^J_b(x)
\end{align}
But now, we need to ensure that $F$ itself is $\Phi$-invariant as well, i.e. $E^J_b(x)=p \delta^J_b$. Thus
\begin{align}
F =\frac{6p}{\kappa\beta} \int_{\mathbb{R}^3} {\rm d}^3x \; 1 = \infty
\end{align}
Similar, the divergent integral over the spatial manifold appears in the ADM case as well. We see that no finite invariant generator for invariant vector fields exists.\\
Thus, we cannot restrict the symplectic form to neither $\overline{\M}_{ADM}$ nor $\overline{\M}_{AB}$, and statements such as restriction of the dynamics cannot be applied here. The only way to remedy this is by working on a compact slice $\sigma$ as we will do in subsection \ref{section:5_compact}.

Finally, the symplectic form usually used in context of LQC is not a restriction of the symplectic form of the full theory \cite{ABL:03}. We expect that this issue is related to problems with the infra-red cut-off in the cosmological perturbation theory \cite{ST:19,NST:19}.
\end{remark}
\subsubsection{Bianchi I}

Let us first analyse Bianchi I. Naively, we would like to consider the group $\R^3$. However, this symmetry does not reduce the metric to its diagonal form. With the group $\R^3$ we are left with $3$ vector constraints and one scalar constraint. The three vector constraints can be used to transform the metric to the diagonal form (and by solving these constraints we get rid of nondiagonal momenta). Such a procedure is not optimal through.

We can assume the symmetry of the bigger group
\begin{equation}\label{BianchiIsymmtery}
\Psi'\cong H\rtimes \mathbb R^3
\end{equation}
where the group $H\subset SO(3)$ is a $4$ element group generated by rotations by $\pi$ around coordinate axes
\begin{equation}\label{group_H}
1,\; \ R_{e_x}(\pi),\ R_{e_y}(\pi), \ R_{e_z}(\pi)=R_{e_x}(\pi)R_{e_y}(\pi)
\end{equation}

\begin{lm}
We recall that for $\Psi_{0}=H$ we have
\begin{align}
T_0(\sigma)_{H}={}^1T_0(\sigma)_H=\{0\},\hspace{10pt}
{}_{(2)}T_0(\sigma)_{H}=\Span\{\delta^1_{a}\delta^1_{b},\delta^2_{a}\delta^2_{b},\delta^3_{a}\delta^3_{b}\}
\end{align}
\end{lm}

\begin{proof}
As the group is orientation preserving, vectors and pseudovectors transform in the same way. Let us concentrate on vectors. Consider an invariant vector $v^a$. Rotating it around the $z$ axis we get from invariance
\begin{equation}
v^x=-v^x
\end{equation}
thus $v^x=0$ and similarly other components.

Let us consider an invariant symmetric bi-covector $q_{ab}$ and apply the same rotation:
\begin{equation}\label{qe-q}
q_{xz}=-q_{xz}
\end{equation}
thus only diagonal coefficients might be nonzero. On the other hand diagonal bi-covectors are invariant.
\end{proof}

Then $\overline{\M}_{ADM}$ is characterised analogously to before in (\ref{invariant_phasespace_cosmology}).\footnote{Let us notice that the $4$ metric is of the form
\begin{equation}
-b(t)^2+\sum_{i=1}^3 a_i(t)^2(dx^i)^2
\end{equation}
} Also, from transitivity of the action
\begin{equation}
C^\infty(\sigma)_\Psi=\Span\{1\},\quad
{}^1C^\infty(\sigma)_\Psi=\{0\},\quad
{}_{(2)}C^\infty(\sigma)_\Psi=\Span\{\delta^1_{a}\delta^1_{b},\delta^2_{a}\delta^2_{b},\delta^3_{a}\delta^3_{b}\}
\end{equation}
We are left with one scalar constraint. We will show how to compute it later.\\

To extend our symmetry restriction to the phase space $\mathcal{M}_{AB}$ of Ashtekar-Barbero variables, we consider the same lift of the action as in FRWL case. Namely ${E_o}^I_a=\delta^I_a$ and ${q_o}_{ab}=\delta_{ab}$.

\begin{fact}
The sets of $(E,A)$ invariant under the action of $\Phi'=\D_{AB}(\Psi)$ from \eqref{BianchiIsymmtery} consists of the connections and densitised triads of the form
\begin{equation}
A^I_a=V_o^{-\frac{1}{3}} c_a\delta^I_a,\quad E_I^a=V_o^{-\frac{2}{3}}p_a\delta_I^a
\end{equation}
where $c_a,p_a\in\mathbb R$ and $V_o$ is an arbitrary parameter of the volume units.
\end{fact}
\begin{proof}
First, due to translational invariance, we can restrict our attention to one point. Then, from observation (\ref{qe-q}) we choose $E_{oI}^a= \delta^a_I$ satisfying our requirement. Then we can determine the form of $O(R_{e_k})(\pi)$ via (\ref{GaugeActionFromDiffeoForSymmetry}) namely
\begin{align}
O(R_{e_x}(\pi))={\rm diag}(1,-1,-1),\quad O(R_{e_y}(\pi))={\rm diag}(-1,1-1),\quad {\rm etc.}
\end{align}
Thus, we have the full action of $\Phi_{AB,E_o}(\Psi')$. Invariance of the connection with respect to it means
\begin{align}
A^I_a \overset{!}{=} [R_{e_k}(\pi)]_a^b A^J_b [O(R_{e_k}(\pi))]^I_J \quad \Rightarrow \quad A^I_a =\delta^I_a c_a
\end{align}
for some function $c_a$. Similar for the triad $E$.
\end{proof}

\begin{remark}
This fact is a special case of the general theorem about the classification of invariant connections \cite{Nomizu:54,Wang:58} as disussed in \cite{KN:96}.
\end{remark}

A simple computation shows now, that the only remaining $O(3)$ gauge transformations are
\begin{equation}\label{BinachiI_remainingGauss}
\G(\operatorname{diag}(\pm 1,\pm 1,\pm 1))
\end{equation}
Thus, we have parametrised the invariant submanifold $\overline{\M}_{AB}$ for group $\Phi'$. Due to the same argumentation as before, one sees that no diffeomorphism constraint remain to be implemented and only one scalar constraint:
\begin{align}\label{scalar_constraint_BinachiI}
C|_{\overline{\M}_{AB}}(x)=-\frac{6}{\kappa\beta^2}\left(\sum_a c_a^2 \sqrt{|p_a|}\right)+C_{\rm matter}|_{\overline{\M}_{AB}}
\end{align}
However, once again reducing the symplectic form is ill-defined and thus not all the tools from symmetry restriction can be applied.

\subsubsection{The open (k=-1) FRWL universe}

Let us consider $O_+(1,3)$ group that is a subgroup of the orthogonal subgroup in the Lorentz group
\begin{equation}
O(1,3)=\{L\in GL(4)\colon L^T \eta L =\eta\}
\end{equation}
that preserves time orientation, that is
\begin{equation}
L^0_0>0
\end{equation}
The connected component $SO_+(1,3)$ is distinguished by additional condition $\det L=1$.

Both groups $O_+(1,3)$ and $SO_+(1,3)$ preserves future part of the unit hyperboloid
\begin{equation}
{\mathcal H}_+=\{x\in\R^4\colon x^T\eta x=1,\ x^0>0\}
\end{equation}
and moreover they acts transitively on ${\mathcal H}_+$. As the orthogonal group preserves $\eta$ the action preserves also metric restricted to ${\mathcal H}_+$,
\begin{equation}
q_o=\eta|_{{\mathcal H}_+}
\end{equation}
This is the model of a $k=-1$ spatial section.

Let us notice that for a given point $0=(1,0,0,0)$ of the future hyperboloid we have the stabiliser group
\begin{equation}
\Psi_0=O(3) \text{ or } SO(3)
\end{equation}
and thus there exists only one (up to a scaling) invariant bi-covector ${q_o}_{ab}$ and it is symmetric. It is just a restriction of the Minkowski metric to the hyperboloid.
\begin{equation}
\overline{\M}_{ADM}=\{(q_{ab}=g{q_o}_{ab},K_{ab}=c{q_o}_{ab})\colon g>0,c\in \R\}
\end{equation}
In this case $P^{ab}=-\frac{2}{\kappa}cg^{3/2}\sqrt{q_o} {q_o}^{ab}$.\\
Again as the $\Psi_0\equiv SO(3)$ and $\Psi$ acts transitively there is only one scalar constraint left. We can compute its value in any point
\begin{align}\label{Constraint_BinachiI_continuum}
C|_{\overline{\M}_{ADM}} = - \frac{6\sqrt{p}}{\kappa}(c^2-1)+ C_{\rm matter}|_{\overline{\M}_{ADM}}
\end{align}

\subsection{Examples in compact cosmology}
\label{section:5_compact}

We will now consider what can be done in the case when we work with the torus instead of $\R^3$ as a spacial section. The examples we consider are: Bianchi I (and its special case: isotropic FLRW), Binachi IX and the closed, isotropic universe with positive curvature. In all these cases the symmetry restriction can be fully executed, allowing to reduce all computations (including the reduction of dynamics) to the invariant submanifold.\\
We have seen that the canonical formulation of cosmological models with noncompact sections is problematic (for example naive pull-back of the symplectic form is ill-defined). Moreover, we are ultimately interested in a quantum theory of gravity. In this case, noncompact slices are usually even more problematic due to some infrared issues. For this reason one often consider cosmology with torus spacelike section. Note that the constraints equations are local thus the situation does not differ from either $k=0$ FRW cosmology or Bianchi I model.

\subsubsection{Bianchi I model with torus sections}
Let the spatial manifold be compact and of toroidal form, ie.e $\sigma= [0,T_1]\times[0,T_2]\times[0,T_3]$.
The group
\begin{equation}
 \R^3\rtimes H
\end{equation}
preserves the relation
\begin{equation}
x\approx x+\sum n_i T_i\hat{e}_i
\end{equation}
where $T_i$ are periods of the torus (Note that for Bianchi I they might be unequal). We can thus consider the action on the equivalence relation space. Taking into account the kernel of the action we get the group
\begin{equation}\label{BianchiI_symmetry_torus}
\Psi_I= (\prod_i \R /T_i)\rtimes H
\end{equation}
To obtain the same phase space $\overline{\M}_{ADM}$ as before for non-compact Bianchi I, we use again fact \ref{fact:6_transitivity} for
\begin{equation}
C^\infty(\sigma)_\Psi=\Span\{1\},\quad
{}^1C^\infty(\sigma)_\Psi=\{0\},\quad
{}_{(2)}C^\infty(\sigma)_\Psi=\Span\{\delta^1_{a}\delta^1_{b},\delta^2_{a}\delta^2_{b},\delta^3_{a}\delta^3_{b}\}
\end{equation}
Again choosing the same $E_o$ we get:

\begin{fact}\label{fact:12}
The sets of $(A,E)$ invariant under the action of $\Phi'$ from \eqref{BianchiIsymmtery} consists of the connections and electric fields of the form
\begin{equation}\label{overlineMred}
\overline{\M}_I := \{ (E^a_I, A^I_a)\colon \;\; E_I^a=V_o^{-1}T_ap_a\delta_I^a,\ A^I_a=T_a^{-1} c_a\delta^I_a \}
\end{equation}
where $c_a,p_a\in\mathbb R$ and $V_o:=\int_\sigma \sqrt{q_o}=\int_\sigma {\rm d}^3x$. If we parametrise this set by $c_a,p_a$ (manifold structure) then the pull-back of $\omega$ from (\ref{symplform}) is equal to 
\begin{equation}\label{fact9:reducedsymplform}
\omega|_{\overline{\M}_{I}}=\frac{2}{\kappa \beta} \sum_a {\rm d}p_a\wedge {\rm d}c_a
\end{equation}
\end{fact}

\begin{proof}
We will determine the symplectic form. Let us work with the symplectic potential
\begin{equation}
\xi_{AB}= \frac{2}{\kappa\beta}\int_\sigma {\rm d}^3x\; E^{a}_I(x)\; d A^I_a(x)
\end{equation}
Let us notice that
\begin{equation}
d A^J_a(x)=T_a^{-1}\delta^J_a\;	d c_a
\end{equation}
thus as $E^{a}_I=V_o^{-1}T_a\delta_J^a\;p_a$
\begin{equation}
\frac{2}{\kappa\beta}\int_\sigma E^{a}_I \, d A^I_a|_{\overline{\M}_{I}}=\frac{2}{\kappa\beta} \sum_a p_a\ dc_a
\end{equation}
Thus (\ref{fact9:reducedsymplform}) follows.
\end{proof}

There are no invariant vector fields or fields with one internal index thus we are left with the usual O(3) transformations (\ref{BinachiI_remainingGauss}) and with one scalar constraint, namely (\ref{scalar_constraint_BinachiI}). And as the reduction of the symplectic form is well-defined, the whole tool-box of symmetry restriction can be applied. This means, that the dynamics of the symmetric system is fully encoded in the flow of $C|_{\overline{\M}_{I}}$ with respect to (\ref{fact9:reducedsymplform}).

\begin{remark}
The $p_a$ and $c_a$ variables are more geometric objects then value of connection and densitised triad. They are given as integrals of $E_I^a$ over basic two tori and $A^I_a$ over basic cycles in $\sigma$, so they are independent of the choice of $T_a$. The scaling diffeomorphism transforming a torus with periods $T_a'$ into torus with periods $T_a$ would not affect these variables: e.g. it scales simultaneously the lengthof the basic circle and inversely the connection. Therefore, such scaling diffeomorphisms have been completely taken care of on the reduced phase space.\\
Let us also point out some subtleties regarding the interpretation of the parameter $V_o$ and $T_a$. We introduced it merely as fiducial element in our coordinate system on $\overline{\M}_{I}$.  The symplectic form itself is defined geometrically and is independent of this fiducial structure. 
It is in contrast to a noncompact case where there is no such choice of variables and no built-in scale. In this context, dilation and implementing invariance under the latter have been studied in the literature \cite{EHT:16,EV:19}.
\end{remark}

\subsubsection{$k=0$ FRWL with torus sections}
\label{s422}

In this case our method fails because now there is no $SO(3)$ symmetry available.

The symmetry group is broken: only discrete subgroup of rotations are allowed. In general if the periods of the torus are nonequal the only rotations that descend to the action on the torus is the group $H$. However, if all the periods are equal (we can normalise them to $2\pi$) then we are left with the residual symmetry
\begin{equation}
SO(3,{\mathbb Z})
\end{equation}
In addition to $H$ group it also contains the missing 90 degree rotations around three axis.
Let us consider
\begin{equation}\label{group_torus}
\Psi_{FRW}=(\mathbb{R}/T)^3\rtimes SO(3,{\mathbb Z})
\end{equation}
The part of $(\R/T)^3$ group is generated by vector fields
\begin{equation}
\vec \xi_1 = \partial_x, \ \ \ \ \ \vec \xi_2 = \partial_y, \ \ \ \ \ \vec \xi_3 = \partial_z
\end{equation}
Let us notice that $H\subset SO(3,{\mathbb Z})$ and it is a normal subgroup. We can first reduce the case to Bianchi I and then assume additional symmetry $\Psi_{FRW}/\Psi_I$. Let us notice that
\begin{equation}
\Psi_{FRW}/\Psi_I=S_3
\end{equation}
permutation of axes. We will now describe how this group acts on the Bianchi I reduced phase space.

\begin{fact}
Action of the group $\Phi_{FRW}=\D(\Psi_{FRW})$ preserve $\overline{\M}_{I}$ from (\ref{overlineMred}) and its action factorises by $S_3$. It is given by permutations on the indices of $c_a$ and $p_a$ respectively.
\end{fact}

\begin{proof}
It is enough to consider what the action does on $A$ and $E$ of torus Bianchi I form.
\end{proof}

The only 4-metric on $(\R/T)^3\times{\mathbb R}$ invariant under the group $\Phi_{FRW}$ is the metric
\begin{equation}
b(t)^2dt^2-a(t)^2\sum_i {dx^i}^2
\end{equation}
which differs from the usual FRW metric by the presence of a lapse function $b(t)^2$ that indicates existence of a scalar constraint.

\begin{fact}
The sets of $(A,E)$ invariant under the action of \eqref{group_torus} consists of the the connections and $E$ fields of the form
\begin{equation}
A^I_j=V_o^{-\frac{1}{3}} c\delta^I_j,\quad E_I^j=V_o^{-\frac{2}{3}}p\delta_I^j
\end{equation}
where $c,p\in\mathbb R$ with $V_o$ as before.
If we parametrise this set by $c,p$ (manifold structure $\overline{\mathcal{M}}_{AB}$) then the pull-back of $\omega$ is equal to
\begin{equation}
\omega|_{\overline{\M}_{AB}}=\frac{6}{\kappa\beta} dp\wedge dc
\end{equation}
\end{fact}

\begin{proof}
Reduction of the Bianchi I case $p_a=p$ and $c_a=c$. Let us notice that $\sum_a dp_a\wedge dc_a=3dp\wedge dc$.
\end{proof}

There are no invariant vector fields or fields with one internal index thus we are left with one scalar constraint.

\subsubsection{The closed (k=1) FRWL universe} 

 In this case the spatial manifold is chosen to be $\sigma=S^3$ and the symmetry group acting on it by rotations is $\Phi=SO(4)$ or $O(4)$. Thus as $\Phi_o=SO(3)$ or $O(3)$
\begin{equation}
\overline{\M}_{ADM}=\{(q_{ab}=\tilde{p}{q_o}_{ab},P^{ab}=\alpha{q_o}^{ab})\colon \tilde{p}>0,\alpha\in \R\}
\end{equation}
where $q_o$ is the round metric on the sphere.
 
If we identify $S^3$ with $SU(2)$ then we have the action ${Spin}(4)=SU(2)\times SU(2)$ on $S^3$
\begin{equation}
Spin(4)\times S^3\ni (g_+,g_-,h)\rightarrow g_+hg_-^{-1}\in S^3
\end{equation}
 with $(g_+,g_-)=(-1,-1)$ acting trivially. In fact, this action factorises by $SO(4)$. Let us introduce right invariant vector fields
\begin{equation}
[R_I(f)](g)=\left.\frac{d}{dt}f(e^{t\tau_I} g )\right|_{t=0}
\end{equation}
with $\tau_I:=-i\sigma_I/2$ with $\sigma_I$ being the Pauli matrices and the dual right invariant forms $\Omega^I$
 \begin{equation}
\Omega^I(R_J)=\delta^I_J
\end{equation}
which read explicitly
\begin{align}
\Omega_I:=-2 tr(\tau_I dh\,h^{-1})
\end{align}
From adjoint invariance of the Killing form $\delta_{IJ}$ the metric
\begin{equation}
q_o=\delta_{IJ}\Omega^I\Omega^J
\end{equation}
is invariant under both left and right $SU(2)$ transformation and this is the round metric on the sphere. The Maurer- Cartan equation reads 
\begin{equation}\label{eq_MC}
d\Omega^I+\frac{1}{2}{\epsilon^I}_{JK}\Omega^J\wedge\Omega^K=0
\end{equation}
We introduce 
\begin{equation}
{E_o}_I^a=\sqrt{\det q_o}R_I^a
\end{equation}
As the stabilizing group is $SO(3)$ or $O(3)$ every invariant $E$ field need to be proportional
\begin{equation}
E^a_I=p{E_o}_I^a
\end{equation}
Thus $\tilde{p}=|p|$.\\

From $e^I=\Omega^I$ and the formula for torison-freeness
\begin{align}
de^I+\frac{1}{2}\epsilon^{I}_{\;JK} \Gamma^J\wedge e^K=0
\end{align}
we determine the spin connection:
\begin{equation}
\Gamma^I_a=\Omega^I_a
\end{equation}
and finally using $K_{ab}=-\frac{\kappa}{2}\alpha\sqrt{\tilde{p}}{q_o}_{ab}$ we arrive at
\begin{equation}
A^I_a=\Gamma^I_a+\beta K_{ab}E^b_J\delta^{IJ}\sqrt{|\det E|}^{-1}=(1+c)\omega^I_a
\end{equation}
where $c=-\frac{\kappa\beta}{2}\alpha\operatorname{sign} p$.\\
Unintuitive, at the first look, parametrisation can be explained if one notice that
the only left $O(3)$ gauge transformation is $\G(-{\mathbb I})$ and it acts
\begin{equation}
(p,c)\rightarrow (-p,-c)
\end{equation}
As there is no invariant vectors, there are no left Gauss constraints,
\begin{equation}
\overline{\M}_{AB}=\{(E_I=p{E_o}_I,A^I=(1+\beta c)\Omega^I)\colon p,c\in \R\}
\end{equation}
where $\tilde{p}=|p|$ (and thus we have double cover of $\overline{\M}_{ADM}$).

\begin{remark}
We can also consider Bianchi IX spacetime.
Let us notice that we have an isomorphism
\begin{equation}
SO(4)/Z_2=SO(3)\times SO(3)
\end{equation}
We consider a inverse image in $SO(4)$ of the group
\begin{equation}
SO(3)\times H
\end{equation}
This is our $\Phi$. By similar fact as before we are left with Bianchi IX metrics and as there are no invariant vector fields we are left with one scalar constraint.
\end{remark}

\section{Gravity on a Graph}
\label{section:5}
In this section, we can draw from the experience learned in the previous section to put theorem \ref{theorem:2} from section \ref{section:2} into action for an altered version of GR. Inspired from Quantum Gravity approaches on the lattice (such as \cite{AQG1}) we consider a discretisation of general relativity. That is, we artificially truncate the full phase space of Ashtekar-Barbero connection and triadic field to finitely (or at least countable) many degree of freedom. \\
In the first subsection \ref{s51_restriction}, this truncation happens along the lines of \cite{Thi:07,Thi:00(QSDVII),DL:17b} and gives us a canonical theory that is formulated on a graph and its dual cell complex on the spatial coordinate manifold. The truncated (or: discretised) phase space will then consist of holonomies associated to the edges of the graph and fluxes corresponding to the (unique) faces of some cell complex that are punctured by an edge of the graph, as we outline in \ref{s511}. Those objects are usually understood as $SU(2)$- respectively $\mathfrak{su}(2)$-valued. Since the original gauge group of GR in the Ashektar-Barbero formulation is in fact $O(3)$, in \ref{s512} we pay special attention to the Gauss constraint and its discretisation. We want to implement the discretised $O(3)$-gauge transformation in such a way, that their action commutes with the truncation process itself. Although this can be achieved in general, it requires the abstract graph-phase space to carry knowledge of the topology of the original smooth manifold. This important result highlights how holes and punctures of space which imply non-contractible loops of the embedded graph remain as additional $O(3)$-gauge transformations, belonging to those of $O(3)$-valued field that cannot be lifted to $SU(2)$-valued fields. While this finishes the truncation on the kinematical level to a given subspace, we note in \ref{s513} that the constraints of general relativity can {\it not} be formulated exactly in terms of the truncated variables. In fact, for any finite graph, e.g. the scalar constraint can be approximated in terms of quantities on the graph (see \cite{Thi98(QSD),Thi98(QSDII)}) in such a way that the approximation will agree with the continuum formulation once the regulator is removed, that is one considers graphs that are finer and finer embedded into the manifold and ultimately filling it. In presence of finite regulators, that is for finite graphs, however the dynamics generated by the approximated constraints differs from the one of the continuum one. One could speculate that such a discretised dynamics could be the prediction of the semiclassical level of some quantum gravity theory on the lattice, as we will outline in the conclusion.\\
It is therefore of interest, to ask whether it is also in this setting possible to symmetry reduce the dynamics of the full truncated phase space to some symmetric subspace, for examples those phase space points characterising cosmological data. Subsection \ref{section:52_binachi} will therefore consider compact, cosmological spacetimes of Bianchi I typ and their truncation to the discrete phase space of a finite, cubic lattice. We introduce a discretisation of the constraints inspired by \cite{Thi98(QSD),Thi98(QSDII),AQG1}. In \ref{s521}, we note that the points of phase space describing Bianchi I can be described via symmetry restriction with a certain symmetry group $\Psi^\gamma_I$.\footnote{At this point we introduce a real number $\mu_o$ associated to the number of vertices on the lattice. In this sense our results will bridge to the $\mu_o$ approach common in effective Loop Quantum Cosmology, see e.g. \cite{DL:17a,DL:17b}.}. In \ref{s522}, we show that this symplectic phase space can be endowed with a reduced symplectic structure, which follows uniquely from using holonomies and gauge-covariant fluxes as basic variables and differs from other recent proposals in the literature \cite{LS:19}. The discrete version of the scalar constraint is found to be invariant under $\Psi^\gamma_i$ and can thus be used for symmetry restriction. We close \ref{s523} with the realisation that theorem \ref{theorem:2} applies, i.e. the flow of discretised GR stays within the symmetry restricted subspace and therefore can be computed on the restricted level as well.\\
Finally, in subsection \ref{s53_FRW} we perform a further symmetry restriction to the discretised version of compact, isotropic flat cosmology, i.e. FLRW. The same restriction theorems apply and we find for the restricted, discretised scalar constraint the same formula derived in \cite{DL:17a}. When connecting to application of semiclassical effects of Quantum Gravity, this strengthens the conjecture that this reduced constraint does indeed map the same trajectory as a semiclassical, effective Hamiltonian obtained from coherent states.

\subsection{Truncation from continuum to discrete}
\label{s51_restriction}

We will introduce a truncation of the connection formulation of general relativity as presented in section \ref{section:3_connection}. 
Before doing so, we shortly justify our approach to the problem. Gravity on the graph is a theory that is supposed to appear as the semiclassical limit of a discrete (fixed graph) version of (Algebraic) Loop Quantum Gravity \cite{Thi:07,GP:00,AL04,AQG1}. The semiclassical limit is not completely understood, but the phase space is known: in it fluxes do not commute and we look for a discretisation with such feature. Fortunately, such a discretisation is available \cite{Thi:00(QSDVII)}. It serves as guiding principle for relating objects in gravity on a graph and full General Relativity theory.

As a first step, we chose a subdivision of $\sigma$ into a cell complex. That is, into countably/finitely many 3-dimensional, compact mutually distinct cells $c_i\subset \sigma$ such that each pair of cells $(c_i,c_j)$ has either no common boundary at all or the boundary is exactly one contractible face $S=c_i\cap c_j$. We fix a point $v_i\in c_i$ for each cell and for all cases where $c_i\cap c_j\neq \emptyset$ we connect those points $v_i$ and $v_j$ via a semi-analytic path $e(t):[0,1]\to \sigma$ such that
\begin{enumerate}
\item $e(0)=v_i$ and $e(1)=v_j$
\item $e$ intersects $c_i\cap c_j$ transversally and only once at the point $e(1/2)$.
\end{enumerate}
For the graph $\gamma$, we will denote set of vertices by $\gamma_v$ and set of edges with orientation by $\gamma_e$. We introduce orientation on $c_i\cap c_j$ by orientation of the edge $e$. We denote such orientated surface by $S_e$. Let us also introduce reverse edge
\begin{equation}
e^{-1}(t)=e(1-t)
\end{equation}
and then $S_{e^{-1}}$ is just the face $S_e$ but with opposite orientation (see \eqref{gamma_s}). By the oriented edge we understand a map $e\colon [0,1]\rightarrow \sigma$ where we identify two maps if they differ by a reparametrisation. We assume that if $e\in \gamma_e$ then also $e^{-1}\in \gamma_e$. Sometimes we are interested only in the path up to the orientation. We introduce equivalence relation $[\cdot]$
\begin{equation}
[e']=[e]\Longleftrightarrow e'=e \lor e'=e^{-1}
\end{equation}
We denote the set of equivalence classes (edges without orientation) by $\gamma_{[e]}$ of the graph $\gamma$. The counting by $\gamma_e$ is over-complete.

The space of surfaces will be denoted by $\gamma_s$. 
We define the surface as a smooth embedding
\begin{equation}
S\colon D\rightarrow \sigma
\end{equation}
where $D=\{(x,y)\in \R^2\colon x^2+y^2<1\}$.
We regard two maps as equal if differing by a rotation
\begin{equation}
S'(x,y)=S(\cos\phi x+\sin\phi y, \cos\phi y-\sin\phi x)
\end{equation}
We assume that
\begin{equation}\label{gamma_s}
S_{e^{-1}}=S_e\circ i
\end{equation}
where $i(a,b)=(-a,b)$ and that $e\left(\frac{1}{2}\right)=S_e(0,0)$. We assume that surfaces of the graph are disjoint or equal up to orientation.

\subsubsection{Truncation of phase space $\M_{AB}$ to $\M_\gamma$ }
\label{s511}
Instead of working with the phase $\M_{AB}$ of connection and triad fields on $\sigma$ we will restrict our attention to a subspace $\M_\gamma\subset \M_{AB}$, in which all information is lost expect for some information related to the graph $\gamma$. This truncation is done by keeping only the following data:\\
The information of the connection $A^I_a(x)$ is truncated to holonomies. That is for each edge $e\in\gamma$, we introduce the $SU(2)$-valued path-ordered exponential: \cite{Thi:07,GP:00,AL04}
\begin{align}\label{holonomy}
h(e) := \mathcal{P}\exp\left(
\int_0^1{\rm d}t\, A^J_a(e_k(t))\tau_J\dot{e}^a_k(t)\right)
\end{align}
where in the path ordered exponential the latest time values are ordered to the right.

\begin{remark}
Let us stress an issue of great importance. The holonomy in (\ref{holonomy}) is a $SU(2)$ group element even if it is obtained from our $\mathfrak{so}(3)$ connection. There is no contradiction here, because $\mathfrak{su}(2)=\mathfrak{so}(3)$ as Lie algebras. However, this choice will have nontrivial influence on our analysis.
\end{remark}

The information about the electric field is truncated to the $\mathfrak{su}(2)$-valued gauge-covariant fluxes \cite{Thi:00(QSDVII)} for each $e$ (making use of its associated face $S_e$):
\begin{align}\label{gcFluxes}
P(e):= h(e_{[0,1/2]})\left[\int_{S_e}{\rm d}x h(\rho_x)\ast E(x)h(\rho_x)^\dagger\right] h(e_{[0,1/2]})^\dagger
\end{align}
where $\ast$ denotes the hodge star operator and 
\begin{equation}
\rho_x\colon [0,1]\rightarrow \sigma
\end{equation}
is defined for $x=S_e(a,b)$ by
\begin{equation}
\rho_x(t)=S_e(ta,tb)
\end{equation}
Let us notice that
\begin{equation}
P(e^{-1})=-Ad_{h(e)}P(e),\quad h(e^{-1})=h(e)^{-1}
\end{equation}
Once a graph $\gamma$ is chosen, along the lines of \cite{Thi:00(QSDVII)} one can perform a kinematical truncation of the phase space by considering only functions of $P(e)$ and $h(e)$. The Poisson brackets are given by\footnote{In order to deduce these brackets from the continuous one in (\ref{AB_PB}), one has to introduce a regularisation of the edges and faces in form of three-dimensional thickenings. At the end, said regularisation is again removed. See \cite{Thi:00(QSDVII)} for further details.}
\begin{align}\label{discrete-PB-algebra}
\{h(e)_{ab}, h(e)_{a'b'}\}=0,\hspace{10pt} \{P^I(e),h(e)\}=\frac{\kappa \beta}{2}\tau_I h(e),\hspace{10pt} \{P^I(e),P^J(e)\}=-\frac{\kappa\beta}{2}\epsilon^{IJ}_{\;\;\;K}P^{K}(e)
\end{align}
where $P^I(e):= -2{\rm Tr}(\tau_I P(e))$ and $\tau_I:=-i\sigma_I/2$ with $\sigma_I$ being the Pauli matrices and $\epsilon_{IJK}$ is the Levi-Civita symbol.\\
By construction of our cell complex, the edges are disjoint and we assume that the faces $S$ are open, hence no intersection exists between faces of different edges, i.e. $S_e \cap S_{e'} =\emptyset$. This is important to introduce a regularisation for the   Poisson bracket between $P^I(e)$ and $P^J(e')$ if the edges share a vertex. This yields finally for disjoint edges $e,e'$ that
\begin{equation}
\{h(e),h(e')\}=\{P^I(e),h(e')\}=\{P^I(e),P^J(e')\}=0.
\end{equation}

\begin{fact}
It turns out that on the truncated phase space the Poisson bracket is nondegenerate and it is given by a symplectic form that is a sum over all edges 
\begin{equation}
\omega_\gamma =\sum_{[e]\in \gamma_{[e]}} \omega_e
\end{equation}
where $[e]\in \gamma$ means that we take every edge once choosing an orientation.
The symplectic form $\omega_e$ is a symplectic form on $T^\ast SU(2)$ of a given edge that is 
\begin{equation}\label{pre-symplectic_potential_graph}
\omega_e=d\xi_e,\quad \xi_e=\frac{2}{\kappa\beta} P^I(e) \Omega_I(e)
\end{equation}
where $\Omega_I(e):=-2 tr(\tau_I dh(e)\,h^{-1}(e))$ is the right invariant Maurer Cartan form on $SU(2)$ represented by $h(e)$. 
\end{fact}

\begin{proof}
We compute:
\begin{align}
\omega_e =&d\xi_e = \frac{2}{\kappa\beta}dP^I(e)\wedge \Omega_I(e) -\frac{4}{\kappa\beta}P^I(e) tr(\tau_I dh(e) \wedge dh^{-1}(e))\nonumber\\
=&\frac{2}{\kappa\beta} (\,dP^I(e)\wedge \Omega_I(e) -\frac{1}{2}P^I(e)\epsilon_{IJK} \Omega_J(e)\wedge\Omega_K(e)\,)
\end{align}
using that $tr(\tau_I\tau_J)=-\delta_{IJ}/2$ and $dh^{-1}=-h^{-1}dh \, h^{-1}$.\\
Now, the Hamiltonian vector fields for $P^I$ and $h$ are found to be
\begin{align}
\chi_{P^I(e)} &=\frac{\kappa\beta}{2}(\, R^I - P^K \epsilon_{KIJ}\frac{\partial}{\partial P^J}\,)\\
\chi_{h_{ab}(e)} &=\frac{\kappa\beta}{2} [\tau_I h]_{ab} \frac{\partial}{\partial P^I} 
\end{align}
where the right-invariant vector fields $R^I$ together with $\partial/\partial P^J$ form a basis of the tangent vector space.\\

Finally, we derive the Poisson brackets via $\{f,g\}:=-\omega(\chi_f,\chi_g)$ to obtain (\ref{discrete-PB-algebra}).
\end{proof}

Let us notice that $\xi_e=\xi_{e^{-1}}$. Since the holonomies are $SU(2)$ valued, the reduced system corresponding to each edge of $\gamma$ as discretisation of the Ashtekar-Barbero phase space is $T^\ast SU(2)$. And consequently the phase corresponding to the full graph is
\begin{equation}
T^\ast(\gamma):=\prod_{[e]\in \gamma_{[e]}} T^\ast SU(2)
\end{equation}
We will denote this phase space by $\M_\gamma$:
\begin{align}
\M_\gamma=\{(h,P)\; : \;h : \gamma_e \to SU(2) \;,\; P:\gamma_e \to su(2) \}
\end{align}

\subsubsection{Truncated gauge transformations}
\label{s512}

Our theory on the graph is a truncation of the full theory, thus not all gauge transformation can be implemented. The main problem is if the variable $(P(e),h(e))$ after applying transformation in the full theory cannot be expressed as functions of original holonomies and fluxes. The necessary condition is that the graph is preserved, thus definitely not all diffeomorphisms can be implemented, but what with $O(3)$ transformations?
A $O(3)$ gauge transformation is a field over $\sigma$ valued in $O(3)$ whose action on $\M_{AB}$ is defined via (\ref{eq-GG}) and (\ref{action_of_gauss_so(3)}). With the variables $(P(e),h(e))$ being functions of the triad and the connection in $\M_{AB}$, the hope is to find for every gauge transformation an equivalent transformation on $\M_{\gamma}$ such that truncation and gauge transformation commute:
\begin{center}
$(A,E)\in\M_{AB} \overset{\G_{AB}( O(3))}{\longrightarrow} \G_{AB}(O)(A,e)\in\M_{AB}$\\
$\downarrow \hspace{100pt} \downarrow$\\
$(h,P)\in\M_{\gamma} \overset{\G^\gamma(O(3))}{\longrightarrow} \G^\gamma(g)(h,P)\in \M_\gamma$
\end{center}

\paragraph{The $O(3)$ gauge transformations}
We concentrate first on $SO(3)$ gauge transformations given by $O\colon \sigma \rightarrow SO(3)$, such that there exists a smooth lift to $SU(2)$ i.e.
\begin{equation}
g\colon \sigma \rightarrow SU(2),\quad \pi(g)=O
\end{equation}
where $\pi\colon SU(2)\rightarrow SO(3)$ is the double cover homomorphism, i.e. a projection. Let us notice that all transformations generated by Gauss constraints have this property. Moreover, the lift is unique up to global multiplication by $-1$.

The advantages of our variables is their behaviour under $SU(2)$ gauge transformations: For a gauge field $g(x)$, the holonomy and gauge-covariant flux transform as \cite{Thi:07,AL04,GP:00}
\begin{align}\label{ActionOfGaussOnLattice}
\G^\gamma(g)\; :\; (h(e),P(e))\; \mapsto\; (\,g(e(0))h(e)g(e(1))^{-1}\, ,\,  g(e(0)) P(e) g(e(0))^{-1}\,)
\end{align}
whenever $g(x)$ is the associated smooth lift of $O$. This makes it comparably easy to construct gauge invariant functions.

The situation is a bit more complicated when the function $O$ does not have a lift. In this case we can construct a double covering of $\sigma$ 
\begin{align}
&\tilde{\sigma}_O=\{(x,g)\in \sigma\times SU(2)\colon \pi(g)=O(x)\},\\
&p_O\colon \tilde{\sigma}_O\ni (x,g)\rightarrow x\in \sigma
\end{align}
We remind that such a double cover looks locally like two copies of the manifold $\sigma$. However, it has the additional requirement that at each point it is continuous. In other words, the map $\pi^{-1}(O(x))$ has two branches and it is smooth in each branch. On the global level, however this allows for the possibility of twisted geometries.\\
On this bigger manifold we can lift $O\circ p_O$ tautologically 
\begin{equation}
\tilde{g}\colon \tilde{\sigma}_O\ni (x,g)\rightarrow g\in SU(2),\quad {\rm such \, that}\quad \pi \circ\tilde{g}=O\circ p_O
\end{equation}
Further, for any graph $\gamma$ we can consider its pre-image in $\tilde{\sigma}_O$, i.e. $\tilde{\gamma}=p_O^{-1}(\gamma)$ which is the double covering of $\gamma$
\begin{equation}
p_{O}\colon \tilde{\gamma}\rightarrow \gamma
\end{equation}
On $\tilde{\gamma}$ we have the mentioned lift $\tilde{g}$,  thus we may define for $e\in \gamma_e$ the generalisation of (\ref{ActionOfGaussOnLattice})
\begin{align}
\G^\gamma(O)\;\colon \hspace{25pt}\M_\gamma\;\;&\rightarrow\;\; \M_\gamma\\
 (h(e),P(e))\; \; &\mapsto
\; \;(\tilde{g}(\tilde{e}(0))h(e)\tilde{g}(\tilde{e}(1))^{-1},\; \tilde{g}(\tilde{e}(0)) P(e) \tilde{g}(\tilde{e}(0))^{-1})
\end{align}
where $\tilde{e}$ is one of two edges of $\tilde{\gamma}$ satisfying $p_O(\tilde{e})=e$. The result does not depend on the choice of this edge.  These are also symplectic transformations. If $O:\sigma\to SO(3)$ does not allow an lift $SU(2)$ by itself, then these transformations are not in connected components.\footnote{The tangent vector to the path of such gauge transformations is generated by Gauss transformations. Thus, if there is a path connecting some given gauge transformation to the identity, taking the same path-ordered gauge transformation in SU(2) we can obtain a lift. }
Apparently these transformations were never considered in earlier works.

\begin{remark} We present an example, which highlights that the additional non-liftable $O(3)$ transformations do indeed have non-trivial action and not all $SU(2)$-gauge invariant observables will be also $O(3)$ gauge-invariant. In explicitly, we  show that the famous Wilson loops along non-contractible paths are {\it not} $O(3)$-gauge invariant.\\
Let us choose a loop $e\colon [0,1]\rightarrow \sigma$ with $e(0)=e(1)=x$, such that it defines nontrivial element of $H_1(\sigma, \Z_2)$ homology group via Hurewicz map.
Double coverings of $\sigma$ are classified by $H_1(\sigma, \Z_2)$. For every double covering there exists $O\colon \sigma\rightarrow SO(3)$ such that $p_O\colon \tilde{\sigma}\rightarrow \sigma$ is homotopic to this covering.\footnote{Two coverings are classified by maps $\sigma\rightarrow \mathbb{RP}^\infty$ (see \cite{Husemoeller}). However $\sigma$ is 3 dimensional so we can homotopically deform the map to a $3$ skeleton of $\mathbb{RP}^\infty$ that is $\mathbb{RP}^3=SO(3)$. By integrating obtained map with some bump function on $SO(3)$ we obtain a smooth map such that the pull back of $\pi : SU(2)\rightarrow SO(3)$ is homotopic to the original two coverings.} If we follow $e$ from one point of $(x,u)=p_O^{-1}(x)$ we will end up in the other point $(x,-u)$, thus
\begin{equation}
\G(O)h(e)=\tilde{g}(x,u)h(e)\tilde{g}(x,-u)^{-1}=-u h(e) u^{-1}
\end{equation}
We notice that the trace
\begin{equation}
\operatorname{tr} \G(O)h(e)=\operatorname{tr} \left(-u h(e) u^{-1}\right)=-\operatorname{tr} h(e)
\end{equation}
is not invariant under this transformation. 
\end{remark}

\begin{remark}
The effect of $O(3)$ gauge transformations that are not liftable is associated to the presence of non-contractible loops on $\sigma$. Indeed, for every contractible loop, it is $\tilde{g}(0) =\tilde{g}(1)$, thus every $O(3)$ gauge transformation on this holonomy can be lifted to $SU(2)$.\\
In this sense - when transcending to $\M_\gamma$ - the discrete theory has knowledge of the original continuous manifold the theory was originally defined on, hidden in the presence/absence of the additional non-liftable $O(3)$ transformations.\footnote{We point out that the most common regularisations of the scalar constraint are invariant under $O(3)$ as they are understood as regularisations along infinitesimal small (i.e. contractible) paths. However, in parallel ``single vertex universe''-descriptions of LQG appear promising \cite{RS13,Bod15} where a the torus is truncated to three non-contractible edges. While this might be interesting from the point of view of maximal coarse-graining, the common regularisation of AQG fails to be a suitable candidate as it is not invariant with respect to the non-liftable transformations.}
This implies that the presence of additional $O(3)$ transformation distinguishes the discretisation of smooth geometry, with respect to those, where some minimal loops in $\gamma$ are not-contractible anymore.\footnote{In the context of quantum theory, one often considers extensions of the gauge groups in question. Such extension can have drastic implications for the gauge group: e.g. in \cite{BT07} the group of automorphisms $\rm{Aut}(\sigma)$ was considered as extension of $\Diff(\sigma)$ and it was shown that with this extension no information about knottings and holes survives. As a consequence, if one is interested in all of the $O(3)$ transformations, one can {\it not} consider the extension $\Diff(\sigma) \to \rm{Aut}(\sigma)$.}
\end{remark}

We now turn our attention to $\G(-{\mathbb I})$ transformation. As in this case the transformation of the connection involves $\Gamma^I$ along the edges we cannot implement it on our truncated phase-space. However, on  the symmetric configurations we  see the shadow of this transformation in form of the symmetry $(p,c)\to (-p,-c)$.

\paragraph{The residual diffeomorphisms}

Let us now consider action of diffeomorphisms that preserve the lattice, namely $\psi\in \Diff(\sigma)$ such that
\begin{equation}
\forall_{e\in \gamma_e} \psi(e):=\psi\circ e\in \gamma_e,\quad 
\forall_{v\in \gamma_v} \psi(v):=\psi\circ v\in \gamma_v,\quad
\forall_{S\in \gamma_s} \psi(S):=\psi\circ S\in \gamma_s,
\end{equation}
We will write in such situation $\psi(\gamma)=\gamma$.

In this case the diffeomorphism induces a permutation on the spaces of edges and vertices
\begin{align}\label{psi_ev}
\psi^e\colon \gamma_e\rightarrow \gamma_e,\quad \psi^e(e)=\psi(e)\\
\psi^v\colon \gamma_v\rightarrow \gamma_v,\quad \psi^v(v)=\psi(v)
\end{align}
Let us remind
\begin{equation}
\D_{AB}(\psi)(E^a_I,A^I_a)=(\psi_*(E^a_I),\psi_*(A^I_{a}))
\end{equation}
The action on phase space $\M_{ADM}$ preserves the space of the functions of the truncated theory. Precisely, from the definition of the holonomy variables
\begin{equation}
\D_{AB}(\psi)(h(e))=h(\psi^e(e)),\quad \D_{AB}(\psi)(P(e))=P(\psi^e(e)),
\end{equation}
These formulas define the implementation of such diffeomorphisms in our truncated theory. We will denote the action of the residual diffeomorphisms by 
\begin{align}
\D_{AB}^\gamma(\psi)\, &\colon\, \M_\gamma\rightarrow \M_\gamma\\
[\D_{AB}^\gamma(\psi) (h,P)](e) &:= ( h(\psi^e(e)), P(\psi^e(e)))
\end{align}
It is a symplectic transformation.
The group of residual diffeomorphisms and $SO(3)$ gauge transformations form together, similarly to the case of continous theory, a semidirect product. The resulting group collects all the possible gauge transformations of the truncated (discretised) theory. However, usually not all symmetries of the graph correspond to symmetries of the full cellular complex.\footnote{A simple example consist of the dipole graph (two vertices connected by several edges) with the two surfaces defined by triangulation of the sphere separating two vertices. We assume that every edge intersects its own surface on the sphere.  The triangulation might be highly nonsymmetric, whereas every permutation of the edges is a symmetry of the graph.} In order to describe $\D^\gamma_{AB}$ as truncated action of $\D_{AB}$, it is necessary that also the surfaces transform properly, i.e. $\gamma(S_e)=S_{e'}$ for some $e,e'\in\gamma_e$. Accordingly, not all possible transformations $\D_{AB}^\gamma(\psi)$ are gauges. In our applications this subtlety is not arising.
\subsubsection{Truncation of the constraints}\label{sec:truncation}
\label{s513}

We will now turn our attention towards the constraints. However, in general the constraints are not invariant functions of the truncation procedure to $\gamma$ that we have outlined before. Therefore, what we are doing is {\it no} reduction of the theory, but a {\it truncation}. We are indeed non-trivial altering the theory, by postulating discretised constraints, which are functions of holonomies and fluxes, which only in a certain limit coincide with the original constraints of general relativity.\\

\begin{itemize}
\item The particularly simple are only Gauss constraints as the group generated by them can be implemented exactly in our truncated theory. Let us consider a one parameter subgroup of $SO(3)$ gauge transformation\footnote{Every such family can be lifted to $SU(2)$.} such that
\begin{equation}
O_t^\Lambda(v)=\Pi(e^{it\tau^I \Lambda_I(v)})
\end{equation}
One can check\footnote{Compare semiclassical description of representations (see \cite{GS:13})} that the generator of this group of symplectic transformations is given by the formula
\begin{equation}
G[ \Lambda]=\sum_{v\in\gamma_v} \Lambda_I(v) G^I(v),\quad G^I(v)=\sum_{e\in \gamma_e\colon e(0)=v} P^I(e) 
\end{equation}
We can regard the Gauss constraint as a function
\begin{equation}
G\colon \M_\gamma\rightarrow C^\infty(\gamma_v, so(3)=su(2))
\end{equation}
where
\begin{equation}
C^\infty(\gamma_v, su(2))=\{f\colon \gamma_v\rightarrow su(2)\}
\end{equation}

We can also write the covariance property under $SO(3)$ gauge transformations:
\begin{equation}
\G(O)G=\operatorname{Ad}_O G
\end{equation}
where $(\operatorname{Ad}_O G)(v)=g\, G(v)\, g^{-1}$ with $\pi(g)=O(v)$, $g\in SU(2)$. The result does not depend on the choice of particular $SU(2)$ group elements $\tilde{O}$ as they differ by $-1$. 

\item  On the other hand the diffeomorphism constraints are tricky. Usually, they are not considered, because in the original LQG their group is automatically taken into account.
In the lattice approach we assume that their generators exists and that 
\begin{equation}
D\colon \M_\gamma\rightarrow \prescript{1}{}{C}^\infty(\gamma)
\end{equation}
where
\begin{equation}\label{eq331}
\prescript{1}{}{C}^\infty(\gamma)=\{f\colon \gamma_e\rightarrow \R,\ f(e)=-f(e^{-1})\}
\end{equation}
is the space which corresponds to the vector fields in continuous theory. However, this is not the only possible choice.
One can argue that the symmetry of the lattice should be implemented by the action of the residual constraints (lifted).\footnote{The advantage of this definition is that in the examples in \ref{section:52_binachi} because for (\ref{eq331}) we will find $\,^1 C^\infty(\gamma)_{\Psi^\gamma}$ and thus $D$ will be zero.}\\

\item The situation with the scalar constraint $C(x)$ is even more complicated, but in this case there exist several propositions. We introduce a function $C^{\epsilon}(v)$ of holonomies and fluxes such that in the limit of infinitely dense lattices it agrees with the continuum expression, i.e.:
\begin{align}\label{cont_limit}
\lim C^{\epsilon}(v)= C(x=v)
\end{align}
In the specific examples we can introduce discretisation based for example on Thiemann (see subsection \ref{section:52_binachi}).The most important for us is the requirement that this constraints behave nicely under residual diffeomorphisms. Namely, for $\psi$ that preserves the graph, we require that
\begin{equation}\label{covariance_Discrete}
\D_{AB}^\gamma(\psi)_*C^\epsilon(v) =C^\epsilon (\psi_v(v))
\end{equation}
This equation is the equivalent of (\ref{spatial_diffeo_inv_of_GR}) and allows to repeat the same argumentation from section \ref{s312} in order to classify the remaining constraints on the symmetry reduced phase space of $\M_\gamma$.\\
Let us notice that $C^\epsilon$ can be regarded as a function from phase space to the scalar functions on the lattice
\begin{equation}
C^\epsilon\colon \M_\gamma\rightarrow C^\infty(\gamma_v)
\end{equation}
where
\begin{equation}
C^\infty(\gamma_v)=\{f\colon \gamma_v\rightarrow \R\}
\end{equation}
We can now can state  the covariance property (\ref{covariance_Discrete})
for $\psi$ preserving the graph and $\psi_v$ defined in \eqref{psi_ev}. 
We assume that the proposition of scalar constraint is $SO(3)$ invariant, i.e. $\G(O)C=C$, what is usually easy to check.

\end{itemize}

\subsubsection{Truncated symmetry group}
\label{s514}

We are mainly interested in the action of the $\Psi$ group. We now assume that $\Psi\in \Diff^+(\sigma)$. We introduce the truncated group
\begin{equation}
\Psi^\gamma=\{\psi\in \Psi\colon \psi(\gamma)=\gamma\}
\end{equation}
Similarly we can introduce (see \eqref{final_sym_group_on_Ab}) a subgroup of a group generated by graph preserving diffeomorphisms and $SO(3)$ gauge transformations
\begin{equation}\label{final_sym_group_on_Ab-lattice}
\Phi_{AB,E_o}^\gamma(\Psi)=\{\;\G^\gamma({E_o}_a^J\psi_\ast({E_{o}}^I_a))\D_{AB}^\gamma(\psi)\colon \psi\in\Psi^\gamma\}
\end{equation}
This is the symmetry group that we will impose on our truncated theory concerning symmetry restriction.

\subsection{Symmetry restricted Bianchi I model on a cubic lattice}
\label{section:52_binachi}

Let $\sigma=(\R/T)^3$ and consider the cubic lattice $\gamma$ oriented along the three directions of the torus, such that each edge in direction $l$ has constant tangent vector $(\dot{e}_{v,l})^a(t) =(T/N_l) (\hat{n}_l)^a$ with $\hat n_l$ the unit vector in direction $l$.
The graphs $\gamma$ is a regular cubic lattices $\gamma$ with periodic boundary conditions. The set of vertices
\begin{equation}
\gamma_v=\Z_{N_1}\times  \Z_{N_2}\times \Z_{N_3}
\end{equation}
and the every vertex is six-valent. We will denote every edge with the direction
\begin{equation}
l\in L=\{\pm \hat{n}_i\colon i=1,2,3\}
\end{equation}
and the vertex from which it starts, namely $e=(l,v)$, thus
\begin{equation}
\gamma_e=\{(l,v)\colon l\in L,\ v\in \gamma_v\}
\end{equation}
The orientation is determines by a starting vertex. Thus
\begin{equation}
e_{l,v}=(-\hat{n}_i,v)=(\hat{n}_i,v-\hat{n}_i)^{-1}
\end{equation}
We can then write any element of the phase space as $\gamma = (P(l,v),h(l,v))=(P(e_{l,v}),h(e_{l,v}))$.

\paragraph{Truncation of the constraints} 
We will now turn our attention towards the constraints, describing in details the missing point from section \ref{s513}: the scalar constraint. We use the following notation for its smearing against a function $\mathcal{N}(v)$:
\begin{align}
C^\epsilon[\mathcal{N}]:= \sum_v C^\epsilon(v) \, \mathcal{N}(v)
\end{align}
Based on the regularisation of Thiemann \cite{Thi98(QSD),Thi98(QSDII)}, we introduce the following discretisation\footnote{We stress again that this function is ad hoc.
 There are many choices that would obey (\ref{cont_limit}). In the present of finite lattice, the produced flow on $T^\ast(\gamma)$  will in general heavily depend on this choice.} specialised to the case of cubic lattices
\begin{align}\label{DiscreteScalarConstraintThiemann}
C^\epsilon(v)&:=C_E(v)+C_L(v)+C_{\rm matter}^\epsilon(v),\\
C_E(v)&=\frac{-1}{2\kappa^2\beta}\sum_{i,j,k\in L}\epsilon(i,j,k){\rm Tr}\left(
[h(\square_{v,ij}) - h^\dagger(\square_{v,ij})]h(k,v)\{h^\dagger(k,v),V^\epsilon[\sigma]\}\right),\nonumber\\
C_E[\sigma]&=\sum_v C_E(v),\quad V^{\epsilon}[\sigma]=\sum_v\sqrt{\sum_{ijk\in L}\frac{\epsilon_{IJK}}{48}\epsilon(i,j,k) P^I(i,v)P^J(j,v)P^K(k,v)},\quad K=\{C_E[\sigma],V^\epsilon[\sigma]\},\nonumber\\
C_L(v)&=\frac{8(1+\beta^2)}{\kappa^4\beta^7}\sum_{ijk\in L}\epsilon(i,j,k){\rm Tr}\left(
h(i,v)^\dagger\{h(i,v),K\}h(j,v)^\dagger\{h(j,v),K\}h(k,v)^\dagger\{h(k,v),V\}
\right)\nonumber
\end{align}
where $\Box_{v,ij}$ denotes a plaquette starting at $v$ in direction $i$ and returning along direction $j$ and $\epsilon(i,j,k):=\operatorname{sgn}\det(i,j,k)= {\rm sgn}(ijk) \epsilon_{|i||j||k|}$. 

Let us notice that $C$ can be regarded as a function from phase space to the scalar functions on the lattice
\begin{equation}
C\colon \M_\gamma\rightarrow C^\infty(\gamma_v)
\end{equation}
satisfies all the requirements from Section \ref{sec:truncation}, namely
\begin{enumerate}
\item It is $SO(3)$ gauge invariant,
\item It is equivariant under $\Psi^\gamma$ transformations. Thus also under the group $\Phi_{AB,E_o}^{\gamma}$ .
\end{enumerate}

\begin{remark}
It is worthwhile to mention that invariance under $\Phi^\gamma_I$ is related to covariance of a scalar function. Thus, it is to be expected that any reasonable proposal for a regularised scalar constraint will be $\Phi^\gamma_I$-invariant. E.g. the proposal of \cite{AAL:14,ALM:15} is invariant as well and could be used alternatively for the investigations in this manuscript.
\end{remark}

\subsubsection{Symmetry group and structure of restricted phase space}
\label{s521}

We will try to lift the symmetry group (\ref{BianchiI_symmetry_torus}) from the Bianchi I model in the torus case to our discretised setting. Let us consider a subgroup of $\Psi_I$ from (\ref{BianchiI_symmetry_torus}) that preserves the graph. We will denote this group by $\Psi_I^\gamma$. The group in question can be easily determined and it is equal to
\begin{equation} \label{group_graphI}
\Psi_I^\gamma\cong(\mathbb Z_{N_1}\times \mathbb Z_{N_2}\times \mathbb Z_{N_3})\rtimes H=\{(m,O)\colon m\in \mathbb Z_{N_1}\times \mathbb Z_{N_2}\times \mathbb Z_{N_3},\ O\in H\}
\end{equation} 
where $N_i$ are the number of vertices along each principal direction of the torus and $H$ is the group of rotations around axis by angle $\pi$ from (\ref{group_H}). The action of $\Psi^\gamma_I$ permutes edges and vertices via the following maps:
\begin{equation}
 \psi_v\colon \gamma_v\rightarrow \gamma_v, \quad
\psi_e\colon \gamma_e =L\times \gamma_v\to \gamma_e
\end{equation}
where for $\psi=(m,O)$
\begin{equation}
\psi_v(v)=O(v+m),\quad \psi_e(l,v)=(Ol,O(v+m))
\end{equation}
where $O$ acts by linear transformations on $L$ and $V$.\\

We will now analyse the structure of the phase space $\M_\gamma$ under symmetry restriction by group $\Phi_I^\gamma=\Phi_{AB,E_o}^\gamma(\Psi_I)$. We will denote this reduced phase space by $\overline{\M}_{\gamma}^{ I}$.

\begin{lm}
The points of $\M_{\gamma}$ that are invariant under the group $\Phi_I^\gamma(\Psi^\gamma_I)$ are
\begin{equation}\label{MredgammaI}
(P(v,l),h(v,l)) = \bigg(\mu_o^2 p_l\tau_l,e^{\mu_o  c_l\tau_l}\bigg)
\end{equation}
and we denote the space spanned by $(p_l,c_l )\in \mathbb{R}^2$ (under the identification $c_l=c_l+4\pi/\mu_o$) as $\overline{\M}_{\gamma}^{I}$ and defined the real number:
\begin{equation}\label{mu0-def}
\mu_o :=(N_1N_2N_3)^{-1/3}
\end{equation}
\end{lm}

\begin{proof}
Due to translation symmetry for any positively oriented link
\begin{equation}
(P(e=(v,l)),h(e=(v,l)))=(P_l,h_l)
\end{equation}
Let us consider rotation by $\pi$ around $k$ axis, i.e. $R_{e_k}(\pi)\in SO(3)$ as in (\ref{group_H}). The easiest way to lift this to $SU(2)$ is via its generators $R_{\vec{n}}(\pi)=\exp( \pi \vec{n}\cdot\vec{\tau}^{(1)})$. $\tau^{(j)}_I$ with $I=1,2,3$ are the generators of the Lie algebra $so(3)=su(2)$ in representation $j$, which read for $j=1/2$ as $\tau_I^{(1/2)}=-i\sigma_I/2$ with $\sigma_I$ being the Pauli matrices. With this, the $R_{e_k}(\pi)$ are straightforwardly lifted to $SU(2)$, namely $\tilde{R}_{\vec{n}}=\exp( \pi \vec{n}\cdot\vec{\tau}^{(1/2)})$. This lift determines $\G^\gamma$ and with $E_{oa}^I=\delta_a^I$ the action of the symmetry transformations $\Phi^\gamma_{AB,E_o}(\Psi_I)$ from (\ref{final_sym_group_on_Ab-lattice}) is complete. The action of $\G^\gamma(\tilde{R}_{\vec{n}})$ on points in the phase space $\M_\gamma$ is given in (\ref{ActionOfGaussOnLattice}), where $\forall x$ $g(x)= \tilde{R}$. Demanding invariance on a link in direction $k$ for $\tilde{R}_{k}$ does not move the link and implies thus
\begin{align}
P_k=\tilde{R}_{k}\, P_k\, \tilde{R}^{-1}_k,\quad h_k=\tilde{R}_k\, h_k\, \tilde{R}_k^{-1}
\end{align}
The definition of $\tilde{R}_k$ already implies that the only invariant $su(2)$ element needs to be proportional to $\tau_k$ hence $P_k=p(k) \,\tau_k$ for some $p(k)$.
Moreover
\begin{equation}
\tilde{R}_k e^{ \vec{n}\cdot\vec{\tau} } \tilde{R}^{-1}_k=e^{i \vec{n}' \cdot \vec{\tau}}
\end{equation}
where $\vec{n}'$ is a vector rotated by $\pi$ around the axis $k$ thus
also $h_k=e^{i\, c(k)\, \sigma_k}$  for some $c(k)$.
\end{proof}

\begin{remark}
We want to point out that the real number $\mu_0$ from (\ref{mu0-def}) will turn out to be reminiscent of the parameter of the same name in effective Loop Quantum Cosmology. The later one comes typically in two different schemes (the $\mu_o$- and $\bar{\mu}$-scheme) while only the first one with constant parameter $\mu_o$ can be understood as the cosmological sector of a lattice-discretisation from continuum GR via the methods outlined in this manuscript.\\
The construction of a suitable symmetry group always requires the introduction of a fiducial triad $E_{oI}^a$ (see e.g. (\ref{fiducial_triad_FRW})). Together with the fiducial connection $A_{oI}^a$ this allows to be lifted to the discrete setting by giving rise to the fiducial fluxes $P_o^I(e)=P^I(e)|_{E_o,A_o}$. We define the volume of the spatial manifold for said fiducial fluxes:
\begin{align}
V_o := V^\epsilon[\sigma] |_{P_o},\hspace{30pt}\text{and}\hspace{30pt} \epsilon:= \mu_o V_o^{\frac{1}{3}}
\end{align}
take over the role of {\it discretisation parameter}. In this sense, the limit $\epsilon\to 0$ of the examples outlined below will restore the continuum cosmological models, i.e. (\ref{overlineMred}) and (\ref{Constraint_BinachiI_continuum}).
\end{remark}

\subsubsection{Symmetry restriction of the symplectic form}
\label{s522}

We should now describe the reduction of the symplectic form. In order to do this, we first push forward the tangent vectors of $\overline{\M}^I_\gamma$ from (\ref{MredgammaI}) into $\M_\gamma$. Note, that any curve $\Gamma(t)=(p_l(t),c_l(t))\in \overline{\M}^I_\gamma$ is embedded into $\M_\gamma$ as
\begin{align}
(\Gamma)_{\M_\gamma}(t)= (P(e)[\Gamma(t)], h(e)[\Gamma(t)]) =(\mu_o^2 p_l(t) \tau_l , e^{\mu_o c_l(t) \tau_l})
\end{align}
The vector fields $X$ to $\Gamma$ are in $^1C^\infty(\overline{M}^I_\gamma)$ and of the form $X= \sum_l (\delta p_l) \partial_{p_l} + (\delta c_l) \partial_{c_l}$ and therefore, we can compute their push forward into $\M_\gamma$:
\begin{align}
X_{\M_\gamma}&= \sum_{e} (\delta h_I(e)) R^I(e) + (\delta P_I(e)) \frac{\partial}{\partial P^I(e)}
\end{align}
where we choose the right-invariant vector fields $R^I$ as basis for the derivatives on $SU(2)$. To determine $\delta h_I(e)$ for $e=(v,l)$ we contract the last line with the dual to $R^I$, that is the right-invariant forms $\Omega_J= dh\cdot h^{-1}$ on $SU(2)$ with $\Omega_J(R^I)=\delta^I_J$:
\begin{align}
\delta h_I(e) &= \Omega_I(e) (\delta h_J(e)R^J(e))=-2 tr[\tau_I (dh(e) h(e)^{-1}) X_{\M_\gamma}]=-2tr[\tau_I X_{\M_\gamma}(h(e)) h(e)^{-1} ]=\nonumber\\
&= -2tr[
\partial_{t} h(\Gamma(t)) h(\Gamma(t))^{-1} |_{t=0}]=-2 tr [\tau_I \mu_o \dot{c}_l(t) \tau_l e^{\mu_o c_l(t)\tau_l} e^{-\mu_o c_l(t)\tau_l} |_{t=0}=\nonumber\\
&=\mu_o\,\delta c_l\,\delta^{I}_l 
\end{align}
using that $X(h)=\partial_t h(\Gamma(t))|_{t=0}$.
Also 
\begin{align}
\delta P^I(e) = \partial_t |_{t=0} P^I(e)[\Gamma(t)]  = \mu_o^2 \, \delta p_l\, \delta^I_l
\end{align}
We had already introduced in (\ref{pre-symplectic_potential_graph}) the pre-symplectic 1-form:
\begin{equation}
\xi=\frac{2}{\kappa\beta}  \sum_{e} P^{I}(e)\Omega_I(e)
\end{equation}
whose pull-back to $\overline{\M}_\gamma$ we can use to define a symplectic form thereon:
\begin{equation}
\omega|_{\overline{\M}_{\graph}^{I}}=d|_{\overline{\M}_{\graph}^{ I}}\xi|_{\overline{\M}_{\graph}^{ I}}
\end{equation}
It remains thus, to determine $\xi|_{\overline{\M}^I_\gamma}$ as a 1-form on $\overline{\M}^I_\gamma$ in such a way that 
\begin{align}\label{reduction_requirement_discrete}
\xi|_{\overline{\M}^I_\gamma} (X) := \xi(X_{\M_\gamma})
\end{align}
First, we notice that
\begin{align}
\Omega_I(e)(X_{\M_\gamma})=\delta^I_l \mu_o\, (\delta c_l) \hspace{10pt}\Rightarrow \hspace{10pt} \xi(X_{\M_\gamma})=\mu_o^3 N_1 N_2 N_3 \frac{2}{\kappa\beta } \sum_l p_l \, (\delta c_l)
\end{align}
Thus, we see that choosing the 1-form
\begin{align}
\xi|_{\overline{\M}^I_\gamma} =  \frac{2}{\kappa\beta}\sum_l p_l \, dc_l
\end{align}
is sufficient for (\ref{reduction_requirement_discrete}) to hold. It follows:

\begin{equation}
\omega|_{\overline{\M}_{\gamma}^{ I}}=d|_{\overline{\M}_{\gamma}^{ I}}\xi|_{\overline{\M}_{\gamma}^{ I}}  = \frac{2}{\kappa\beta}\sum_l dp_l\wedge dc_l
\end{equation}
And finally we obtain exactly the same bracket, as in the continuum:\footnote{Interestingly, this rigorous derivation of the reduced symplectic structure highlights that no trace of choosing holonomy-dependent gauge-covariant fluxes remains for Bianchi I. This is contrast to \cite{LS:19} where, after using gauge-covariant fluxes as well, a different symplectic structure was simply postulated on the cosmological phase space of FLRW degrees of freedom: Upon relabelling $p\mapsto \tilde{p}\, {\rm sinc}(\mu \tilde{c}) ,\; c\mapsto \tilde{c}$, in \cite{LS:19} the phase space was ad-hoc equipped with $\{\tilde{p},\tilde{c}\}=\kappa\beta/6$ and therefore does {\it not} mirror the reduced setting correctly.}
\begin{align}
\{p_l,c_l\}_{\overline{\M}_\graph^I}=\frac{\kappa\beta}{2}
\end{align}

\subsubsection{Constraints}
\label{s523}

What is the fate of the constraints?  As in the previous section, we can invoke lemma \ref{lemma5} to determine the spaces possible of remaining constraint which are not yet trivially satisfied on $\overline{\M}_{\gamma}^{ I}$. The calculation is analogous to the torus case before

\begin{cor}
The invariant spaces under the action of $\Phi^\gamma_I$ are 
\begin{equation}
C^\infty(\gamma_v, su(2))_{\Phi_I^\gamma}=\{0\},\quad C^\infty(\gamma_v)_{\Phi_I^\gamma}=\{\operatorname{const}\},\quad \prescript{1}{}{C}^\infty(\gamma_v)_{\Phi_I^\gamma}
=\{0\}
\end{equation}
\end{cor}

We are left with one scalar constraint that we can choose as $C(v_0)$ where $v_0$ is any vertex due to translational invariance and choosing a homogeneous lapse function $\mathcal{N}(v)=\mathcal{N}$. We emphasise that in our setting $\mathcal{N}$ is constant. On the reduced phase space, we obtain the corresponding generator simply via theorem \ref{theorem:2} of symmetry restriction of dynamics as $C(v_0)|_{\overline{\M}_{\gamma}^{ I}}$. We have to pay special attention to the fact that the definition of $C(v_0)$ in (\ref{DiscreteScalarConstraintThiemann}) involves Poisson brackets which are to be taken on the space $\M_\gamma$. However, one can see that lemma \ref{lemma:1a} \& \ref{lemma:3} apply, i.e. Poisson brackets of invariant functions with non-invariant ones can be reduced to the Poisson brackets on $\overline{\M}_{\gamma}^{ I}$:\\

\begin{lm}
The functions $C_E[\mathcal{N}]:= \sum_v C_E(v)\mathcal{N}(v)$, $C_L[\mathcal{N}]$, $V^\epsilon[\sigma]$ and $K$ from (\ref{DiscreteScalarConstraintThiemann}) are invariant under action of $\Phi^\gamma_I=\Phi^\gamma_{AB,E_o}(\Psi^\gamma_I)$ from (\ref{group_graphI}).
\end{lm}

\begin{proof}
The translational invariance of $(\mathbb{Z}_{N_1}\times \mathbb{Z}_{N_2}\times \mathbb{Z}_{N_3})$ is incorporated due to the sums over all vertices $v\in\gamma_v$. Thus, it suffices to restrict our attention to one vertex $v$ and the rotations in $H$ around this vertex.\\
Let us notice that for $R_{e_K}(\pi)$ the associated $\psi_e$ maps as follows: $\psi_e : (v,l) \mapsto (v, (-1)^{\delta^l_k+1} l)$ and
\begin{align}\label{invariant_proof}
\G^\gamma(R_{e_k}(\pi) ) P^I((v,l)) = (-1)^{\delta^I_k+1}P^I(\psi_e(v,l)),\quad
\G^\gamma (R_{e_k}(\pi) ) h(e) = \tilde{R} h(\psi_e(e)) \tilde{R}_k^{-1}
\end{align}
Thus, we see that with $\epsilon(i,j,k)=sgn(\det(\dot{e}_i,\dot{e}_j,\dot{e}_k)$  the contribution to the volume $\epsilon(i,j,k)\times $ $\times P^I(e_i)P^J(e_j)P^K(e_k)\epsilon_{IJK}$ remains invariant for all of $H$.\\
Also, every closed loop of holonomies $Tr(h(\Box))$, will cancel the resulting $\tilde{R}_k$ terms from (\ref{invariant_proof}) and thus $\G^\gamma Tr(h(\Box)) =Tr(h(\psi_e \Box))$. As in $C_E(v)$ we sum over all permutation of loops at $v$ it is easy to see that it remains invariant.\\
The same argumentation can be extended to $C_L$ and $K$ as well.
\end{proof}

As $V[\sigma]$ and $C_E$ are $\Phi^\gamma_I$-invariant functions on $\M_\gamma$, we get
\begin{align}
&C_E[\mathcal{N}]|_{\overline{\M}_{\graph}^{ I}} =\\
&\sum_v  \frac{-\mathcal{N}(v)}{2\kappa^2\beta}\sum_{ijk}\epsilon(i,j,k){\rm Tr}\bigg(
 \left[h(\square_{ij})|_{\overline{\M}_{\graph}^{ I}} - h^\dagger(\square_{ij})|_{\overline{\M}_{\graph}^{I}}\right]  h(k,v)|_{\overline{\M}_{\graph}^{I}} \{h^\dagger(k,v)|_{\overline{\M}_{\graph}^{I}},V^\epsilon[\sigma]|_{\overline{\M}_{\graph}^{I}}\}_{\overline{\M}_{\graph}^{I}}\bigg)\nonumber\\
&=N_1N_2N_3 \mathcal{N} \frac{2}{\kappa}\mu_o\left(\sqrt{\frac{p_2p_3}{p_1}}\sin(c_2\mu_o)\sin(c_3\mu_o)+ {\rm cyclic}\right)
\end{align}
because $V^\epsilon[\sigma]$ is invariant ($V^\epsilon[\sigma]|_{\overline{\M}_{\graph }^{I}}=\sqrt{p_1p_2p_3}$) and the result is independent of $v$.
But we know that $C_E$ is invariant under $\Phi^\gamma_I$ thus
\begin{align}
K|_{\overline{\M}_{\graph}^{I}}&=\{C^E|_{\overline{\M}_\graph^I}[\sigma],V^\epsilon[\sigma]|_{\overline{\M}_{\graph}^{I}}\}_{\overline{\M}_{\graph}^{I}}=\\
&=
-N_1N_2N_3 \frac{2}{\kappa}\mu_o^2\frac{\kappa\beta}{4}\left( \cos(c_1\mu_o) \left[
p_2 \sin (c_2 \mu_o))+p_3 \sin (c_3 \mu_o)
\right]+{\rm cylic}\right)
\nonumber
\end{align}
Moreover, function $K$ is also invariant. Hence all these Poisson brackets can be easily evaluated (using $\mu_o^3=N_1N_2N_3$)
\begin{align}
C_L[\mathcal{N}]|_{\overline{\M}_\graph^I} &=\sum_v \mathcal{N}(v)\; \frac{8(1+\beta^2)}{\kappa^4\beta^7}\sum_{ijk}\epsilon(i,j,k){\rm Tr}\bigg(
 h(i,v)|_{\overline{\M}_\graph^I} \{h^\dagger(i,v)|_{\overline{\M}^\graph_I},K|_{\overline{\M}^I_\graph}\}_{\overline{\M}^I_\graph} \; \times
\\
&\hspace{60pt}\times\;  h(j,v)|_{\overline{\M}^I_\graph} \{h^\dagger(j,v)|_{\overline{\M}^I_\graph},K|_{\overline{\M}^I_\graph}\}_{\overline{\M}^I_\graph} \,
 h(k,v)|_{\overline{\M}^I_\graph} \{h^\dagger(k,v)|_{\overline{\M}^I_\graph},V^\epsilon[\sigma]|_{\overline{\M}^I_\graph}\}_{\overline{\M}^I_\graph}
\bigg)\notag\\
&=-\frac{(1+\beta^2)}{\kappa \beta^2 }\frac{\mathcal{N}}{2\mu_o^2} \left[
\sqrt{\frac{p_2p_3}{p_1}}\sin(c_2\mu_o)\sin(c_3\mu_o)A(c_3,c_1)A(c_1,c_2)+ {\rm cyclic}
\right]\nonumber
\end{align}
where $A(c_I,c_J) := \cos(c_I \mu_o) + \cos(c_J \mu_o)$.
We end up with the following symmetry restricted scalar constraint containing all information of the dynamics for the discrete Bianchi I system:
\begin{equation}
C^\epsilon[\mathcal{N}]|_{\overline{\M}_\gamma^I}=
\frac{2\mathcal{N}}{\kappa \mu_o^2 \beta^2}\sqrt{\frac{p_2p_3}{p_1}}\sin(c_2 \mu_o)\sin(c_3 \mu_o)\left[\beta^2-(1+\beta^2)\frac{A(c_3,c_1)A(c_1,c_2)}{4}\right] + {\rm cyclic}
\end{equation}
A similar version to this regularisation had already been under active investigation in the earlier papers \cite{GM19a,GM19b}, albeit using instead of $\mu_o$ a slight modification, the $\bar{\mu}$-scheme which can not be derived via our restriction method.\\
In our setting, for any observables we are interested in all Poisson brackets can be computed on the symmetry restricted level.

\subsection{Symmetry restricted, discrete FLRW model}
\label{s53_FRW}

Finally, we will specialise to the case which closely resembles a discretisation of a spatially isotropic system, therefore reminiscent to the models of Friedmann-Lema\^{i}tre-Robertson-Walker as we discussed in section \ref{s422}. Once again, the system will be described by a single degree of freedom, related to the scale factor.\\
For this purpose, we assume that $N_1=N_2=N_3=N$ and in addition to $H$ group we consider also rotation by $\pi/2$ around the main axes:
\begin{align}
\Psi^\graph = \Psi^\gamma_I \times H', \quad H'=\{R_{e_k}(n\pi/2)\; \colon \; k=1,2,3,\; n=0,1,2,3\}
\end{align}
Therefore, we can directly use our results from the previous section of discretised Bianchi I and perform a {\it second} symmetry restriction on this phase space with respect to permutation of axes:

\begin{fact}
The points of $\overline{\M}_\gamma^I\subset \M_{\gamma}$ that are invariant under the group $\Phi^\graph_{AB,E_o}(\Psi^\graph)$ are
\begin{equation}
(P(e),h(e))=(\mu_o^2 p\tau_l,e^{i \mu_o c\tau_l})
\end{equation}
and we denote their space by $\overline{\M}_\gamma^{FRW}$.
\end{fact}

\begin{proof}
We know that the points on $\overline{\M}^\gamma_I$ are parametrised by
\begin{equation}
(P(v,l),h(v,l))=(\mu_o^2 p_l\tau_l,e^{i \mu_o c_l\sigma_l})
\end{equation}
and implementing invariance with respect to rotations by $\pi/2$ shows that
\begin{align*}
p_l=p,\quad c_l=c
\end{align*}
\end{proof}
The symplectic form is a further reduction of symplectic form $\omega|_{\overline{\M}^I_\gamma}$ from the discretised Bianchi I case, i.e.
\begin{equation}
\frac{2}{\kappa\beta}\sum_l dp_l\wedge dc_l = \frac{2}{\kappa\beta}\;3\;dp\wedge dc
\end{equation}
thus finally we obtain exactly the same bracket, as in the continuum:
\begin{align}
\{p,c\}_{\overline{\M}^{FRW}_\graph}=\frac{\kappa\beta}{6}
\end{align}

In the previous section, we saw that we were left with a only single constraint from the Bianchi I case which remained non-trivial. Understood as maps to the dual of the Lie algebra in the sense of lemma \ref{lemma5}, we are here left with one scalar constraint that we can choose as $C(v_0)$ where $v_0$ is some arbitrary vertex and hence $C_E[\mathcal{N}]= C(v_0) \sum_{v\in\gamma_v}\mathcal{N}(v) N$.\\

Alternatively, invoking compactness of the symmetry group and (\ref{remaining_constraints_gphi}), we can also compute the constraint directly by restriction from Bianchi I case as we can see that all constituents are still $\Phi^\gamma_{AB,E_o}(\Psi^\gamma)$-invariant functions:
\begin{align}
C_E[\mathcal{N}] |_{\overline{\M}^{FRW}_\gamma} &= C_E[\mathcal{N}]|_{\overline{\M}^I_\gamma} |_{\overline{\M}^{FRW}_\gamma}=
\frac{6\mathcal{N}}{\kappa}\sqrt{p}\frac{\sin(c \mu_o)^2}{\mu_o^2 }\\
C_L[\mathcal{N}] |_{\overline{\M}^{FRW}_\gamma} &= C_L[\mathcal{N}]|_{\overline{\M}^I_\gamma} |_{\overline{\M}^{FRW}_\gamma}=-\frac{1+\beta^2}{\beta^2}\frac{6\mathcal{N}}{\kappa}\sqrt{p}\frac{\sin(2c \mu_o)^2 }{4\mu_o^2}
\end{align}

Plugging all together we obtain the final form of the restricted constraint, reinforcing the result of \cite{DMY:08,DL:17a,DL:17b}:
\begin{align}\label{final_effective_Hamiltonian}
C^\epsilon [\mathcal{N}] |_{\overline{\M}^{FRW}_\gamma}=\frac{6\mathcal{N}}{\kappa}\sqrt{p}\left[
\frac{\sin(c \mu_o)^2}{\mu_o^2 }-\frac{1+\beta^2}{\beta^2}\frac{\sin(2c \mu_o)^2 }{4\mu_o^2}
\right]
\end{align}
We also see that all Poisson brackets can be computed already on the reduced level.\\

\begin{remark}
The discrete set-up bears special interest for semiclassical investigations of a theory for quantum gravity on the lattice. If one can show that some discretised quantum field theory features coherent states which follow a semiclassical trajectory, then this trajectory is exactly the one mapped by the reduced Hamiltonian on the reduced phase space given by a constraint
\begin{equation}
C_{v_0}|_{\overline{\M}_{FRW}^\graph}=0
\end{equation}
\end{remark}

\begin{remark}
This conclusion was obtained already in \cite{HL:19} although by direct calculations. We showed that it follows from general arguments and can be extended also to Bianchi I.
\end{remark}

\section{Conclusion}
\label{section:6}

In this manuscript, we have put forward the tools for {\it symmetry restriction}. It can prove useful when interested in those solutions to classical systems, which obey certain symmetries, i.e. points in phase space which are invariant with respect to the action of some symmetry group. Particular examples for its applications have been presented in the second half of this article for the case of general relativity. We will now reiterate the main statements of the paper and summarise further noteworthy results obtained along the way. Finally, we will close with an outlook to quantum theories.

\subsection{Summary: Symmetry restriction for classical theories}

We have investigated the properties of symplectic manifolds $(\mathcal{M}, \omega)$ on which a symmetry group $\Phi$ acts via symplectomorphisms. Those points of $\M$ which remain invariant under the action of $\Phi$ form the $\Phi$-invariant set $\overline{\M}$. We were especially interested in cases, where the symplectic form $\omega$ can be restricted to a non-degenerate symplectic form $\omega|_{\overline{ \M}}$ on a submanifold $\overline{\M}$. The main result of this paper is encapsulated in theorem \ref{theorem:2}, the {\it symmetry restriction of dynamics}: consider a $\Phi$-invariant Hamiltonian $H$ whose flow generates dynamics $H|_{\overline{\mathcal M}}$ and its restriction to $\overline{\M}$. Under the above conditions, the evolution of any phase space point in $\overline{\M}$ due to $H$ and $\omega$ agrees with the evolution due to  $H|_{\overline{ \mathcal M}}$ and $\omega|_{\overline{ \M}}$. Therefore, allowing to restrict our attention to the simplified setting of the symmetry restricted manifold.\\
That this indeed eases computations has been demonstrated for general relativity: we revisited two incarnations of its Hamiltonian formulation and recalculated several cosmological scenarios employing the framework of symmetry reduction. An important distinction had to be made between the compact and non-compact cases, as the latter ones typically do not allow a reduction of the symplectic form due to appearing infinities.\\
A further application of symmetry restriction has been found in the case of discretised gravity on a graph: we truncated the continuum phase space of the Ashtekar-Barbero formulation of GR to  a subset which is only non-trivial along the edges of a finite graph (and its dual cell complex). After having introduced Thiemanns regularisation of the scalar constraints, we could show that symmetry restriction applies and thus the full dynamics of discrete gravity for those phase space points describing discrete homogeneity can be encapsulated on a reduced submanifold which is parametrised by the scale factors and their momenta. This confirmed the classical part of the conjecture of (\ref{final_effective_Hamiltonian}) being an effective Hamiltonian put forward in \cite{DL:17a}.\\
 
Over the course of this paper, many results have been obtained which we deem noteworthy and thus highlight them separately:
\begin{itemize}
\item The Poisson bracket between a $\phi$-invariant function and {\it any} other function can be reduced to the Poisson bracket of the restricted functions on the $\phi$-invariant submanifold (which basically gives the theorem "symmetry restriction of dynamics").
\item If $\Phi\subset G$, then symmetry restriction can be understood as a complementary procedure to symplectic reduction: points in phase space employing symmetries  typically do not allow symplectic reduction (as $G$ needs to act free). However, for these points we can first perform symmetry restriction and then symplectic reduction of the remaining group without loss of information.
\item The constraint algebra (even for nongroup constraints such as GR) is typically kept faithfully under symmetry restriction.%lemma 6
\item For the Hamiltonian formulation of GR, we have strengthened the understanding of the gauge which is not only the Bergman-Komar group (the later one agrees with the component of identity of $\Diff(\sigma)$ on shell). We have shown that indeed -- as expected -- the gauge transformations extend also to the disconnected part of $\Diff(\sigma)$. %(proof of 3.1.2)
\item In the Ashtekar-Barbero formulation, there are actually many ways to formulate the  vector constraint (as the generator of spatial diffeomorphisms): by adding arbitrary (possibly phase-space dependent) smearing of the Gauss constraint $G$. Of course, this fact is already well-known (e.g. the standard version of $D'$ in the literature and the actual generator $D$ differ by $G[\Gamma]$), but we highlight this fact as a specific choice of smearing allows to lift diffeomorphisms from $\M_{ADM}$ to $\M_{AB}$ such that symmetry restriction of $\Phi$ on both phase spaces commutes with symplectic reduction $\M_{AB}\mapsto \M_{ADM}$.
\item We recalled the well-known fact that symmetry restriction on non-compact Bianchi I spaces (and FLRW) is not possible due to no-well defined symplectic structure. For compact case it is possible. The dependence on the fiducial volume is related to a particular choice of coordinate system on the restricted phase space. The symmetry under rescaling is no longer a gauge transformation, as it was for the noncompact Bianchi I.
\item When discretisting the Ashtekar-Barbero formulation, a subtlety arises concerning the Gauss transformations. Classically, the AB variables are typically understood as a $SO(3)$ gauge theory, however one may also extend it to all $O(3)$ transformations by adding the points of $\det E <0$ to the phase space. Now, in LQG one typically lifts $SO(3)$ to $SU(2)$.
However, on nontrivial manifolds not all $SO(3)$ gauge transformations allow a lift to $SU(2)$.
On finite graphs, we showed how one can lift these additional $SO(3)$ transformations together with $O(3)$ gauge transformations and we showed  hat they encode information of the topology of the discretised manifold.
\item Along the lines of \cite{Thi:00(QSDVII)} we could equip the discrete phase space with a Poisson-bracket  from continuum and reduce it to the symmetry-invariant submanifold of Bianchi I (and FLRW I). The Poisson bracket derived in this way differs from the one postulated in \cite{LS:19}.
\item Of special importance is the choice of the cell complex for the case of gravity on a graph. In order to apply symmetry restriction, the graph itself needs to remain invariant under the action of the subgroup of $\Diff(\sigma)$ of interest. In general it is a highly non-trivial task to find for a given (discrete) symmetry a subgraph obeying this property. Thus, it is for example non-trivial to employ symmetry restriction to discrete systems mirroring spherical symmetries, such as black holes or $k=+1$ cosmology \cite{ADL:19,LW:19}.
\item Symmetry restriction of dynamics is possible for any discretisation of the constraints which obeys a version of covariance. Thus not only the Thiemann regularisation (employed in this paper) but also the one put forward in \cite{ALM:15} could be used.
\end{itemize}

\subsubsection{Outlook: Semiclassical phase space and the Hamiltonian}

In the papers \cite{DMY:08,DL:17a,DL:17b} a new effective Hamiltonian for Loop Quantum Cosmology has been proposed. The idea is that the evolution of a semiclassical state (peaked sharply on FLRW cosmology) is governed by said effective Hamiltonian. The latter one is obtained as the expectation value of the LQG Hamiltonian operator in the coherent states peaked at the semi-classical point of the phase space. Moreover, the coherent states are labelled by an additional parameter $\lambda$ and the smaller it is the more classical the system will behave. \\

In fact, this classicality behaviour can be only achieved in the relative sense: for it, the system needs to be rescaled. If we divide the system into position and momenta one of the option (so called geometric optics approximation) is to rescale momenta. 
\begin{equation}
(x,p)\rightarrow (x,\lambda^{-1}p)
\end{equation}
Usually this means that the system evolve with different speed. We thus need to rescale also the time
\begin{equation}
t\rightarrow \lambda^\beta t
\end{equation}
where $\beta$ might be any real number (dependent on the explicit form of the Hamiltonian). Under certain conditions, it turns out that the coherent states evolve according to the classical dynamics of the effective Hamiltonian, i.e. their expectation values in properly rescaled quantisation of suitably nice observables $\hat{A}_\lambda$ follows the semi-classical trajectory.

The similar property is the following. The expectation value is approximated by the rules for quantisation of suitable operators (that allow semiclassical expansion) \cite{GS:94}
\begin{enumerate}
\item $\langle \hat{A}_\lambda\hat{B}_\lambda\rangle=AB+O(\lambda)$ up to higher order in semiclassical parameter
\item In the case of the commutator the formula above give zero and we can compute the higher order
\begin{equation}
\langle [\hat{A}_\lambda,\hat{B}_\lambda]\rangle=-i\lambda\{A_\lambda,B_\lambda\}+O(\lambda^2)
\end{equation}
\item $\langle\sqrt{\hat{A}_\lambda}\rangle=\sqrt{A}_\lambda+O(\lambda)$
\end{enumerate}
Of course the above formulas are not universally true and need to be proven in a given situation. 

These results allows us to conjecture that the evolution of algebraic LQG is governed by certain classical mechanical system on the phase space that under standard quantisation procedure produces the Ashtekar-Lewandowski Hilbert space\cite{AL:94}. We will substantiate these steps in a future publication.

\vspace{1cm}
\noindent\textbf{Acknowledgements:} \\
The authors want to thank Andrea Dapor for collaboration in the early phases of the project and for many fruitful and helpful discussions on all aspects of this paper.\\
Further, we thank Benjamin Bahr, Muxin Han, Jerzy Lewandowski for useful comments.\\
W.K. acknowledges the support by Polish National Science Centre (NCN) Sheng-1 with Grant No. 2018/30/Q/ST2/00811.
K.L. acknowledges support by the German Research Foundation (DFG) under Germany\'s
Excellence Strategy - EXC 2121 ``Quantum Universe'' - 390833306. This work was partially funded by DFG-project BA 4966/1-2.

\newpage
%
%
%%Sets the bibliography style to UNSRT and imports the 
%%bibliography file "references.bib".
	\bibliographystyle{unsrt}
\bibliography{References}{}

\end{document}